\documentclass[11pt,oneside]{article}

\usepackage[T1]{fontenc}
\usepackage[margin=1in]{geometry}
\usepackage{mathpazo}
\usepackage[svgnames]{xcolor}

\newcommand{\red}[1]{\textcolor{red}{#1}}

\newif\ifshowtodos
\showtodostrue

\newif\ifshownotes
\shownotestrue

\newif\ifshowworking
\showworkingtrue

\usepackage{booktabs}
\usepackage{caption}
\usepackage{graphicx}

\let\oldincludegraphics\includegraphics%
\renewcommand{\includegraphics}[2][]{\IfFileExists{#2}{\oldincludegraphics[#1]{#2}}{\red{[FILE NOT FOUND]}}}

\usepackage{amsmath}
\usepackage{amssymb}
\usepackage{amsthm}

\DeclareMathOperator*{\argmax}{\arg\max}

\newcommand{\etahat}{\hat{\vphantom{\pi}\eta}}
\newcommand{\FB}{^\text{FB}}
\newcommand{\IC}{^\text{MIC}}

\newcommand{\parfrac}[2]{\frac{\partial #1}{\partial #2}}
\newcommand{\R}{\mathbb{R}}
\newcommand{\SB}{^\text{SB}}
\newcommand{\shat}{\hat{\vphantom{\pi}s}}
\newcommand{\Tcal}{\mathcal{T}}
\newcommand{\ubar}{\overline{u}}
\newcommand{\ybar}{\overline{y}}

\let\oldleft\left
\let\oldright\right
\renewcommand{\left}{\mathopen{}\mathclose\bgroup\oldleft}
\renewcommand{\right}{\aftergroup\egroup\oldright}

\newtheorem{lemma}{Lemma}
\newtheorem{proposition}{Proposition}
\newtheorem{theorem}{Theorem}
\theoremstyle{definition}

\newif\ifbodyproofs
\bodyproofstrue

\let\oldparagraph=\paragraph
\renewcommand{\paragraph}[1]{\oldparagraph{#1.}}

\usepackage{needspace}
\AtBeginEnvironment{lemma}{\Needspace*{6\baselineskip}}
\AtBeginEnvironment{proposition}{\Needspace*{4\baselineskip}}
\AtBeginEnvironment{theorem}{\Needspace*{4\baselineskip}}

\usepackage{natbib}
\newcommand\citeapos[1]{\citeauthor{#1}'s (\citeyear{#1})}

\usepackage[colorlinks, linkcolor=DarkRed, citecolor=DarkBlue]{hyperref}
\usepackage[capitalize, nameinlink]{cleveref}

\crefname{appsec}{Appendix Section}{Appendix Sections}
\crefname{claim}{Claim}{Claims}
\crefname{figure}{Figure}{Figures}

\bodyproofsfalse

\title{
Course design in the age of AI
}
\author{%
Benjamin Davies\thanks{
Department of Economics, Stanford University; bldavies@stanford.edu.
I thank seminar participants at Stanford and UC Irvine for helpful comments.
}
}
\date{Draft version: \today}

\begin{document}

\maketitle

\begin{abstract}
    \noindent
    I develop a model of learning-by-doing and course design, and use it to study the impacts of artificial intelligence (AI).
    A myopic student faces a sequence of tasks that he can work on or delegate to AI.
    Work requires costly effort but builds skill; delegation requires no effort but builds no skill.
    A teacher designs the task sequence (``course'') to maximize the student's skill development, given his choices to work or delegate.
    Without AI, the teacher makes earlier tasks more effort-intensive and later tasks more skill-intensive.
    With AI, the teacher must redesign early tasks to induce effort, leading to less skill development.
    If AI complements effort, then improvements in AI quality make high-skill students learn faster but low-skill students learn slower.
    
    \vskip\baselineskip
\end{abstract}

\clearpage
\section{Introduction}

Using artificial intelligence (AI) can boost productivity \citep{Brynjolfsson-etal-2025-QJE,Dillon-etal-2026-AERInsights,Noy-Zhang-2023-Science}, but can also inhibit learning and cognition \citep{Bastani-etal-2025-PNAS,Kosmyna-etal-2025-,Liu-etal-2026-,Shaw-Nave-2026-,Shen-Tamkin-2026-,Stromberg-etal-2026-}.%
\footnote{
Not all evidence points to negative effects of AI on learning.
For example, \cite{Huntington-Klein-2026-} estimates null effects of AI access on high school students' test scores, while \cite{Chung-etal-2026-} find that personalized AI tutoring raises students' scores on exams completed without AI assistance.
}
In particular, using AI can ``short-circuit the productive struggle essential for learning'' \citep[p.~2]{Poulidis-etal-2025-}.

At the same time, AI use in schools is widespread \citep{Campos-Singleton-2026-,Doss-etal-2025-}.
Students use AI to augment their learning and automate their coursework, and view AI as beneficial \citep{Contractor-Reyes-2026-}.
However, AI allows students to delegate coursework that was intended as an opportunity to develop skill.
\cite{Chirikov-2026-} offers empirical evidence of such delegation, while \cite{Bastani-etal-2025-PNAS} and \cite{Stromberg-etal-2026-} offer evidence that students who delegate develop less skill.%
\footnote{
AI (mis)use also threatens to invalidate educational assessments \citep{Chirikov-etal-2026-Science}:
if students delegate tasks to AI, then performance on those tasks cannot be used to infer students' ability.
}

These findings identify a problem, but not its solution.
One response is to ban students from using AI.
However, bans may not be enforceable, and they forgo the benefits of using AI to augment learning.
Banning AI would take coursework as given and change students' access to AI.
An alternative is to do the opposite: take AI access as given and change students' coursework.
This raises a question that prior work leaves open:
given that students can delegate to AI, how should teachers redesign courses to ensure students develop skill?

To answer this question,
I develop a model of learning-by-doing and course design.
I use the model to analyze the impacts of AI on optimal course design and skill development.
My model and results offer guidance on redesigning courses for the age of AI.

My main contribution is conceptual:
I provide a framework for thinking about educational responses to AI.
In particular, I frame courses as designable objects and delegation as a constraint on course design.
This framing distinguishes my paper from other theoretical papers on how AI impacts learning \citep{Adilov-Cline-2026-StudMicro,Fudenberg-Levine-2026-,Garicano-Rayo-2026-,Huang-Vishnoi-2026-,Peterson-2025-}.
In addition, my theoretical results align with empirical evidence and offer new, empirically testable predictions.

\paragraph{Model overview}

A myopic student (he) faces a sequence of tasks that he can work on or delegate to AI.
Working requires costly effort but builds skill.
Delegation does not require effort but does not build skill.
A teacher (she) designs a task sequence (i.e., a ``course'') that maximizes the student's overall skill development given his best-responses (i.e., his consequent choices to work or delegate).

When the student works on a task, he combines effort and skill to produce an output.
I interpret ``effort'' as the cognitive resources spent applying a production method and ``skill'' as the ability to identify suitable methods.
Effort is a costly productive input but skill is costless.
Moreover, effort and skill are complements because a student with more skill can allocate effort more efficiently.
The student builds more skill when he exerts more effort, so the teacher designs her course to induce as much effort as possible.

I characterize tasks by their ``effort intensity'': the extent to which production relies on effort vis-\`a-vis skill.
Effort-intensive tasks reward applying given methods (e.g., ``use the power rule to differentiate these polynomial functions''), while skill-intensive tasks reward identifying appropriate methods (e.g., ``build a tractable economic model'').
The effort intensity of a task determines how much the student produces by exerting effort, which determines whether he prefers to work or delegate.

The central conflict in my model is between the student's and teacher's incentives.
The student values output, which he can produce by working or delegating.
The student does not value skill development because he is myopic.
In contrast, the teacher values skill development only.
If the student works, then he produces output and gains skill.
This is good for him and good for the teacher.
In contrast, if the student delegates, then he produces output but does not gain skill.
This is (myopically) good for him but bad for the teacher.
So she wants to prevent the student from delegating, using the only design instruments at her disposal: tasks' effort intensities.
However, her choices of effort intensities are constrained by the student's agency: given a task, he is free to choose whether to work or delegate.

\paragraph{Main results}

First, I consider the ``first-best'' case in which the student always prefers to work because he always produces more than AI.
Then the teacher optimally designs tasks that become gradually less effort-intensive and more skill-intensive (\cref{thm:fb-solution}).
This is because the teacher maximizes the student's skill development by maximizing his effort.
On earlier tasks, the student has less skill and so the teacher induces effort by making production rely on effort.
On later tasks, the student has more skill, which makes effort more productive because effort and skill are complementary inputs.
The teacher capitalizes on this complementarity by making tasks less effort-intensive and more skill-intensive.

Second, I consider the ``second-best'' case in which the student may prefer to delegate because AI may produce more than him.
Then the first-best course design is not incentive-compatible and the teacher must distort it to induce effort.
She does so by making earlier tasks more skill-intensive relative to the first-best.
Her second-best design comprises tasks with increasing-then-decreasing effort intensities (\cref{thm:sb-solution}).
This ensures the student develops some skill.
However, he develops less than under the first-best design.
When the teacher designs a second-best course, the student gains more skill overall when he starts with more skill and when AI is less productive (\cref{thm:sb-overall-gain}).
Intuitively, skillful students can produce more on their own, are less tempted to delegate, and so can be induced to work with less distortion.
In contrast, making AI more productive tightens the incentive constraints on the teacher's course design, and so she must distort it further from the first-best to induce effort.

Finally, I extend my model to allow for the possibility that AI complements effort.%
\footnote{
For example, the student could use AI to perform effort-intensive subtasks he has used skill to identify, rather than delegate entire tasks and let the AI identify subtasks on its own.
}
I capture this possibility by introducing an ``AI quality'' parameter that increases the marginal product of effort \emph{and} the payoff from full delegation.
Improvements in AI quality make the student gain skill faster when he has high skill but slower when he has low skill (\cref{thm:complement-sb-effort}).
Intuitively, when the student has high skill, he uses AI to enhance his effort and so quality improvements make him more willing to work.
In contrast, when the student has low skill, he uses AI to avoid effort and so quality improvements make him more tempted to delegate.
The different uses of AI lead to different learning outcomes, consistent with empirical evidence \citep{Contractor-Reyes-2026-a,McNichols-etal-2026-Proc.LAK2616thInt.Learn.Anal.Knowl.Conf.,Pallant-etal-2026-StudHighEd,Shen-Tamkin-2026-,Zamarripa-2026-}.

\paragraph{Related literature}

This paper is closest to recent theoretical work on AI and learning.
\cite{Fudenberg-Levine-2026-} show that early access to AI can weaken the incentive to work on tasks that build and maintain expertise.
\citeauthor{Fudenberg-Levine-2026-}'s agent chooses between exogenous tasks; in contrast, I treat tasks as endogenous and study how they should be designed.
\cite{Garicano-Rayo-2026-} study AI's impact on apprenticeships, where ``juniors pay for training by doing menial work'' that AI performs well.
In \citeauthor{Garicano-Rayo-2026-}'s model, the planner is a master who controls the pace of knowledge transfer, and the agent's outside option is to quit and use AI independently.
In contrast, in my model, the planner is a teacher who designs the task sequence and the agent's outside option is to delegate within a task.
\cite{Huang-Vishnoi-2026-} study the conditions under which repeated delegation makes skill collapse to a low equilibrium; I allow a designer to prevent such collapses.
\cite{Peterson-2025-} considers an educational planner who misreads the labor-market return to easily taught skills.
\citeauthor{Peterson-2025-}'s friction is informational, whereas mine is incentive-based: the teacher in my model must prevent delegation.

My focus on incentives echoes the theoretical literature on educational standards.
\cite{Becker-Rosen-1992-EER}, \cite{Costrell-1994-AER}, and \cite{Betts-1998-AER} view education as a principal-agent problem: teachers set standards or grading policies to motivate effort from students.
\cite{Dubey-Geanakoplos-2010-GEB} show that coarse grades can motivate more effort than exact scores when students care about their rank.
\cite{Adilov-Cline-2026-StudMicro} study how the possibility of AI-enabled cheating impacts optimal grading policies, showing that it precludes teachers from using high grade boundaries to motivate effort.
\citeauthor{Adilov-Cline-2026-StudMicro} also show that effort falls as AI quality rises.
In contrast, I show effort can rise or fall, depending on whether the student uses AI as a complement or substitute for effort (see Theorem~\ref{thm:complement-sb-effort}).
My teacher also motivates effort using a different instrument: she designs tasks, rather than the rewards attached to them.%
\footnote{
That a principal can motivate an agent by designing tasks, rather than providing transfers, is a central idea in the literature on multitask incentive problems \citep{Holmstrom-Milgrom-1991-JLEO}.
I apply the idea to an educational setting.
}

A separate literature treats how and what students learn as an optimization problem.
\cite{Lazear-2001-QJE} studies optimal class sizes when students can disrupt their classmates' learning.
\cite{Ellison-Pathak-2025-} model course design as the allocation of teaching time across skills.
My model shares \citeauthor{Ellison-Pathak-2025-}'s premise that courses can be chosen optimally, but I focus on a different margin (effort intensity, rather than skill coverage), and introduce endogenous student effort and skill dynamics.
Such dynamics appear in \citeapos{Shaikh-2025-} model of course design.
\citeauthor{Shaikh-2025-} shows that if course content is sufficiently cumulative, then teachers should put more grade weight on earlier tasks so that students are motivated to build foundational skills.
My model complements \citeauthor{Shaikh-2025-}'s by endogenizing the task sequence and allowing students to delegate to AI.

Finally, this paper draws on the literatures on human capital formation and learning-by-doing.
\cite{Schultz-1960-JPE} frames education as investment in human capital;
\cite{Becker-1962-JPE}, \cite{Mincer-1962-JPE}, and \cite{Ben-Porath-1967-JPE} present formal theories of such investment.
\cite{Heckman-2000-ResEcon} and \cite{Cunha-Heckman-2007-AER} emphasize early investment and dynamic complementarities (``skill begets skill'').
Such complementarities arise in my model: exerting effort builds skill, which raises the future productivity of effort and leads to more skill development.
Similar complementarities arise in the economic literature on learning-by-doing, surveyed by \cite{Thompson-2010-HandbookoftheEconomicsofInnovation}.
Its central idea is that knowledge and skill accumulate via experience, which can take the form of investment \citep{Arrow-1962-REStud}, production \citep{Rosen-1972-QJE}, or technology use \citep{Jovanovic-Nyarko-1996-ECTA}.
Likewise, in my model, skill accumulates via effortful production.
However, I depart from the learning-by-doing literature by introducing ``doing-without-learning'' as an outside option.

\paragraph{Roadmap}

The paper is organized as follows.
\cref{sec:model} describes the model.
\cref{sec:br} characterizes the student's best-responses.
\cref{sec:solution} characterizes the teacher's optimally designed courses.
\cref{sec:complement} considers the possibility that AI complements effort.
\cref{sec:assumptions} discusses my modeling assumptions.
\cref{sec:conclusion} concludes.
\cref{app:proofs} contains proofs of my mathematical claims.

\section{Model}
\label{sec:model}

There is a myopic student (he) and a forward-looking teacher (she).
The student faces a sequence of tasks~$t\in\Tcal\equiv\{0,\ldots,T-1\}$ with~$T\ge2$.
He can work on a task or delegate it to an AI assistant.
Work requires costly effort but builds skill.
Delegation requires no effort but builds no skill.
The teacher designs tasks that maximize the student's skill development given his best-responses.

\paragraph{Working}

Working on a task involves combining effort and skill to produce an output.
I interpret ``effort'' as the cognitive resources spent applying a production method, and ``skill'' as the ability to identify suitable methods and transfer them across tasks.%
\footnote{
\label{fn:problem-domain}%
For example, a student could transfer methods across tasks by recognizing those tasks as instances of a common class.
This would align with \citeauthor{Anzai-Simon-1979-PsychRev}'s (\citeyear[p.~124]{Anzai-Simon-1979-PsychRev}) discussion of problem-solving strategies:
``Since the most efficient strategies for a particular task may not be the most obvious ones to someone encountering the task for the first time ...
Initially, the solver might hit on one of the `obvious' strategies, and then gradually progress to more efficient ones with increasing familiarity with the problem domain.''
}
Effort is a costly productive input but skill is costless.%
\footnote{
Intuitively, the student either knows how to approach a problem efficiently or does not, depending on his familiarity with the problem domain (see~\cref{fn:problem-domain}).
}
Moreover, effort and skill are complementary inputs because a student with higher skill can allocate effort more efficiently.

I characterize tasks by their ``effort intensity'': the extent to which production relies on effort vis-\`a-vis skill.
Effort-intensive tasks reward applying given methods, while skill-intensive tasks reward identifying appropriate methods.

Let~$s_t>0$ denote the student's skill level when he faces task~$t\in\Tcal$ and let~$\eta_t\in(0,1]$ denote the task's effort intensity.
If the student works on task~$t$ and exerts effort~$x_t\ge0$, then he produces output
\begin{equation}
    \label{eq:output}
    y(x_t,s_t,\eta_t)\equiv x_t^{\eta_t}s_t^{1-\eta_t},
\end{equation}
pays cost~$c(x_t)\equiv x_t^2/2$,
and receives payoff
\begin{equation}
    \label{eq:payoff}
    u(x_t,s_t,\eta_t)\equiv\pi\,y(x_t,s_t,\eta_t)-c(x_t),
\end{equation}
where~$\pi>0$ is a task-invariant ``motivation'' parameter that I explain below.

The Cobb-Douglas production function~\eqref{eq:output} captures the complementarity between effort and skill, and has a single argument---the effort intensity~$\eta_t$---that controls how much production relies on effort vis-\`a-vis skill.%
\footnote{
If~$\eta_t=0$, then output~$s_t$ is independent of effort~$x_t$ and the student optimally chooses~$x_t=0$.
But then he builds zero skill, which is never optimal for the teacher when she solves her problem~\eqref{eq:problem}.
Thus, it is without loss to bound~$\eta_t$ away from zero.
}
The cost function~$c(x_t)$ captures the cognitive cost of exerting effort.
I discuss my choice of production and cost functions in \cref{sec:cobb-douglas,sec:quadratic-costs}.

The payoff~\eqref{eq:payoff} captures a tension between wanting to produce and not wanting to exert effort.
The student places value~$\pi\,y(x_t,s_t,\eta_t)$ on the output he produces by working.
The parameter~$\pi$ indexes his motivation to produce relative to the cost of exerting effort.%
\footnote{
Specifically, the parameter~$\pi$ controls the per-unit trade-off between output and effort: the higher is~$\pi$, the more motivated is the student to produce a unit of output relative to the cost of exerting a unit of effort.
}
This motivation may be intrinsic (e.g., he derives satisfaction from producing) or extrinsic (e.g., his output determines his grade).
I assume~$\pi$ is constant across tasks and discuss this assumption in my conclusion.

\paragraph{Delegation}

If the student delegates task~$t$, then he exerts zero effort, pays zero cost, produces output~$\ybar\ge0$, and receives payoff~$\ubar\equiv\pi\ybar$.
This payoff is the value the student places on output generated by an AI assistant (e.g., via prompting an AI chatbot with ``here is my homework, give me the answers'').%
\footnote{
One can also interpret~$\ybar$ as an expected output.
For example, the student may know only the AI assistant's typical performance over a range of tasks, rather than its exact performance on a specific task.
Alternatively, the AI-generated output may contain random errors (e.g., due to ``hallucination'').
Interpreting~$\ybar$ as an expectation does not change my analysis or results.
}

The AI-generated output~$\ybar$ does not depend on the student's skill level~$s_t$ or the task's effort intensity~$\eta_t$.
I discuss this assumption in \cref{sec:constant-ybar}.

\paragraph{Best-responses}

On each task~$t\in\Tcal$, the student's best-response effort
\[ x^*(s_t,\eta_t)\in\argmax_{x\ge0}\,u(x,s_t,\eta_t) \]
yields payoff
\[ u^*(s_t,\eta_t)\equiv\max_{x\ge0}\,u(x,s_t,\eta_t). \]
He works on task~$t$ if~$u^*(s_t,\eta_t)\ge\pi\ybar$ and delegates task~$t$ otherwise.

\paragraph{Skill}

The student has initial skill~$s_0>0$.
On each task~$t\in\Tcal$, he gains skill~$\Delta(s_t,\eta_t)\equiv s_{t+1}-s_t$ proportional to the effort he exerts:
\begin{equation}
    \label{eq:accumulation}
    \Delta(s_t,\eta_t)\equiv \begin{cases}
        r\,x^*(s_t,\eta_t) & \text{if student works} \\
        0 & \text{if student delegates}
    \end{cases}
\end{equation}
with~$r>0$ a fixed skill accumulation rate.
I discuss my specification of~\eqref{eq:accumulation} in \cref{sec:accumulation-spec}.

\paragraph{Teacher's problem}

The teacher chooses effort intensities~$(\eta_t)_{t\in\Tcal}$ that maximize the student's overall skill gain
\[ s_T-s_0=\sum_{t\in\Tcal}\Delta(s_t,\eta_t) \]
given his best-responses.%
\footnote{
Thus, the teacher maximizes the course's ``value-added.''
This is consistent with the empirical literature that uses value-added to evaluate teachers \cite[see, e.g.,][]{Bacher-Hicks-Koedel-2023-HandbookoftheEconomicsofEducation}.
}
Formally, she solves
\begin{equation}
    \label{eq:problem}
    \max_{(\eta_t)_{t\in\Tcal}}\,(s_T-s_0)\ \ 
    \text{subject to}\ 
    \eqref{eq:accumulation}
    \ \text{and}\ 
    0<\eta_t\le1\ 
    \text{for each}\ t\in\Tcal.
    \tag{P}
\end{equation}
I characterize the solution to~\eqref{eq:problem} in \cref{sec:solution}.%
\footnote{
A solution to~\eqref{eq:problem} exists for all parameter values because choosing~$\eta_t=1$ for each~$t\in\Tcal$ is always feasible.
\cref{thm:sb-solution} implies the solution is unique when~\eqref{eq:sb-restriction-1} and~\eqref{eq:sb-restriction-2} hold.
If they do not hold, then~\eqref{eq:problem} can have many solutions because, for example, the student could have final skill~$s_T=s_0$ regardless of the chosen sequence~$(\eta_t)_{t\in\Tcal}$.
}

\paragraph{Summary}

\begin{table}
    \small
    \centering
    \caption{Model summary}
    \label{tab:model}
    \begin{tabular}{cll}
        \toprule
        Object & Description & Type \\
        \midrule
        $T$ & Number of tasks & Primitive \\
        $s_0$ & Initial skill level & Primitive \\
        $\pi$ & Student's per-unit valuation of output & Primitive \\
        $r$ & Skill accumulation rate & Primitive \\
        $\ybar$ & Output generated by AI & Primitive \\
        $s_t$ & Skill level when student faces task~$t$ & State variable \\
        $\eta_t$ & Effort intensity of task~$t$ & Control variable \\
        \bottomrule
    \end{tabular}
\end{table}
\cref{tab:model} summarizes the exogenous and endogenous objects in my model.
The number of tasks~$T$, initial skill level~$s_0>0$, motivation parameter~$\pi>0$, skill accumulation rate~$r>0$, and AI-generated output~$\ybar\ge0$ are exogenous primitives.
(The delegation payoff~$\ubar\equiv\pi\ybar$ is determined by~$\pi$ and~$\ybar$.)
The student's skill level~$s_t>0$ when he faces task~$t$ is a state variable; the task's effort intensity~$\eta_t\in(0,1]$ is a control variable, chosen by the teacher.
The sequences~$(s_t)_{t\ge1}$ and~$(\eta_t)_{t\ge0}$ are endogenous: they depend on the student's choices, the teacher's choices, and the primitives.
Consequently, the student's best-response efforts~$x^*(s_t,\eta_t)$ and payoffs~$u^*(s_t,\eta_t)$ also depend on the primitives.
I suppress this dependence from my discussion to economize on notation.%
\footnote{
The endogenous objects also depend on the AI quality parameter~$q$, which I introduce in \cref{sec:complement}.
I suppress the objects' dependence on~$q$ until \cref{sec:complement} but make the dependence explicit in my proofs.
}

\section{Best-responses}
\label{sec:br}

Suppose the student has skill~$s>0$ and faces a task with effort intensity~$\eta\in(0,1]$.
If he works on the task, then he exerts best-response effort~$x^*(s,\eta)$ and yields payoff~$u^*(s,\eta)$.
I characterize this effort and payoff in \cref{lem:br-effort,lem:br-payoff}.

\begin{lemma}[Best-response efforts]
    \label{lem:br-effort}
    Given a skill level~$s>0$ and effort intensity~$\eta\in(0,1]$, the student exerts best-response effort
    \begin{equation}
        \label{eq:br-effort}
        x^*(s,\eta)=s\left(\frac{\pi\eta}{s}\right)^{1/(2-\eta)}.
    \end{equation}
    This effort is
    (i)~positive,
    (ii)~increasing in the motivation parameter~$\pi$,
    (iii)~increasing in~$s$ when~$\eta<1$ and constant in~$s$ when~$\eta=1$, and
    (iv)~increasing in~$\eta$ if
    \begin{equation}
        \label{eq:br-effort-intensity-condition}
        \eta\exp\left(\frac{2}{\eta}-1\right)>\frac{s}{\pi}
    \end{equation}
    and non-increasing in~$\eta$ otherwise.
\end{lemma}
\ifbodyproofs\subsection{Proof of \texorpdfstring{\cref{lem:br-effort}}{Lemma~\ref{lem:br-effort}}}

\cref{lem:br-effort} is the special case of \cref{lem:complement-br-effort} with~$q=1$.

\begin{lemma}[Best-response efforts with complementary AI]
    \label{lem:complement-br-effort}
    Given a skill level~$s>0$ and effort intensity~$\eta\in(0,1]$, the student exerts best-response effort~$x^*(s,\eta;q)$ defined as in~\eqref{eq:complement-br-effort}.
    This effort is
    (i)~positive,
    (ii)~increasing in the motivation parameter~$\pi$,
    (iii)~increasing in~$s$ when~$\eta<1$ and constant in~$s$ when~$\eta=1$,
    (iv)~increasing in~$\eta$ if
    \begin{equation}
        \label{eq:complement-br-effort-intensity-condition}
        \eta\exp\left(\frac{2}{\eta}-1\right)>\frac{s}{\pi q^2}
    \end{equation}
    and non-increasing in~$\eta$ otherwise, and
    (v)~increasing in the AI quality parameter~$q$.
\end{lemma}
\begin{proof}
    The student's payoff~$u(x,s,\eta;q)$ is concave in~$x$, so his best-response effort~$x^*\equiv x^*(s,\eta;q)$ satisfies the first-order condition
    \begin{equation}
        \label{eq:br-effort-foc}
        0
        = \parfrac{u(x^*,s,\eta;q)}{x} \\
        = \pi\eta q^\eta(x^*)^{\eta-1} s^{1-\eta}-x^*,
    \end{equation}
    which can be rearranged for
    \[ x^*=s\left(\frac{\pi\eta q^\eta}{s}\right)^{1/(2-\eta)}. \]
    Clearly~$x^*$ is positive since~$s$, $\eta$, $\pi$, and~$q$ are positive.

    Now~$x\mapsto\log(x)$ is an increasing transformation on~$(0,\infty)$, so~$x^*$ is increasing in a parameter if and only if
    \[ \log(x^*)=\frac{1}{2-\eta}\log(\pi\eta)+\frac{\eta}{2-\eta}\log(q)+\frac{1-\eta}{2-\eta}\log(s) \]
    is increasing in that parameter.
    Differentiating with respect to~$\pi$ gives
    \[ \parfrac{\log(x^*)}{\pi}=\frac{1}{\pi(2-\eta)}>0. \]
    Differentiating with respect to~$s$ gives
    \[ \parfrac{\log(x^*)}{s}=\frac{1-\eta}{s(2-\eta)}\ge0 \]
    with equality if and only if~$\eta=1$.
    Differentiating with respect to~$\eta$ gives
    \begin{align*}
        \parfrac{\log(x^*)}{\eta}
        &= \frac{1}{(2-\eta)^2}\log(\pi\eta)+\frac{1}{\eta(2-\eta)}+\frac{2}{(2-\eta)^2}\log(q)-\frac{1}{(2-\eta)^2}\log(s) \\
        &= \frac{1}{(2-\eta)^2}\left(\frac{2-\eta}{\eta}-\log\left(\frac{s}{\pi\eta q^2}\right)\right),
    \end{align*}
    which is positive if and only if~\eqref{eq:complement-br-effort-intensity-condition} holds.
    Finally, differentiating with respect to~$q$ gives
    \[ \parfrac{\log(x^*)}{q}=\frac{\eta}{q(2-\eta)}>0.\qedhere \]
\end{proof}
\fi

\begin{lemma}[Best-response payoffs]
    \label{lem:br-payoff}
    Given a skill level~$s>0$ and effort intensity~$\eta\in(0,1]$, the student's best-response effort~$x^*(s,\eta)$ yields payoff
    \begin{equation}
        \label{eq:br-payoff}
        u^*(s,\eta)=\frac{1}{\eta}\left(1-\frac{\eta}{2}\right)\left(x^*(s,\eta)\right)^2.
    \end{equation}
    This payoff is
    (i)~positive,
    (ii)~increasing in the motivation parameter~$\pi$,
    (iii)~increasing in~$s$ when~$\eta<1$ and constant in~$s$ when~$\eta=1$, and
    (iv)~decreasing in~$\eta$ when~$\eta<s/\pi$ and non-decreasing in~$\eta$ otherwise.
\end{lemma}
\ifbodyproofs\subsection{Proof of \texorpdfstring{\cref{lem:br-payoff}}{Lemma~\ref{lem:br-payoff}}}
\label{appsec:br-payoff}

\cref{lem:br-payoff} is the special case of \cref{lem:complement-br-payoff} with~$q=1$.

\begin{lemma}[Best-response payoffs with complementary AI]
    \label{lem:complement-br-payoff}
    Given a skill level~$s>0$ and effort intensity~$\eta\in(0,1]$, the student's best-response effort~$x^*(s,\eta;q)$ yields payoff~$u^*(s,\eta;q)$ defined as in~\eqref{eq:complement-br-payoff}.
    This payoff is
    (i)~positive,
    (ii)~increasing in the motivation parameter~$\pi$,
    (iii)~increasing in~$s$ when~$\eta<1$ and constant in~$s$ when~$\eta=1$,
    (iv)~decreasing in~$\eta$ when~$\eta<s/\pi q^2$ and non-decreasing in~$\eta$ otherwise, and
    (v)~increasing in the AI quality parameter~$q$.
\end{lemma}
\begin{proof}
    By \cref{lem:complement-br-effort}, the best-response effort~$x^*\equiv x^*(s,\eta;q)$ yields payoff
    \begin{align*}
        u^*(s,\eta;q)
        &\equiv u(x^*,s,\eta;q) \\
        &= \pi(qx^*)^{\eta}s^{1-\eta}-\frac{(x^*)^2}{2} \\
        &\overset{\star}{=} \frac{1}{\eta}\left(1-\frac{\eta}{2}\right)(x^*)^2
    \end{align*}
    where~$\star$ uses the substitution~$\pi\eta(qx^*)^{\eta}s^{1-\eta}=(x^*)^2$ implied by the first-order condition~\eqref{eq:br-effort-foc}.
    Parts~(i)--(iii) and~(v) follow immediately from \cref{lem:complement-br-effort}.
    For part~(iv), we have
    \begin{align*}
        \parfrac{x^*}{\eta}
        &= x^*\parfrac{\log(x^*)}{\eta} \\
        &= \frac{x^*}{(2-\eta)^2}\left(\frac{2-\eta}{\eta}-\log\left(\frac{s}{\pi\eta q^2}\right)\right)
    \end{align*}
    from the proof of \cref{lem:complement-br-effort} and therefore
    \begin{align*}
        \parfrac{u^*(s,\eta;q)}{\eta}
        &= -\frac{1}{\eta^2}(x^*)^2+\frac{2}{\eta}\left(1-\frac{\eta}{2}\right)x^*\parfrac{x^*}{\eta} \\
        &= -\frac{(x^*)^2}{\eta(2-\eta)}\log\left(\frac{s}{\pi\eta q^2}\right) \\
        &= -\frac{2u^*(s,\eta;q)}{(2-\eta)^2}\log\left(\frac{s}{\pi\eta q^2}\right),
    \end{align*}
    which is positive if and only if~$\eta>s/\pi q^2$.
\end{proof}

Together, \cref{lem:complement-br-effort,lem:complement-br-payoff} imply
\begin{align}
    u^*(s,\eta;q)
    &= \frac{1}{\eta}\left(1-\frac{\eta}{2}\right)\left(s\left(\frac{\pi\eta q^\eta}{s}\right)^{1/(2-\eta)}\right)^2 \notag \\
    &= \left(1-\frac{\eta}{2}\right)\pi^{2/(2-\eta)}\eta^{\eta/(2-\eta)}q^{2\eta/(2-\eta)}s^{2(1-\eta)/(2-\eta)} \label{eq:complement-br-payoff-expanded}
\end{align}
for all skill levels~$s>0$ and effort intensities~$\eta\in(0,1]$.
\fi

I prove \cref{lem:br-effort,lem:br-payoff}, and all subsequent results, in \cref{app:proofs}.

The best-response effort~$x^*(s,\eta)$ and payoff~$u^*(s,\eta)$ are increasing in the motivation parameter~$\pi$.
Intuitively, the student is more willing to exert effort when he is more motivated to produce.

If a task is fully effort-intensive (i.e., if~$\eta=1$), then the student's skill when he faces the task is irrelevant, and he simply trades off the marginal product and cost of effort.
Indeed, if~$\eta=1$, then his best-response effort~$x^*(s,1)=\pi$ and payoff~$u^*(s,1)=\pi^2/2$ do not depend on his skill level~$s$.
However, if a task is \emph{not} fully effort-intensive (i.e., if~$\eta<1$), then having more skill makes the marginal product of effort higher and makes the student more willing to exert effort.
Indeed, if~$\eta<1$, then~$x^*(s,\eta)$ and~$u^*(s,\eta)$ are both increasing in~$s$.

\begin{figure}
    \centering
    \includegraphics[width=0.55\linewidth]{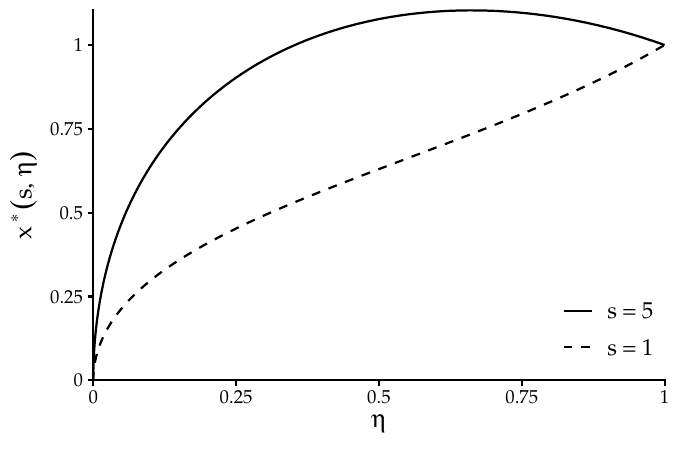}
    \caption{
    Best-response efforts~$x^*(s,\eta)$ when~$\pi=1$ and~$s\in\{1,5\}$.
    }
    \label{fig:br-effort}
\end{figure}
\cref{fig:br-effort} shows how~$x^*(s,\eta)$ depends on~$\eta$ when~$\pi=1$ and~$s\in\{1,5\}$.
If~$s=1$, then the condition~\eqref{eq:br-effort-intensity-condition} holds for all~$\eta\in(0,1]$ and so \cref{lem:br-effort} implies~$x^*(1,\eta)$ is increasing in~$\eta$.
In contrast, if~$s=5$, then~\eqref{eq:br-effort-intensity-condition} holds precisely when~$\eta$ is small and so~$x^*(5,\eta)$ has an inverted-U-shaped relationship with~$\eta$.
In general, for all~$s>0$, there is a unique effort intensity that maximizes~$x^*(s,\eta)$.
I characterize this maximizer in \cref{sec:fb}.

\section{Optimal courses}
\label{sec:solution}

This section characterizes the effort intensities~$(\eta_t)_{t\in\Tcal}$ that solve the teacher's problem~\eqref{eq:problem}.
She wants to maximize the sum of skill increments~$\Delta(s_t,\eta_t)\equiv s_{t+1}-s_t$ defined by~\eqref{eq:accumulation}.
Each increment depends on whether the student works or delegates, and, if he works, how much effort he exerts as a best-response to his skill level and the task's effort intensity.

The student works on task~$t\in\Tcal$ if and only if his best-response payoff~$u^*(s_t,\eta_t)$ is at least the delegation payoff~$\ubar\equiv\pi\ybar$:
\begin{equation}
    \label{eq:ic-constraint}
    u^*(s_t,\eta_t)\ge\ubar. \tag{IC}
\end{equation}
By \cref{lem:br-effort,lem:br-payoff}, the best-response effort~$x^*(s_t,\eta_t)$ and payoff~$u^*(s_t,\eta_t)$ are non-decreasing in the student's skill level~$s_t$ when he faces task~$t$.
It follows that the skill increment
\[ \Delta(s_t,\eta_t)=\begin{cases}
        r\,x^*(s_t,\eta_t) & \text{if~\eqref{eq:ic-constraint} holds} \\
        0 & \text{otherwise}
\end{cases} \]
is also non-decreasing in~$s_t$.
So raising the skill gained on task~$t$ cannot lower the skill gained on subsequent tasks.
Thus, for each task~$t\in\Tcal$, the teacher chooses an effort intensity~$\eta_t\in(0,1]$ that maximizes the student's best-response effort~$x^*(s_t,\eta_t)$ subject to the incentive constraint~\eqref{eq:ic-constraint}:
\begin{equation}
    \label{eq:task-solution}
    \eta_t\in\argmax_{\eta\in(0,1]}\,x^*(s_t,\eta)\ \ \text{subject to}\ \ u^*(s_t,\eta)\ge\ubar.
\end{equation}
I analyze the relationship between~$\eta_t$ and~$s_t$ below.

\subsection{First-best courses}
\label{sec:fb}

Consider the ``first best'' case in which AI is unproductive: it generates output~$\ybar=0$ and delegation payoff~$\ubar=0$.
In this case, the student prefers to work on every task regardless of its effort intensity.%
\footnote{
If~$\ybar=0$, then~$u^*(s,\eta)\ge u(0,s,\eta)=0=\ubar$ for all~$s>0$ and~$\eta\in(0,1]$, and so the student prefers to work on every task regardless of his skill level or the task's effort intensity.
}
As a result, the teacher faces no incentive constraints when she chooses tasks' effort intensities; she can simply choose them to maximize the student's best-response efforts.

\subsubsection{Optimal effort intensities}
\label{sec:fb-solution}

Consider the condition~\eqref{eq:br-effort-intensity-condition} that determines whether the best-response effort~$x^*(s,\eta)$ is increasing in the effort intensity~$\eta$.
The left-hand side of~\eqref{eq:br-effort-intensity-condition} is continuous and decreasing in~$\eta$ over~$(0,1]$, and attains its minimum of~$e\equiv\exp(1)$ when~$\eta=1$.
So if~$s<\pi e$, then~\eqref{eq:br-effort-intensity-condition} holds for all~$\eta\in(0,1]$ and the teacher maximizes~$x^*(s,\eta)$ by choosing~$\eta=1$.
In contrast, if~$s>\pi e$, then the intermediate value theorem implies there is a unique~$\eta\FB(s)\in(0,1)$ satisfying
\begin{equation}
    \label{eq:fb-intensity-foc}
    \eta\FB(s)\exp\left(\frac{2}{\eta\FB(s)}-1\right)=\frac{s}{\pi},
\end{equation}
and the teacher maximizes~$x^*(s,\eta)$ by choosing~$\eta=\eta\FB(s)$.
Moreover, since the left-hand side of~\eqref{eq:br-effort-intensity-condition} is decreasing in~$\eta$ but the right-hand side is increasing in the skill level~$s$, the ``first-best effort intensity''~$\eta\FB(s)$ must decline as~$s$ rises.
I illustrate this decline in \cref{fig:fb-intensity}.
\begin{figure}
    \centering
    \includegraphics[width=0.55\linewidth]{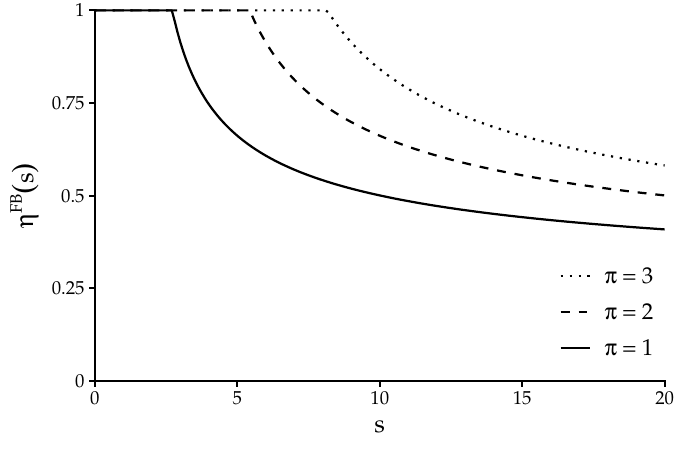}
    \caption{
    First-best effort intensities~$\eta\FB(s)$ when~$\pi\in\{1,2,3\}$.
    }
    \label{fig:fb-intensity}
\end{figure}

\begin{lemma}[First-best effort intensities]
    \label{lem:fb-intensity}
    For all skill levels~$s>0$, there is a unique~$\eta\FB(s)\in(0,1]$ such that
    (i)~if~$s\le\pi e$, then~$\eta\FB(s)=1$, and
    (ii)~if~$s>\pi e$, then~$\eta\FB(s)\in(0,1)$ satisfies~\eqref{eq:fb-intensity-foc} and is decreasing in~$s$.
    Moreover, the best-response effort~$x^*(s,\eta)$ attains its maximum over~$\eta\in(0,1]$ when~$\eta=\eta\FB(s)$.
\end{lemma}
\ifbodyproofs\subsection{Proof of \texorpdfstring{\cref{lem:fb-intensity}}{Lemma~\ref{lem:fb-intensity}}}
\label{appsec:fb-intensity}

\cref{lem:fb-intensity} is the special case of \cref{lem:complement-fb-intensity} with~$q=1$.

\begin{lemma}[First-best effort intensities with complementary AI]
    \label{lem:complement-fb-intensity}
    For all skill levels~$s>0$, there is a unique~$\eta\FB(s;q)\in(0,1]$ such that:
    \begin{enumerate}

        \item[(i)]
        if~$s\le\pi eq^2$, then~$\eta\FB(s;q)=1$;

        \item[(ii)]
        if~$s>\pi eq^2$, then~$\eta\FB(s;q)\in(0,1)$ satisfies~\eqref{eq:complement-fb-intensity-foc}, is decreasing in~$s$, and is increasing in the motivation parameter~$\pi$ and AI quality parameter~$q$;

        \item[(iii)]
        the best-response effort~\eqref{eq:complement-br-effort} attains its maximum over~$\eta\in(0,1]$ when~$\eta=\eta\FB(s;q)$.

    \end{enumerate}
    Moreover, $\eta\FB(s;q)$ is continuous in both~$s$ and~$q$, and converges to zero as~$s\to\infty$.
\end{lemma}
\begin{proof}
    By \cref{lem:complement-br-effort}, the best-response effort~$x^*(s,\eta;q)$ is increasing in~$\eta$ when~\eqref{eq:complement-br-effort-intensity-condition} holds and non-increasing otherwise.
    The left-hand side of~\eqref{eq:complement-br-effort-intensity-condition} falls continuously from~$\infty$ to~$e$ as~$\eta$ rises from zero to one.
    So there are two cases to consider:
    \begin{enumerate}

        \item[(i)]
        Suppose~$s\le\pi eq^2$.
        Then~\eqref{eq:complement-br-effort-intensity-condition} holds for all~$\eta\in(0,1)$ and so, by \cref{lem:br-effort}, the best-response effort~$x^*(s,\eta;q)$ attains its maximum over~$\eta\in(0,1]$ precisely when~$\eta=1$.

        \item[(ii)]
        Now suppose~$s>\pi eq^2$.
        Then, by the intermediate value theorem, there is a unique effort intensity~$\eta\FB(s;q)\in(0,1)$ satisfying~\eqref{eq:complement-fb-intensity-foc} such that~$x^*(s,\eta;q)$ attains its maximum over~$\eta\in(0,1]$ when~$\eta=\eta\FB(s;q)$.
        The left-hand side of~\eqref{eq:complement-fb-intensity-foc} is decreasing in~$\eta\FB(s;q)$; the right-hand side is increasing in~$s$, and decreasing in~$\pi$ and~$q$.
        It follows that~$\eta\FB(s;q)$ is decreasing in~$s$, and increasing in~$\pi$ and~$q$.

    \end{enumerate}
    Parts~(i)--(iii) of the result follow.
    It remains to show~$\eta\FB(s;q)$ is continuous in~$s$ and~$q$, and converges to zero as~$s\to\infty$.

    The left-hand side of~\eqref{eq:complement-fb-intensity-foc} is continuously differentiable in~$\eta\FB(s;q)$ and has a non-zero derivative with respect to~$\eta\FB(s;q)$.
    The implicit function theorem then implies~$\eta\FB(s;q)$ is (differentiable and thus) continuous in both~$s$ and~$q$ when~$s>\pi eq^2$.
    We also have~$\eta\FB(s;q)\to1$ as~$s\to\pi eq^2$ from above, from which it follows that~$\eta\FB(s;q)$ is globally continuous in~$s$ and~$q$.

    Finally, the right-hand side of~\eqref{eq:complement-fb-intensity-foc} grows without bound as~$s\to\infty$, while the left-hand side diverges to~$\infty$ only as~$\eta\FB(s;q)\to0$.
    Thus~$\eta\FB(s;q)\to0$ as~$s\to\infty$.
\end{proof}
\fi

\cref{thm:fb-solution} uses the function~$\eta\FB:(0,\infty)\to(0,1]$ to characterize first-best courses.

\begin{theorem}[First-best courses]
    \label{thm:fb-solution}
    Suppose AI generates output~$\ybar=0$ and~$(\eta_t)_{t\in\Tcal}$ solves the teacher's problem~\eqref{eq:problem}.
    Then~$\eta_t=\eta\FB(s_t)$ for each~$t\in\Tcal$.
    Moreover, there exists~$t\FB\in\Tcal\cup\{T\}$ such that~$\eta_t=1$ when~$t<t\FB$ and~$\eta_t$ is decreasing in~$t$ over~$\{t\FB,\ldots,T-1\}$.
\end{theorem}
\ifbodyproofs\subsection{Proof of \texorpdfstring{\cref{thm:fb-solution}}{Theorem~\ref{thm:fb-solution}}}
\label{appsec:fb-solution}

In the general setting described in \cref{sec:complement}, the student gains skill
\begin{equation}
    \label{eq:complement-accumulation}
    s_{t+1}-s_t=\begin{cases}
        r\,x^*(s_t,\eta_t;q) & \text{if~\eqref{eq:complement-ic-constraint} holds} \\
        0 & \text{if otherwise}
    \end{cases}
\end{equation}
on each task~$t\in\Tcal$.
The teacher chooses~$(\eta_t)_{t\in\Tcal}$ to solve
\begin{equation}
    \label{eq:complement-problem}
    \max_{(\eta_t)_{t\in\Tcal}}(s_T-s_0)\ \ 
    \text{subject to}\ 
    \eqref{eq:complement-accumulation}
    \ \text{and}\ 
    0<\eta_t\le1\ 
    \text{for each}\ t\in\Tcal.
    \tag{\ref{eq:problem}$;q$}
\end{equation}
Following the logic at the beginning of \cref{sec:solution}, for each task~$t\in\Tcal$ we have
\begin{equation}
    \label{eq:complement-task-solution}
    \eta_t\in\argmax_{\eta\in(0,1]}\,x^*(s_t,\eta;q)\ \ \text{subject to}\ \ u^*(s_t,\eta;q)\ge\pi q\ybar.
\end{equation}
Letting~$q=1$ recovers~\eqref{eq:task-solution}.
So \cref{thm:fb-solution} is the special case of \cref{thm:complement-fb-solution} with~$q=1$.

\begin{theorem}[First-best courses with complementary AI]
    \label{thm:complement-fb-solution}
    Suppose~$\ybar=0$ and~$(\eta_t)_{t\in\Tcal}$ solves~\eqref{eq:complement-problem}.
    Then~$\eta_t=\eta\FB(s_t;q)$ for each task~$t\in\Tcal$.
    Moreover, there exists~$t\FB\in\Tcal\cup\{T\}$ such that~$\eta_t=1$ when~$t<t\FB$ and~$\eta_t$ is decreasing in~$t$ over~$\{t\FB,\ldots,T-1\}$.
\end{theorem}
\begin{proof}
    \cref{lem:complement-fb-intensity} implies
    \[ \argmax_{\eta\in(0,1]}\,x^*(s_t,\eta;q)=\left\{\eta\FB(s_t;q)\right\} \]
    for each~$t\in\Tcal$.
    Moreover, since~$\ybar=0$, \cref{lem:complement-br-payoff} implies~$u^*(s_t,\eta;q)>\pi q\ybar$ for all~$\eta\in(0,1]$.
    It follows from~\eqref{eq:complement-task-solution} that~$\eta_t=\eta\FB(s_t;q)$ for each~$t\in\Tcal$.

    Since the student works on each task, we have~$s_{t+1}>s_t$ for each~$t\in\Tcal$ by \cref{lem:complement-br-effort}.
    So if~$s_T\le\pi eq^2$, then \cref{lem:complement-fb-intensity} implies~$\eta_t=1$ for each~$t\in\Tcal$ and letting~$t\FB=T$ yields the result.
    In contrast, if~$s_T>\pi eq^2$, then there is a least~$t\FB\in\Tcal$ such that~$s_t>\pi eq^2$ for each~$t\ge t\FB$, and \cref{lem:complement-fb-intensity} then implies~$\eta_t=1$ when~$t<t\FB$ and~$\eta_{t+1}<\eta_t$ when~$t\ge t\FB$.
\end{proof}
\fi

If the student has skill~$s_t$ when he faces task~$t$, then the first-best course sets~$\eta_t=\eta\FB(s_t)$.
The student works on the task, exerts best-response effort~$x\FB(s_t)\equiv x^*(s_t,\eta\FB(s_t))$, and increases his skill to~$s_{t+1}=s_t+rx\FB(s_t)>s_t$.
The next task has effort intensity~$\eta\FB(s_{t+1})$, which is at most as large as~$\eta\FB(s_t)$ because~$\eta\FB$ is non-increasing on its domain and~$s_{t+1}>s_t$.

Thus, in the first-best course, earlier tasks are more effort-intensive and later tasks are more skill-intensive.
The intuition is as follows.
On all tasks, the teacher wants to induce as much effort as possible because this maximizes the student's skill development.
On earlier tasks, the student has less skill and so the teacher induces effort by making it contribute more to output.
On later tasks, the student has more skill, which makes effort more productive because effort and skill are complementary inputs.
The teacher capitalizes on this complementarity by making tasks less effort-intensive and more skill-intensive.%
\footnote{
In \cref{sec:cobb-douglas}, I show that the constant-then-decreasing shape of~$\eta\FB(s)$ relies on effort and skill being complementary inputs.
}

\subsubsection{Best-responses}
\label{sec:fb-br}

For all skill levels~$s>0$, let
\[ x\FB(s)\equiv x^*(s,\eta\FB(s))
    \ \ \text{and}\ \ 
    u\FB(s)\equiv u^*(s,\eta\FB(s)) \]
denote the student's best-response effort and payoff when the teacher designs a first-best course.
Both are larger when~$s$ is larger:

\begin{proposition}[First-best outcomes]
    \label{prop:fb-br-s}
    The functions~$x\FB:(0,\infty)\to(0,\infty)$ and~$u\FB:(0,\infty)\to(0,\infty)$ are continuous, constant on~$(0,\pi e]$, and increasing on~$(\pi e,\infty)$.
\end{proposition}
\ifbodyproofs\subsection{Proof of \texorpdfstring{\cref{prop:fb-br-s}}{Proposition~\ref{prop:complement-fb-br-s}}}

\cref{prop:fb-br-s} is the special case of \cref{prop:complement-fb-br-s} with~$q=1$.

\begin{proposition}[First-best outcomes with complementary AI]
    \label{prop:complement-fb-br-s}
    The first-best effort~$x\FB(s;q)$ and payoff~$u\FB(s;q)$ defined in \cref{sec:complement-fb} are
    continuous in~$s$,
    constant in~$s$ when~$s\le\pi eq^2$, and
    increasing in~$s$ when~$s>\pi eq^2$.
\end{proposition}
\begin{proof}
    Suppose~$s\le\pi eq^2$.
    Then~$\eta\FB(s;q)=1$ by \cref{lem:complement-fb-intensity}, and so~$x\FB(s;q)=x^*(s,1;q)=\pi q$ and~$u\FB(s;q)=u^*(s,1;q)=\pi^2q^2/2$ are constant in~$s$.

    Now suppose~$s>\pi eq^2$.
    Then~$\eta\FB(s;q)<1$ is the interior solution to
    \[ \max_\eta x^*(s,\eta;q)\ \ \text{subject to}\ \ 0<\eta\le1. \]
    So the envelope theorem implies
    \begin{align*}
        \parfrac{x\FB(s;q)}{s}
        &= \parfrac{x^*(s,\eta;q)}{s}\bigg\vert_{\eta=\eta\FB(s;q)} \\
        &= \frac{1-\eta}{2-\eta}\left(\frac{\pi\eta q^\eta}{s}\right)^{1/(2-\eta)}\bigg\vert_{\eta=\eta\FB(s;q)} \\
        &= \frac{1-\eta\FB(s;q)}{2-\eta\FB(s;q)}\cdot\frac{x\FB(s;q)}{s} \\
        &> 0.
    \end{align*}
    Thus~$x\FB(s;q)$ is increasing in~$s$.
    Then, by \cref{lem:complement-br-payoff,lem:complement-fb-intensity}, the best-response payoff
    \[ u\FB(s;q)=\frac{1}{\eta\FB(s;q)}\cdot\left(1-\frac{\eta\FB(s;q)}{2}\right)\cdot(x\FB(s;q))^2 \]
    equals the product of three factors that are increasing in~$s$.
    Thus~$u\FB(s;q)$ is also increasing in~$s$.

    Finally, since~$x\FB(s;q)$ and~$u\FB(s;q)$ are continuous functions of~$s$ and~$\eta\FB(s;q)$, and~$\eta\FB(s;q)$ is continuous in~$s$, both~$x\FB(s;q)$ and~$u\FB(s;q)$ are continuous in~$s$.
\end{proof}
\fi

I use \cref{prop:fb-br-s} to derive suitable parameter restrictions in my analysis of ``second-best'' courses below.

\subsection{Second-best courses}
\label{sec:sb}

Now I allow the AI-generated output~$\ybar$ and delegation payoff~$\ubar\equiv\pi\ybar$ to exceed zero.
This reintroduces the incentive constraint~\eqref{eq:ic-constraint} on the effort intensities~$(\eta_t)_{t\in\Tcal}$.

\subsubsection{Parameter restrictions}

Suppose the teacher chooses~$\eta_t=\eta\FB(s_t)$ for each task~$t\in\Tcal$.
Then, by \cref{prop:fb-br-s}, the student's first-best payoff~$u\FB(s_t)$ on task~$t$ is non-decreasing in his skill level~$s_t$ when he faces the task.
But~$s_t$ is non-decreasing in~$t$ and so~$u\FB(s_t)$ is also non-decreasing in~$t$.
Thus, if the constraint~\eqref{eq:ic-constraint} binds the choice of~$\eta_t$ for \emph{any} task~$t\in\Tcal$, then it binds the choice of~$\eta_t$ for the initial task~$t=0$.
For this reason, I assume
\begin{equation}
    \label{eq:sb-restriction-1}
    u\FB(s_0)<\ubar; \tag{SB}
\end{equation}
that is, the student prefers to delegate the initial task in the first-best course.
This is the minimal assumption needed to ensure the first- and second-best courses differ: if~\eqref{eq:sb-restriction-1} does not hold, then the first-best course induces work on every task and so the constraint~\eqref{eq:ic-constraint} has no bite.

Now consider the student's best-response payoff~$u^*(s_0,\eta_0)$ from working on the initial task.
If~$u^*(s_0,\eta)<\ubar$ for all effort intensities~$\eta\in(0,1]$, then it is not possible for the teacher to induce work on the initial task, and so~$s_1=s_0$ regardless of the teacher's chosen~$\eta_0$.
But then~$u^*(s_1,\eta)=u^*(s_0,\eta)<\ubar$ for all~$\eta\in(0,1]$, implying the teacher cannot induce work on task~$t=1$.
Iterating this argument over~$t\in\Tcal$ yields~$s_T=s_{T-1}=\cdots=s_0$ regardless of~$(\eta_t)_{t\in\Tcal}$.
Thus, if~$u^*(s_0,\eta)<\ubar$ for all~$\eta\in(0,1]$, then the teacher's problem~\eqref{eq:problem} is degenerate because she cannot induce work on \emph{any} tasks.
To avoid this degeneracy, I assume
\begin{equation}
    \label{eq:sb-restriction-2}
    \sup_{\eta\in(0,1]}u^*(s_0,\eta)>\ubar; \tag{ND}
\end{equation}
that is, it is possible for the teacher to induce work on the initial task (and, thus, subsequent tasks).

Assumptions~\eqref{eq:sb-restriction-1} and~\eqref{eq:sb-restriction-2} impose bounds on the AI-generated output~$\ybar\equiv\ubar/\pi$:

\begin{lemma}[Parameter restrictions]
    \label{lem:sb-restriction}
    ~
    \begin{enumerate}

        \item[(i)]
        If~\eqref{eq:sb-restriction-1} holds, then~$\pi/2<\ybar$.

        \item[(ii)]
        If~\eqref{eq:sb-restriction-1} and~\eqref{eq:sb-restriction-2} hold, then~$\ybar<s_0$.

    \end{enumerate}
\end{lemma}
\ifbodyproofs\subsection{Proof of \texorpdfstring{\cref{lem:sb-restriction}}{Lemma~\ref{lem:sb-restriction}}}
\label{appsec:sb-restriction}

\cref{lem:sb-restriction} is the special case of \cref{lem:complement-sb-restriction} with~$q=1$.

\begin{lemma}[Parameter restrictions with complementary AI]
    \label{lem:complement-sb-restriction}
    ~
    \begin{enumerate}

        \item[(i)]
        If~\eqref{eq:complement-sb-restriction-1} holds, then~$\pi q/2<\ybar$.

        \item[(ii)]
        If~\eqref{eq:complement-sb-restriction-1} and~\eqref{eq:complement-sb-restriction-2} hold, then~$q\ybar<s_0$.

    \end{enumerate}
\end{lemma}
\begin{proof}
    If~$s_0\le\pi eq^2$, then~$u\FB(s_0;q)=u^*(s_0,1;q)=\pi^2q^2/2$.
    If~$s_0>\pi eq^2$, then~$\eta\FB(s_0;q)<1<e<s_0/\pi q^2$ and so~$u\FB(s_0;q)=u^*(s_0,\eta\FB(s_0;q);q)>u^*(s_0,1;q)=\pi^2q^2/2$ by \cref{lem:complement-br-payoff}.
    Thus~$u\FB(s_0;q)\ge\pi^2q^2/2$ independently of~$s_0$.
    So if~\eqref{eq:complement-sb-restriction-1} holds, then~$\pi^2q^2/2<q\ubar=\pi q\ybar$ and thus~$\pi q/2<\ybar$.

    Suppose~\eqref{eq:complement-sb-restriction-1} and~\eqref{eq:complement-sb-restriction-2} hold.
    Then, by \cref{lem:complement-br-payoff}, the best-response payoff~$u^*(s_0,\eta;q)$ is decreasing in~$\eta$ when~$\eta<s_0/\pi q^2$ and increasing in~$\eta$ when~$\eta>s_0/\pi q^2$.
    It follows that
    \begin{align}
        \sup_{\eta\in(0,1]}u^*(s_0,\eta;q)
        &= \max\left\{\lim_{\eta\to0}u^*(s_0,\eta;q),\,u^*(s_0,1;q)\right\} \notag \\
        &= \max\left\{\pi s_0,\,\frac{\pi^2q^2}{2}\right\} \label{eq:sb-restriction-supremum}.
    \end{align}
    But~$\pi^2q^2/2<q\ubar$ from above.
    So if~$\pi s_0\le\pi^2q^2/2$, then~\eqref{eq:sb-restriction-supremum} contradicts~\eqref{eq:complement-sb-restriction-2}.
    Thus~$\pi s_0>\pi^2q^2/2$ and therefore
    \[ q\ybar=\frac{q\ubar}{\pi}<\frac{1}{\pi}\sup_{\eta\in(0,1]}u^*(s_0,\eta;q)=s_0.\qedhere \]
\end{proof}
\fi

If assumptions~\eqref{eq:sb-restriction-1} and~\eqref{eq:sb-restriction-2} hold, then the AI-generated output~$\ybar$ must exceed~$\pi/2$ and be smaller than~$s_0$.
The intuition is as follows.
If the initial task is fully effort-intensive ($\eta_0=1$), then the student exerts effort~$x^*(s_0,1)=\pi$ at cost~$\pi^2/2$, and his ``cost-adjusted'' output is~$u^*(s_0,1)/\pi=\pi/2$.
In contrast, if the task is fully skill-intensive ($\eta_0\to0$), then he produces best-response output
\[ \lim_{\eta\to0}\,x^*(s_0,\eta)^\eta s_0^{1-\eta}=s_0. \]
So the lower bound~$\pi/2<\ybar$ says the AI out-produces the student (net of costs) on a fully effort-intensive task, whereas the upper bound~$\ybar<s_0$ says the student out-produces the AI on a fully skill-intensive task.

Since the student's skill level is non-decreasing across tasks, assumptions~\eqref{eq:sb-restriction-1} and~\eqref{eq:sb-restriction-2} imply~$s_t>\ybar$ for each task~$t\in\Tcal$.
Accordingly, I focus on skill levels~$s>\ybar$ for the remainder of this section.%
\footnote{
If the student has skill less than~$\ybar$, then it is impossible for the teacher to ``design away'' the delegation by changing her course.
Instead, she would have to remove the threat by banning the student from using AI.
For example, she could mandate that tasks are completed in-person without access to technology.
However, this would preclude the student from using AI beneficially as a complement---a possibility I consider in \cref{sec:complement}.
}

\subsubsection{Incentive compatibility}
\label{sec:ic-intensity}

Suppose~\eqref{eq:sb-restriction-1} holds and fix a skill level~$s>\ybar$.
By \cref{lem:br-payoff}, the best-response payoff~$u^*(s,\eta)$ attains its minimum over~$\eta\in(0,1]$ at
\[ \eta_\text{min}(s)\equiv\min\left\{\frac{s}{\pi},1\right\}. \]
The payoff falls continuously from
\[ \lim_{\eta\to0}u^*(s,\eta)=\pi s \]
to
\[ u^*(s,\eta_\text{min}(s))=\frac{\pi^2}{2}\begin{cases}
        \left(2-\frac{s}{\pi}\right)\frac{s}{\pi} & \text{if}\ \frac{s}{\pi}\le1 \\
        1 & \text{if}\ \frac{s}{\pi}>1
    \end{cases} \]
as the effort intensity~$\eta$ rises from zero to~$\eta_\text{min}(s)$, then (if~$\eta_\text{min}(s)<1$) rises continuously to
\[ u^*(s,1)=\frac{\pi^2}{2} \]
as~$\eta$ rises to one.
But~$\pi^2/2<\ubar<\pi s$ by \cref{lem:sb-restriction} and the definition of~$\ubar\equiv\pi\ybar$.
So, by the intermediate value theorem, there is a unique effort intensity~$\eta\IC(s)\in\left(0,\eta_\text{min}(s)\right)$ that makes the student indifferent between working and delegating:
\begin{equation}
    \label{eq:ic-intensity-definition}
    u^*(s,\eta\IC(s))=\ubar.
\end{equation}
Moreover, since the best-response payoff~$u^*(s,\eta)$ is decreasing in~$\eta$ over~$(0,\eta_{\text{min}}(s))$ and bounded above by~$\pi^2/2<\ubar$ over~$[\eta_\text{min}(s),1]$, we have~$u^*(s,\eta)\ge\ubar$ if and only if~$\eta\le\eta\IC(s)$.

\begin{lemma}[Incentive-compatible effort intensities]
    \label{lem:ic-intensity}
    Suppose~\eqref{eq:sb-restriction-1} holds.
    For all skill levels~$s>\ybar$, there is a unique effort intensity
    \[ \eta\IC(s)\in\left(0,\min\left\{\frac{s}{\pi},1\right\}\right) \]
    satisfying~\eqref{eq:ic-intensity-definition}.
    The function~$\eta\IC:(\ybar,\infty)\to(0,1)$ is continuous and increasing, with
    \[ \lim_{s\to\ybar}\,\eta\IC(s)=0
    \ \ \text{and}\ \ 
    \lim_{s\to\infty}\,\eta\IC(s)=1. \]
    Moreover, the incentive constraint~\eqref{eq:ic-constraint} holds if and only if~$\eta_t\le\eta\IC(s_t)$.
\end{lemma}
\ifbodyproofs\subsection{Proof of \texorpdfstring{\cref{lem:ic-intensity}}{Lemma~\ref{lem:ic-intensity}}}
\label{appsec:ic-intensity}

\cref{lem:ic-intensity} is the special case of \cref{lem:complement-ic-intensity} with~$q=1$.

\begin{lemma}[Incentive-compatible effort intensities with complementary AI]
    \label{lem:complement-ic-intensity}
    Suppose~\eqref{eq:complement-sb-restriction-1} holds.
    For all skill levels~$s>q\ybar$, there is a unique effort intensity
    \[ \eta\IC(s;q)\in\left(0,\min\left\{\frac{s}{\pi q^2},1\right\}\right) \]
    satisfying~\eqref{eq:complement-ic-intensity-definition}.
    The mapping~$s\mapsto\eta\IC(s;q)$ is continuous and increasing on its domain, and
    \[ \lim_{s\to q\ybar}\eta\IC(s;q)=0
    \ \ \text{and}\ \ 
    \lim_{s\to\infty}\eta\IC(s;q)=1. \]
    Moreover, the incentive constraint~\eqref{eq:complement-ic-constraint} holds if and only if~$\eta_t\le\eta\IC(s_t;q)$.
\end{lemma}
\begin{proof}
    Suppose~\eqref{eq:complement-sb-restriction-1} holds and let~$s>q\ybar$.
    Then~$\pi^2q^2/2<q\ubar<\pi s$ by \cref{lem:complement-sb-restriction} and the fact that~$\ubar\equiv\pi\ybar$.
    By \cref{lem:complement-br-payoff}, the best-response payoff~$u^*(s,\eta;q)$ falls continuously from~$\pi s$
    to
    \[ \min_{\eta\in(0,1]}u^*(s,\eta;q)=\frac{\pi^2q^2}{2}\begin{cases}
        \left(2-\frac{s}{\pi q^2}\right)\frac{s}{\pi q^2} & \text{if}\ \frac{s}{\pi q^2}<1 \\
        1 & \text{if}\ \frac{s}{\pi q^2}\ge1
    \end{cases} \]
    as~$\eta$ rises from zero to
    \[ \eta_\text{min}(s;q)\equiv\min\left\{\frac{s}{\pi q^2},1\right\}, \]
    then (if~$s/\pi q^2<1$) rises continuously to~$u^*(s,1;q)=\pi^2q^2/2$ as~$\eta$ rises to one.
    So, by the intermediate value theorem, there is a unique effort intensity
    \[ \eta\IC(s;q)\in\left(0,\eta_\text{min}(s;q)\right) \]
    satisfying~\eqref{eq:complement-ic-intensity-definition}.
    Moreover, since~$u^*(s,\eta;q)$ is decreasing in~$\eta$ over~$(0,\eta_{\text{min}}(s;q))$ and bounded above by~$\pi^2q^2/2$ over~$[\eta_\text{min}(s;q),1]$, the incentive constraint~$u^*(s,\eta;q)\ge\pi q\ybar$ holds if and only if~$\eta\le\eta\IC(s;q)$.
    
    Next, I show~$\eta\IC(s;q)$ is continuous and increasing in~$s>q\ybar$.
    Now~$u^*(s,\eta;q)$ is continuously differentiable in~$s$ and~$\eta$, and \cref{lem:complement-br-payoff} implies
    \[ \parfrac{u^*(s,\eta;q)}{\eta}\bigg\vert_{\eta=\eta\IC(s;q)}<0 \]
    because~$0<\eta\IC(s;q)<\eta_\text{min}(s;q)$.
    So, by the implicit function theorem, the mapping~$s\mapsto\eta\IC(s;q)$ is differentiable and, thus, continuous.
    Its derivative~$\partial\eta\IC(s;q)/\partial s$ satisfies
    \[ \parfrac{u^*(s,\eta;q)}{s}\bigg\vert_{\eta=\eta\IC(s;q)}+\parfrac{u^*(s,\eta;q)}{\eta}\bigg\vert_{\eta=\eta\IC(s;q)}\parfrac{\eta\IC(s;q)}{s}=0, \]
    obtained by applying the chain rule to the identity~\eqref{eq:complement-ic-intensity-definition}.
    But \cref{lem:complement-br-payoff} implies
    \[ \parfrac{u^*(s,\eta;q)}{s}\bigg\vert_{\eta=\eta\IC(s;q)}>0
        \ \ \text{and}\ \ 
        \parfrac{u^*(s,\eta;q)}{\eta}\bigg\vert_{\eta=\eta\IC(s;q)}<0, \]
    from which it follows that~$\partial\eta\IC(s;q)/\partial s>0$.
    
    Finally, I determine the limits of~$\eta\IC(s;q)$ as~$s\to q\ybar$ and as~$s\to\infty$.

    Consider the lower limit as~$s\to q\ybar$.
    Now~$\eta\IC(s;q)$ is increasing in~$s$ and bounded below by zero, and so the limit~$\ell\ge0$ of~$\eta\IC(s;q)$ as~$s\to q\ybar$ from above exists.
    If~$\ell>0$, then the continuity of~$u^*$ implies~$u^*(q\ybar,\ell;q)=\pi q\ybar$.
    But for all~$\eta\in(0,\eta_\text{min}(s;q))$ we have
    \[ u^*(q\ybar,\eta;q)<\lim_{\eta\to0}u^*(q\ybar,\eta;q)=\pi q\ybar; \]
    a contradiction.
    Thus~$\ell=0$.

    Now consider the upper limit as~$s\to\infty$.
    For all~$\eta'\in(0,1)$, we have~$u^*(s,\eta';q)\to\infty$ as~$s\to\infty$ (by \cref{lem:complement-br-payoff}) and so~$\eta\IC(s;q)>\eta'$ eventually.
    So we must have~$\eta\IC(s;q)\to1$ as~$s\to\infty$.
\end{proof}
\fi

\begin{figure}
    \centering
    \includegraphics[width=0.55\linewidth]{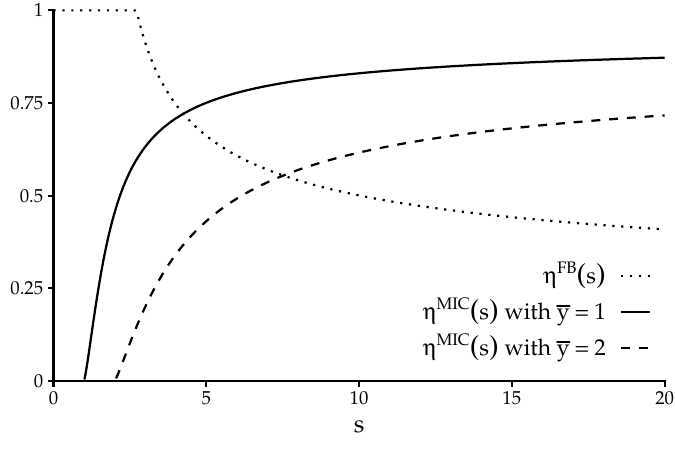}
    \caption{
    First-best and maximum IC effort intensities when~$\pi=1$ and~$\ybar\in\{1,2\}$.
    }
    \label{fig:ic-intensity}
\end{figure}
\cref{fig:ic-intensity} shows that the ``maximum incentive-compatible (IC) effort intensity''~$\eta\IC(s)$ is increasing in the skill level~$s$ and decreasing in the AI-generated output~$\ybar$.
Intuitively, if the student has more skill, then he is more willing to work---by \cref{lem:br-payoff}, his best-response payoff~$u^*(s,\eta)$ is increasing in~$s$---and the incentive constraint~\eqref{eq:ic-constraint} has less bite.
In contrast, if AI is more productive, then the student is more tempted to delegate and the incentive constraint has more bite.

\subsubsection{Optimal effort intensities}
\label{sec:sb-solution}

The student works on task~$t$ if and only if the incentive constraint~\eqref{eq:ic-constraint} holds.
But \cref{lem:ic-intensity} says the constraint holds precisely when the task's effort intensity~$\eta_t$ is at most~$\eta\IC(s_t)$.
So if~$\eta\IC(s_t)$ is larger than the first-best intensity~$\eta\FB(s_t)$ defined in \cref{sec:fb}, then the first-best course is incentive-compatible because setting~$\eta_t=\eta\FB(s_t)$ makes the student prefer to work.
In contrast, if~$\eta\IC(s_t)$ is \emph{smaller} than~$\eta\FB(s_t)$, then the first-best course is \emph{not} incentive-compatible because setting~$\eta_t=\eta\FB(s_t)$ makes the student prefer to delegate.
Then the teacher must choose~$\eta_t$ smaller than~$\eta\FB(s_t)$ to satisfy~\eqref{eq:ic-constraint}.
But \cref{lem:br-effort,lem:fb-intensity} imply the student's best-response effort~$x^*(s_t,\eta)$ is increasing in~$\eta$ over~$(0,\eta\FB(s_t))$.
So if~$\eta\IC(s_t)<\eta\FB(s_t)$, then the teacher maximizes~$x^*(s_t,\eta_t)$ subject to~\eqref{eq:ic-constraint} by setting~$\eta_t=\eta\IC(s_t)$.

Thus, for each task~$t\in\Tcal$, the teacher optimally sets the effort intensity~$\eta_t$ equal to the minimum of~$\eta\IC(s_t)$ and~$\eta\FB(s_t)$.
I characterize this minimum in \cref{lem:sb-intensity}.

\begin{lemma}[Second-best effort intensities]
    \label{lem:sb-intensity}
    Suppose~\eqref{eq:sb-restriction-1} holds and define
    \begin{equation}
        \label{eq:sb-intensity}
        \eta\SB(s)\equiv\min\left\{\eta\IC(s),\,\eta\FB(s)\right\}
    \end{equation}
    for all skill levels~$s>\ybar$.
    There is a unique skill level~$\shat\in\left(\max\left\{\ybar,\,\pi e\right\},\,\infty\right)$ satisfying~$\eta\IC(\shat)=\eta\FB(\shat)$.
    Moreover, the second-best effort intensity
    \[ \eta\SB(s)=\begin{cases}
            \eta\IC(s) & \text{if}\ s<\shat \\
            \eta\FB(s) & \text{if}\ s\ge\shat
        \end{cases} \]
    is increasing in~$s$ when~$s<\shat$ and decreasing in~$s$ when~$s>\shat$.
\end{lemma}
\ifbodyproofs\subsection{Proof of \texorpdfstring{\cref{lem:sb-intensity}}{Lemma~\ref{lem:sb-intensity}}}

\cref{lem:sb-intensity} is the special case of \cref{lem:complement-sb-intensity} with~$q=1$.
I prove \cref{lem:complement-sb-intensity} in \cref{appsec:complement-sb-intensity}.
\fi

\cref{fig:ic-intensity} shows that the ``second-best effort intensity''~$\eta\SB(s)=\min\{\eta\IC(s),\eta\FB(s)\}$ is increasing in the skill level~$s$ at low levels and decreasing in~$s$ at high levels.
The function~$\eta\SB:(\ybar,\infty)\to(0,1)$ attains its unique maximum at~$s=\shat$, where the maximum IC and first-best effort intensities are equal.
\cref{thm:sb-solution} uses the function to characterize second-best courses.

\begin{theorem}[Second-best courses]
    \label{thm:sb-solution}
    Suppose~\eqref{eq:sb-restriction-1} and~\eqref{eq:sb-restriction-2} hold, and~$(\eta_t)_{t\in\Tcal}$ solves the teacher's problem~\eqref{eq:problem}.
    Then~$\eta_t=\eta\SB(s_t)$ for each task~$t\in\Tcal$.
    Moreover, there exists~$t\SB\in\Tcal$ such that~$\eta_t$ is increasing in~$t$ over~$\{0,\ldots,t\SB\}$ and decreasing in~$t$ over~$\{t\SB+1,\ldots,T-1\}$.
\end{theorem}
\ifbodyproofs\subsection{Proof of \texorpdfstring{\cref{thm:sb-solution}}{Theorem~\ref{thm:sb-solution}}}

\cref{thm:sb-solution} is the special case of \cref{thm:complement-sb-solution} with~$q=1$.

\begin{theorem}[Second-best courses with complementary AI]
    \label{thm:complement-sb-solution}
    Suppose~\eqref{eq:complement-sb-restriction-1} and~\eqref{eq:complement-sb-restriction-2} hold, and~$(\eta_t)_{t\in\Tcal}$ solves~\eqref{eq:complement-problem}.
    Then~$\eta_t=\eta\SB(s_t;q)$ for each~$t\in\Tcal$.
    Moreover, there exists~$t\SB\in\Tcal$ such that~$\eta_t$ is increasing in~$t$ over~$\{0,\ldots,t\SB\}$ and decreasing in~$t$ over~$\{t\SB+1,\ldots,T-1\}$.
\end{theorem}

\begin{proof}
    By \cref{lem:complement-sb-restriction}, we have~$s_t>q\ybar$ for each task~$t\in\Tcal$.

    Fix~$t\in\Tcal$.
    \cref{lem:complement-br-effort} implies~$x^*(s_t,\eta;q)$ is increasing in~$\eta$ when~$\eta<\eta\FB(s_t;q)$ and decreasing in~$\eta$ when~$\eta>\eta\FB(s_t;q)$.
    But, by \cref{lem:complement-ic-intensity}, we have~$u^*(s_t,\eta;q)\ge\pi q\ybar$ if and only if~$\eta\le\eta\IC(s_t;q)$.
    It follows from~\eqref{eq:complement-task-solution} that~$\eta_t=\eta\SB(s_t;q)$.
    Then~$u\SB(s_t;q)\equiv u^*(s_t,\eta\SB(s_t;q);q)\ge\pi q\ybar$ since~$\eta_t=\eta\SB(s_t;q)$ satisfies~\eqref{eq:complement-ic-constraint}.

    Now define~$\shat(q)$ as in \cref{lem:complement-sb-intensity}.
    Then~$\eta\SB(s;q)$ is continuous, increasing in~$s$ over~$(\pi y\ybar,\shat(q))$, and decreasing in~$s$ over~$(\shat(q),\infty)$.
    So, since~$\Tcal$ is finite and~$s_{t+1}>s_t$ for each~$t\in\Tcal$, the function~$\psi:\Tcal\to(0,1]$ defined by~$\psi(t)\equiv\eta\SB(s_t;q)$ has either two maximizers (when~$s_{t-1}<\shat(q)<s_t$ and~$\eta\IC(s_{t-1};q)=\eta\FB(s_t;q)$ for some~$t>0$) or one.
    Denoting its smallest maximizer by~$t\SB$ yields the result.
\end{proof}
\fi

First-best courses comprise tasks with non-increasing effort intensities (see \cref{thm:fb-solution}).
In contrast, second-best courses comprise tasks with increasing-then-decreasing effort intensities.
This non-monotonicity stems from the student's option to delegate.
On early tasks, his skill level~$s_t$ is low, so his payoff~$u\FB(s_t)$ in the first-best course is also low (by \cref{prop:fb-br-s}) and delegation is relatively attractive.
To prevent delegation, the teacher must raise the student's payoff from working, which she achieves by distorting effort intensities downward.
As the student works on tasks and gains skill, his payoff~$u\FB(s_t)$ in the first-best course rises, and so the teacher does not need to distort the course's design as much to make the student prefer working over delegating.
Thus, on early tasks, the second-best course becomes gradually more effort-intensive.
Eventually, the student has built enough skill that his payoff~$u\FB(s_t)$ in the first-best course exceeds his payoff~$\ubar\equiv\pi\ybar$ from delegating.
Then the teacher proceeds with the first-best course, which becomes gradually less effort-intensive (by \cref{thm:fb-solution}).

\cref{thm:sb-solution} formalizes a suggestion made by \cite{Chirikov-etal-2026-Science} in their discussion of AI misuse and assessment reform.
The authors suggest ``requiring students to document their process, justify their choices, [and] demonstrate understanding [of] underlying reasoning'' \citep[p.~820]{Chirikov-etal-2026-Science}.
These redesign choices correspond to making tasks more skill-intensive: the student must be able to identify and explain appropriate methods for completing tasks.

\subsubsection{Best-responses}
\label{sec:sb-br}

Suppose~\eqref{eq:sb-restriction-1} holds, fix a skill level~$s>\ybar$, and let
\[ x\FB(s)\equiv x^*(s,\eta\FB(s))
    \ \ \text{and}\ \  
    x\SB(s)\equiv x^*(s,\eta\SB(s)) \]
denote the student's best-response efforts under first- and second-best course designs.
\cref{lem:fb-intensity} implies
\[ x\SB(s)\le\max_{\eta\in(0,1]}x^*(s,\eta)=x\FB(s), \]
with equality if and only if~$\eta\SB(s)=\eta\FB(s)$.
\begin{figure}
    \centering
    \includegraphics[width=0.55\linewidth]{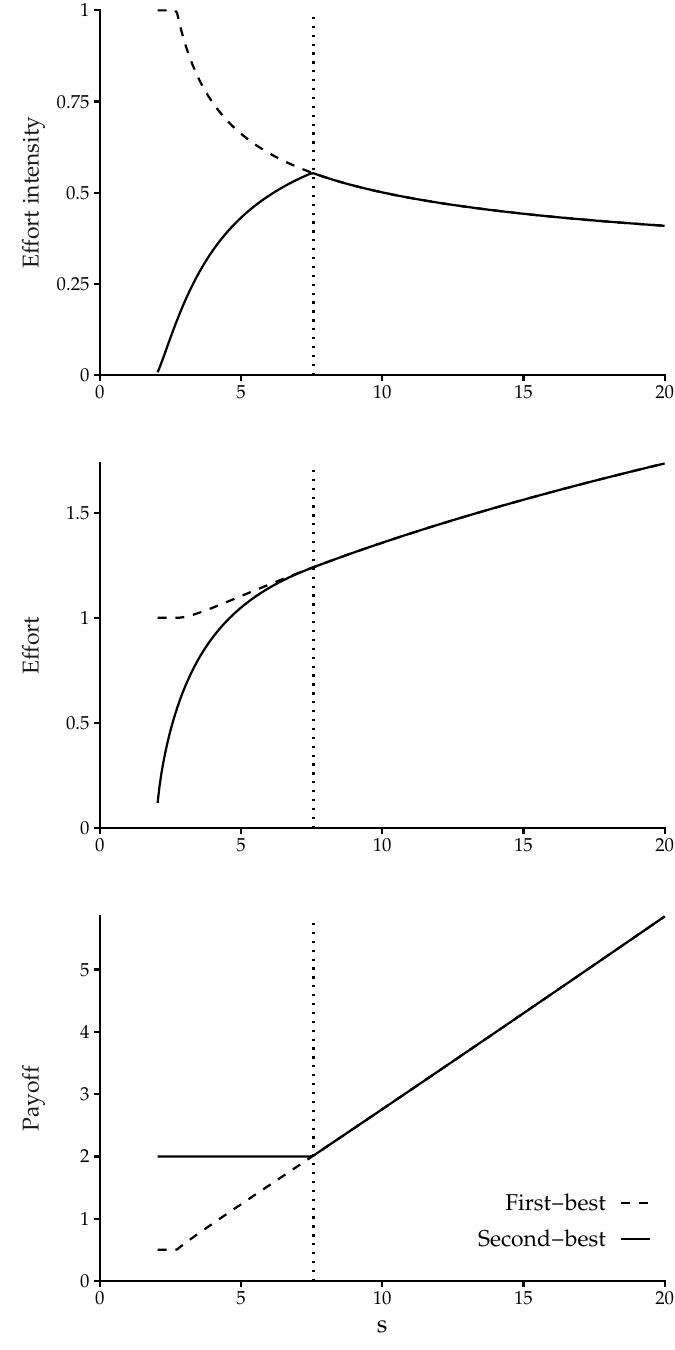}
    \caption{
    Optimal effort intensities, best-response efforts, and best-response payoffs when~$\pi=1$ and~$\ybar=2$.
    Dotted line indicates threshold~$\shat$ defined in \cref{lem:sb-intensity}.
    }
    \label{fig:sb-outcomes}
\end{figure}
\cref{fig:sb-outcomes} illustrates this bound, and shows~$x\SB(s)$ is increasing in the skill level~$s$ when the student has motivation parameter~$\pi=1$ and the AI generates output~$\ybar=2$.
\cref{prop:sb-effort} establishes that~$x\SB(s)$ is increasing in~$s$ for \emph{all} values of~$\pi$ and~$\ybar$, despite the second-best effort intensity~$\eta\SB(s)$ being non-monotone in~$s$.

\begin{proposition}[Second-best efforts]
    \label{prop:sb-effort}
    If~\eqref{eq:sb-restriction-1} holds, then~$x\SB(s)$ is increasing in~$s$ over~$(\ybar,\infty)$.
\end{proposition}
\ifbodyproofs\subsection{Proof of \texorpdfstring{\cref{prop:sb-effort}}{Proposition~\ref{prop:sb-effort}}}

With~$q\ge1$ arbitrary, the second-best effort is given by
\[ x\SB(s;q)\equiv x^*(s,\eta\SB(s;q);q). \]
\cref{prop:sb-effort} follows immediately from letting~$q=1$ and invoking part~(i) of \cref{lem:complement-sb-effort}.

\begin{lemma}[Second-best efforts with complementary AI]
    \label{lem:complement-sb-effort}
    Suppose~\eqref{eq:complement-sb-restriction-1} holds and let~$s>q\ybar$.
    Then the second-best effort~$x\SB(s;q)$ is
    (i)~increasing in the skill level~$s$,
    (ii)~increasing in the motivation parameter~$\pi$,
    and
    (iii)~decreasing in the AI-generated output~$\ybar$ when~$u\FB(s;q)<\pi q\ybar$ and constant in~$\ybar$ when~$u\FB(s;q)\ge\pi q\ybar$.
\end{lemma}
\begin{proof}
    For all skill levels~$s>q\ybar$, let
    \[ H(s;q)\equiv\left\{\eta\in(0,1]:u^*(s,\eta;q)\ge\pi q\ybar\right\} \]
    be the set of incentive-compatible effort intensities.
    \cref{lem:complement-ic-intensity} implies~$H(s;q)=(0,\eta\IC(s;q)]$, while \cref{lem:complement-br-effort,lem:complement-fb-intensity} imply~$x^*(s,\eta;q)$ is increasing in~$\eta$ on~$(0,\eta\FB(s;q)]$ and decreasing in~$\eta$ on~$[\eta\FB(s;q),1]$.
    So the maximum 
    \[ x\SB(s;q)\equiv\max\left\{x^*(s,\eta;q):\eta\in H(s;q)\right\} \]
    is attained uniquely at~$\min\{\eta\IC(s;q),\eta\FB(s;q)\}=\eta\SB(s;q)$.
    Note also~$\eta\SB(s;q)\le\eta\IC(s;q)<1$.
    
    I use this information to prove parts~(i)--(iii):
    \begin{enumerate}

        \item[(i)]
        Let~$s'>s>q\ybar$ and~$\eta=\eta\SB(s;q)\in(0,1)$.
        \cref{lem:complement-br-effort} implies~$x^*(s',\eta;q)>x^*(s,\eta;q)$.
        But \cref{lem:complement-br-payoff} implies~$u^*(s,\eta;q)\ge u^*(s',\eta;q)$, so~$H(s;q)\subseteq H(s';q)$ and hence~$\eta\in H(s';q)$.
        So
        \[ x\SB(s';q)\ge x^*(s',\eta;q)>x^*(s,\eta;q)=x\SB(s;q). \]

        \item[(ii)]
        By \cref{lem:complement-br-effort}, raising~$\pi$ also raises~$x^*(s,\eta;q)$ at the initial~$\eta=\eta\SB(s;q)$.
        But~\eqref{eq:complement-br-payoff-expanded} implies
        \[ \parfrac{}{\pi}\left[\frac{u^*(s,\eta;q)}{\pi}\right]=\frac{\eta u^*(s,\eta;q)}{\pi^2(2-\eta)}>0 \]
        for all~$s>0$ and~$\eta\in(0,1]$, and so raising~$\pi$ also expands the feasible set~$H(s;q)$.
        It follows (by an argument analogous that used for part~(i)) that~$x\SB(s;q)$ is increasing in~$\pi$.

        \item[(iii)]
        Suppose~$u\FB(s;q)<\pi q\ybar$.
        Then~$x\SB(s;q)=x^*(s,\eta\IC(s;q);q)$ and~$\eta\IC(s;q)<\eta\FB(s;q)$, and \cref{lem:complement-br-effort} implies~$x^*(s,\eta;q)$ is increasing in~$\eta$ in a neighborhood of~$\eta\IC(s;q)$.
        But~$\eta\IC(s;q)$ is defined to solve the implicit equation
        \[ u^*(s,\eta\IC(s;q);q)=\pi q\ybar, \]
        and, since~$\eta\IC(s;q)<s/\pi q^2$, \cref{lem:complement-br-payoff} implies~$u^*(s,\eta;q)$ is decreasing in~$\eta$ in a neighborhood of~$\eta\IC(s;q)$.
        So if~$\ybar$ rises, then~$\eta\IC(s;q)$ must fall and hence~$x\SB(s;q)$ must also fall.

        Now suppose~$u\FB(s;q)\ge\pi q\ybar$.
        Then~$x\SB(s;q)=x^*(s,\eta\FB(s;q);q)$, which is constant in~$\ybar$ because~$\eta\FB(s;q)$ is constant in~$\ybar$.
        \qedhere

    \end{enumerate}
\end{proof}
\fi

Now let
\[ u\FB(s)\equiv u^*(s,\eta\FB(s))
    \ \ \text{and}\ \ 
    u\SB(s)\equiv u^*(s,\eta\SB(s)) \]
denote the student's best-response payoffs under first- and second-best course designs.
\cref{lem:ic-intensity,lem:sb-intensity} imply
\[ u\SB(s)=\max\{\ubar,u\FB(s)\}\ge u\FB(s), \]
with equality if and only if~$\eta\SB(s)=\eta\FB(s)$.%
\footnote{
Since~$u\FB(s)$ is increasing in~$s$ when~$s>\pi e$ (by \cref{prop:fb-br-s}) and exceeds~$\ubar$ precisely when~$s$ exceeds the threshold~$\shat>\pi e$ defined in \cref{lem:sb-intensity}, the second-best payoff~$u\SB(s)$ is constant in~$s$ when~$s<\shat$ and increasing in~$s$ when~$s\ge\shat$.
}
Thus, at all skill levels, the student exerts weakly less effort but obtains weakly higher payoffs under a second-best design than under a first-best design.
That he exerts less effort implies he develops less skill, making the second-best design worse from the teacher's perspective even though it may be better from the student's perspective.%
\footnote{
While the student has higher payoffs for a given skill level under a second-best design than under a first-best design, his skill levels~$(s_t)_{t\in\Tcal}$ under a second-best design are lower because he does not exert as much effort on early tasks.
Consequently, his payoff~$u\SB(s_t)$ on later tasks~$t$ in the second-best course may be smaller than his payoff~$u\FB(s_t)$ on such tasks in the first-best course.
The student's overall preference between courses depends on whether the increase in his payoffs on early tasks (from~$u\FB(s_t)$ to~$\ubar$) dominates the decrease in his payoffs on later tasks.
}

\subsection{Second-best skill gains}
\label{sec:sb-gains}

Suppose~\eqref{eq:sb-restriction-1} and~\eqref{eq:sb-restriction-2} hold, and the teacher designs a course maximizing the student's overall skill gain
\[ s_T-s_0=\sum_{t\in\Tcal}\Delta(s_t,\eta_t). \]
By \cref{thm:sb-solution}, the teacher chooses~$\eta_t=\eta\SB(s_t)$ for each task~$t\in\Tcal$.
This choice satisfies the incentive constraint~\eqref{eq:ic-constraint}, so the student exerts effort~$x\SB(s_t)\equiv x^*(s_t,\eta\SB(s_t))$ and gains skill
\[ \Delta(s_t,\eta_t)=rx\SB(s_t) \]
by working on each task~$t\in\Tcal$.

\begin{lemma}[Skill increments]
    \label{lem:sb-increments}
    Suppose~\eqref{eq:sb-restriction-1} and~\eqref{eq:sb-restriction-2} hold, and~$(\eta_t)_{t\in\Tcal}$ solves the teacher's problem~\eqref{eq:problem}.
    For each task~$t\in\Tcal$, the skill increment~$\Delta(s_t,\eta_t)$ is
    (i)~increasing in the student's initial skill~$s_0$,
    (ii)~increasing in the motivation parameter~$\pi$, 
    (iii)~increasing in the skill accumulation rate~$r$, and
    (iv)~decreasing in the AI-generated output~$\ybar$.
\end{lemma}
\ifbodyproofs\subsection{Proof of \texorpdfstring{\cref{lem:sb-increments}}{Lemma~\ref{lem:sb-increments}}}

\cref{lem:sb-increments} is the special case of \cref{lem:complement-sb-increments} with~$q=1$.

\begin{lemma}[Skill increments with complementary AI]
    \label{lem:complement-sb-increments}
    Suppose~\eqref{eq:complement-sb-restriction-1} and~\eqref{eq:complement-sb-restriction-2} hold, and~$(\eta_t)_{t\in\Tcal}$ solves~\eqref{eq:complement-problem}.
    For each task~$t\in\Tcal$, the skill increment~$\Delta(s_t,\eta_t;q)\equiv s_{t+1}-s_t$ is
    (i)~increasing in the student's initial skill level~$s_0$,
    (ii)~increasing in the motivation parameter~$\pi$, 
    (iii)~increasing in the skill accumulation rate~$r$,
    and
    (iv)~decreasing in the AI-generated output~$\ybar$.
\end{lemma}

\begin{proof}
    \cref{thm:complement-sb-solution} implies~$\Delta(s_t,\eta_t;q)=rx\SB(s_t;q)$ for each~$t\in\Tcal$.
    Moreover, \cref{lem:complement-sb-restriction} implies~$s_0>q\ybar$.
    I use this information establish parts~(i)--(iv):
    \begin{itemize}
        
        \item[(i)]
        \cref{lem:complement-sb-effort} says~$x\SB(s;q)$ is increasing in~$s$ over~$(q\ybar,\infty)$.
        So if~$s_0$ rises, then~$x\SB(s_0;q)$ rises, which makes~$s_1=s_0+rx\SB(s_0;q)$ rise, which makes~$x\SB(s_1;q)$ rise, which makes~$s_2=s_1+rx\SB(s_1;q)$ rise, and so on.
        Thus, if~$s_0$ rises, then~$x\SB(s_t;q)$ and~$\Delta(s_t,\eta_t;q)$ rise for each~$t\in\Tcal$.

        \item[(ii)]
        \cref{lem:complement-sb-effort} says~$x\SB(s;q)$ is increasing in~$s$ over~$(q\ybar,\infty)$ and is increasing in~$\pi$.
        So if~$\pi$ rises, then~$x\SB(s_0;q)$ rises, which makes~$s_1=s_0+rx\SB(s_0;q)$ rise, which together with~$\pi$ rising makes~$x\SB(s_1;q)$ rise, which makes~$s_2=s_1+rx\SB(s_1;q)$ rise, and so on.
        Thus, if~$\pi$ rises, then~$x\SB(s_t;q)$ and~$\Delta(s_t,\eta_t;q)$ rise for each~$t\in\Tcal$.

        \item[(iii)]
        \cref{lem:complement-sb-effort} says~$x\SB(s;q)$ is increasing in~$s$ over~$(q\ybar,\infty)$, and \cref{lem:complement-br-effort} implies~$x\SB(s_0;q)>0$.
        So if~$r$ rises, then~$s_1=s_0+rx\SB(s_0;q)$ rises, which makes~$x\SB(s_1;q)$ rise, which together with~$r$ rising makes~$s_2=s_1+rx\SB(s_1;q)$ rise, and so on.
        Thus, if~$r$ rises, then~$\Delta(s_t,\eta_t)=s_{t+1}-s_t$ rises for each~$t\in\Tcal$.

        \item[(iv)]
        Since~\eqref{eq:complement-sb-restriction-1} holds, we have~$u\FB(s_0;q)<\pi q\ybar$.
        \cref{lem:complement-sb-effort} then implies~$x\SB(s_0;q)$ is decreasing in~$\ybar$.
        \cref{lem:complement-sb-effort} also implies~$x\SB(s;q)$ is increasing in~$s$ over~$(q\ybar,\infty)$, and is non-increasing in~$\ybar$ for all~$s>q\ybar$.
        So if~$\ybar$ rises, then~$x\SB(s_0;q)$ falls, which makes~$s_1=s_0+rx\SB(s_0;q)$ fall, which together with~$\ybar$ rising makes~$x\SB(s_1;q)$ fall, which makes~$s_2=s_1+rx\SB(s_1;q)$ fall, and so on.
        Thus, if~$\ybar$ rises, then~$x\SB(s_t;q)$ and~$\Delta(s_t,\eta_t)$ fall for each~$t\in\Tcal$.
        \qedhere

    \end{itemize}
\end{proof}
\fi

The following result follows immediately from \cref{lem:sb-increments}.

\begin{theorem}[Overall skill gain]
    \label{thm:sb-overall-gain}
    Suppose~\eqref{eq:sb-restriction-1} and~\eqref{eq:sb-restriction-2} hold, and~$(\eta_t)_{t\in\Tcal}$ solves the teacher's problem~\eqref{eq:problem}.
    Then the student's overall skill gain~$(s_T-s_0)$ is
    (i)~increasing in his initial skill~$s_0$,
    (ii)~increasing in the motivation parameter~$\pi$,
    (iii)~increasing in the skill accumulation rate~$r$, and
    (iv)~decreasing in the AI-generated output~$\ybar$.
\end{theorem}
\ifbodyproofs\subsection{Proof of \texorpdfstring{\cref{thm:sb-overall-gain}}{Theorem~\ref{thm:sb-overall-gain}}}

\cref{thm:sb-overall-gain} is the special case of \cref{thm:complement-sb-overall-gain} with~$q=1$.

\begin{theorem}[Overall skill gain with complementary AI]
    \label{thm:complement-sb-overall-gain}
    Suppose~\eqref{eq:complement-sb-restriction-1} and~\eqref{eq:complement-sb-restriction-2} hold, and~$(\eta_t)_{t\in\Tcal}$ solves~\eqref{eq:complement-problem}.
    Then~$(s_T-s_0)$ is 
    (i)~increasing in the student's initial skill level~$s_0$,
    (ii)~increasing in the motivation parameter~$\pi$, 
    (iii)~increasing in the skill accumulation rate~$r$,
    and
    (iv)~decreasing in the AI-generated output~$\ybar$.
\end{theorem}
\begin{proof}
    The result follows immediately from \cref{lem:complement-sb-increments} and the fact that
    \[ s_T-s_0=\sum_{t\in\Tcal}(s_{t+1}-s_t).\qedhere \]
\end{proof}
\fi

On each task, the student gains more skill when:
(i)~he has more skill at the outset (i.e., $s_0$ is higher), since skill complements effort and grows across tasks;
(ii)~he is more motivated (i.e., $\pi$ is higher), since he is more willing to exert effort;
(iii)~he gains more skill per unit of effort exerted (i.e., $r$ is higher).
The student gains \emph{less} skill when AI is more productive (i.e., $\ybar$ is higher), since the teacher needs to distort the course design further from the first-best to induce work.

The effects of~$s_0$, $\pi$, $r$, and~$\ybar$ on the skill gained from each task also compound across tasks.
This is because~$s_{t+1}=s_t+rx\SB(s_t)$ is increasing in~$x\SB(s_t)$, so any parameter change that raises~$x\SB(s_t)$ also raises~$s_{t+1}$, which, by \cref{prop:sb-effort}, raises~$x\SB(s_{t+1})$ too---a positive feedback loop.
For example, a one-off boost to the student's initial skill raises the effort he exerts on the initial task \emph{and} the effort he exerts on subsequent tasks, since he enters each task more skilled than he would have without the boost.
Thus, as emphasized by \cite{Heckman-2000-ResEcon} and~\cite{Cunha-Heckman-2007-AER}, ``skill begets skill'' due to dynamic complementarities.

\section{Extension: Complementary AI}
\label{sec:complement}

So far I have assumed AI substitutes for effort.
Now I consider the possibility that AI complements effort.
For example, rather than prompt AI to ``do my homework,'' the student could use skill to break his homework into subtasks and use AI to accelerate the effortful completion of those subtasks.%
\footnote{
For example, suppose the task is to prove a mathematical theorem.
The student could pass the entire theorem to a Large Language Model (LLM) before submitting a simple user prompt: ``prove the theorem.''
Alternatively, he could devise a proof strategy that involves proving a series of lemmas, then using the LLM to prove each of those lemmas.
}
So the student can use AI to avoid or enhance his effort, consistent with survey evidence on students' reported uses of AI \citep{Contractor-Reyes-2026-}.

To capture the possibility that AI can complement or substitute for effort, I introduce an ``AI quality'' parameter~$q\ge1$ that multiplies the effort he exerts when he works \emph{and} the output~$\ybar$ he produces when he delegates to AI.
Then the student has working payoff
\[ u(x,s,\eta;q)\equiv\pi(qx)^\eta s^{1-\eta}-\frac{x^2}{2} \]
and delegation payoff
\[ q\ubar=\pi q\ybar. \]
Letting~$q=1$ yields the payoffs defined in \cref{sec:model}.%
\footnote{
All of my results in \cref{sec:br,sec:solution} extend to the case with arbitrary~$q\ge1$ without qualitative changes.
}
Increasing~$q$ makes effort more productive \emph{and} raises the delegation payoff.
If the student works on a task, then his best-response effort
\[ x^*(s,\eta;q)\in\argmax_{x\ge0}\,u(x,s,\eta;q) \]
and payoff
\[ u^*(s,\eta;q)\equiv\max_{x\ge0}\,u(x,s,\eta;q) \]
both depend on~$q$, and so the first- and second-best course designs will also depend on~$q$.%
\footnote{
I write~$q$ (a primitive) after a semicolon to distinguish it from~$s$ (a state variable) and~$\eta$ (a control variable).
I include~$q$ as an argument but not the other primitives~$(T,\pi,r,\ybar)$ because my focus is on how outcomes depend on~$q$.
}
I study this dependence in \cref{sec:complement-fb,sec:complement-sb}, following the steps in my analysis of the case with~$q=1$ (see \cref{sec:fb,sec:sb}).

\subsection{First-best courses}
\label{sec:complement-fb}

Consider the first-best case in which~$\ybar=0$ and the student always prefers to work.
Letting~$z\equiv qx$ denote his AI-enhanced effort gives
\begin{equation}
    \label{eq:complement-enhanced-payoff}
    u(x,s,\eta;q)=\frac{1}{q^2}\left(\pi q^2 z^\eta s^{1-\eta}-\frac{z^2}{2}\right)
\end{equation}
for all skill levels~$s>0$ and effort intensities~$\eta\in(0,1]$.
The bracketed term on the right-hand side of~\eqref{eq:complement-enhanced-payoff} equals the payoff from the case with~$q=1$, with~$x$ replaced by~$z$ and~$\pi$ replaced by~$\pi q^2$.
So, by \cref{lem:br-effort}, the student maximizes the bracketed term by choosing
\[ z=s\left(\frac{(\pi q^2)\eta}{s}\right)^{1/(2-\eta)}, \]
which corresponds to exerting best-response effort
\begin{equation}
    \label{eq:complement-br-effort}
    x^*(s,\eta;q)=s\left(\frac{\pi\eta q^\eta}{s}\right)^{1/(2-\eta)}.
\end{equation}
The effort~$x^*(s,\eta;q)$ is increasing in the AI quality parameter~$q$.
Intuitively, raising~$q$ makes effort more productive, so the student is willing to exert more.

Replacing~$\pi$ by~$\pi q^2$ also yields a more general version of the first-best effort intensity
\[ \eta\FB(s;q)\in\argmax_{\eta\in(0,1]}\,x^*(s,\eta;q) \]
defined in \cref{lem:fb-intensity}:
\begin{enumerate}

    \item[(i)]
    If~$s\le\pi eq^2$, then~$\eta\FB(s;q)=1$.
    
    \item[(ii)]
    If~$s>\pi eq^2$, then~$\eta\FB(s;q)\in(0,1)$ is the unique solution to
    \begin{equation}
        \label{eq:complement-fb-intensity-foc}
        \eta\FB(s;q)\exp\left(\frac{2}{\eta\FB(s;q)}-1\right)=\frac{s}{\pi q^2},
    \end{equation}
    is decreasing in~$s$, and is increasing in~$\pi$ and~$q$.%
    \footnote{
    See \cref{appsec:fb-intensity} for the derivation of~$\eta\FB(s;q)$ and its properties.
    }

\end{enumerate}
Since the teacher wants to maximize~$x^*(s_t,\eta_t;q)$ for each task~$t\in\Tcal$ (see \cref{sec:solution} and \cref{appsec:fb-solution}), she optimally chooses~$\eta_t=\eta\FB(s_t;q)$, which is non-decreasing in~$q$ for all~$s_t>0$ and increasing in~$q$ when~$s_t>\pi eq^2$.
Thus, she optimally assigns more effort-intensive tasks when AI has higher quality.
She does so to capitalize on the student's higher willingness to exert effort.

Now fix a skill level~$s>0$, and define the first-best effort and payoff
\[ x\FB(s;q)\equiv x^*(s,\eta\FB(s;q);q)
    \ \ \text{and}\ \ 
    u\FB(s;q)\equiv u^*(s,\eta\FB(s;q);q) \]
analogously to \cref{sec:fb}.
Improvements in AI quality raise~$x\FB(s;q)$ but may lower~$u\FB(s;q)$:

\begin{proposition}[First-best outcomes and AI quality]
    \label{prop:complement-fb-br-q}
    Fix a skill level~$s>0$.
    The first-best effort~$x\FB(s;q)$ is increasing in the AI quality parameter~$q$.
    The first-best payoff~$u\FB(s;q)$ is decreasing in~$q$ when~$q<\sqrt{s/\pi e}$ and increasing in~$q$ when~$q\ge\sqrt{s/\pi e}$.
\end{proposition}
\ifbodyproofs\subsection{Proof of \texorpdfstring{\cref{prop:complement-fb-br-q}}{Proposition~\ref{prop:complement-fb-br-q}}}

\begin{proof}[\nopunct\unskip]
    Suppose~$s\le\pi eq^2$.
    Then~$\eta\FB(s;q)=1$ by \cref{lem:complement-fb-intensity}, and so~$x\FB(s;q)=x^*(s,1;q)=\pi q$ and~$u\FB(s;q)=u^*(s,1;q)=\pi^2q^2/2$ are increasing in~$q$.

    Now suppose~$s>\pi eq^2$.
    Then~$\eta\FB(s;q)$ satisfies~\eqref{eq:complement-fb-intensity-foc}, which can be rewritten as
    \begin{equation}
        \label{eq:complement-fb-intensity-foc-alt}
        \frac{1}{q^2}\exp\left(-\left(\frac{2}{\eta\FB(s;q)}-1\right)\right)=\frac{\pi\eta\FB(s;q)}{s}.
    \end{equation}
    Then \cref{lem:complement-br-effort} implies
    \begin{align*}
        x\FB(s;q)
        &= s\left(\frac{\pi\eta\FB(s;q)q^{\eta\FB(s;q)}}{s}\right)^{1/(2-\eta\FB(s;q))} \\
        &= s\left(\frac{1}{q^2}\exp\left(-\left(\frac{2}{\eta\FB(s;q)}-1\right)\right)q^{\eta\FB(s;q)}\right)^{1/(2-\eta\FB(s;q))} \\
        &= \frac{s}{q}\exp\left(-\frac{1}{\eta\FB(s;q)}\right),
    \end{align*}
    which is increasing in~$\eta\FB(s;q)$.
    But \cref{lem:complement-fb-intensity} implies~$\eta\FB(s;q)$ is increasing in~$q$, and it follows that~$x\FB(s;q)$ is increasing in~$q$.
    Moreover, \cref{lem:complement-br-payoff} implies
    \begin{align*}
        u\FB(s;q)
        &= \frac{2-\eta\FB(s;q)}{2\eta\FB(s;q)}\left(x\FB(s;q)\right)^2 \\
        &= \frac{\left(2-\eta\FB(s;q)\right)s^2}{2\eta\FB(s;q)eq^2}\exp\left(-\left(\frac{2}{\eta\FB(s;q)}-1\right)\right) \\
        &\overset{\star}{=} \frac{\pi s}{2e}\left(2-\eta\FB(s;q)\right),
    \end{align*}
    where~$\star$ uses the substitution~\eqref{eq:complement-fb-intensity-foc-alt}.
    So~$u\FB(s;q)$ is decreasing in~$\eta\FB(s;q)$, which is increasing in~$q$.
    Thus~$u\FB(s;q)$ is decreasing in~$q$.%
    \footnote{
    Recall from \cref{lem:complement-fb-intensity} that~$\eta\FB(s;q)\to1$ as~$s\to\pi eq^2$ from above.
    So
    \[ \lim_{s\to\pi eq^2}\frac{\pi s}{2e}\left(2-\eta\FB(s;q)\right)=\frac{\pi(\pi eq^2)}{2e}=\frac{\pi^2q^2}{2}=u\FB(s;q)\bigg\vert_{s\le\pi eq^2}, \]
    confirming the continuity of~$u\FB(s;q)$ in~$s$ as claimed in \cref{prop:complement-fb-br-s}.
    }
\end{proof}
\fi

The intuition for \cref{prop:complement-fb-br-q} is as follows.
The teacher chooses~$\eta\FB(s_t;q)$ to maximize the student's effort, rather than his payoff, and the effort-maximizing intensity is larger than the payoff-maximizing intensity because the former does not account for the cost of exerting effort.
As AI quality improves, the teacher optimally makes tasks more effort-intensive to keep extracting maximal effort from the now-more-productive student.
This pulls~$\eta\FB(s_t;q)$ further from the payoff-maximizing intensity, and so the student's payoff \emph{falls} even as his effort and skill gain rise.
Only when~$\eta\FB(s_t;q)$ reaches its ceiling of one---which, by the definition of~$\eta\FB(s;q)$, happens precisely when~$q\ge\sqrt{s_t/\pi e}$---do further quality improvements lead to higher payoffs with no change in effort intensities.

\subsection{Second-best courses}
\label{sec:complement-sb}

Now consider the second-best case in which~$\ybar>0$ and the student may prefer to delegate.
Given a skill level~$s>0$ and effort intensity~$\eta\in(0,1]$, the student's best-response effort~$x^*(s,\eta;q)$ yields payoff%
\footnote{
See \cref{appsec:br-payoff} for the derivation of~\eqref{eq:complement-br-payoff}.
}
\begin{equation}
    \label{eq:complement-br-payoff}
    u^*(s,\eta;q)=\frac{1}{\eta}\left(1-\frac{\eta}{2}\right)\left(x^*(s,\eta;q)\right)^2.
\end{equation}
He works on task~$t\in\Tcal$ if and only if his best-response payoff from working is at least his payoff from delegating:
\begin{equation}
    \label{eq:complement-ic-constraint}
    u^*(s_t,\eta_t;q)\ge \pi q\ybar.
    \tag{\ref{eq:ic-constraint}$;q$}
\end{equation}

\subsubsection{Parameter restrictions}

As in \cref{sec:sb}, I assume the first-best course is not incentive-compatible:
\begin{equation}
    \label{eq:complement-sb-restriction-1}
    u\FB(s_0;q)<\pi q\ybar.
    \tag{\ref{eq:sb-restriction-1}$;q$}
\end{equation}
I also assume it is possible to design an incentive-compatible course: 
\begin{equation}
    \label{eq:complement-sb-restriction-2}
    \sup_{\eta\in(0,1]}u^*(s_0,\eta;q)>\pi q\ybar.
    \tag{ND$;q$}
\end{equation}
Conditions~\eqref{eq:complement-sb-restriction-1} and~\eqref{eq:complement-sb-restriction-2} generalize conditions~\eqref{eq:sb-restriction-1} and~\eqref{eq:sb-restriction-2} from \cref{sec:sb}, and imply the bounds
\begin{equation}
    \label{eq:complement-sb-restriction}
    \frac{\pi q^2}{2}<q\ybar<s_0
\end{equation}
on the AI-generated output~$q\ybar$.%
\footnote{
See \cref{appsec:sb-restriction} for the derivation of~\eqref{eq:complement-sb-restriction}.
}
The intuition for these bounds extends from \cref{sec:sb}: the AI out-produces the student (net of costs) on fully effort-intensive tasks (on which his cost-adjusted output is~$u^*(s,1;q)/\pi=\pi q^2/2$ regardless of his skill level~$s$), but the student out-produces the AI on fully skill-intensive tasks (on which he produces~$s$).

\subsubsection{Incentive compatibility}

Suppose~\eqref{eq:complement-sb-restriction-1} holds and fix a skill level~$s>q\ybar$.
As in \cref{sec:sb}, there is a unique effort intensity
\[ \eta\IC(s;q)\in\left(0,\min\left\{\frac{s}{\pi q^2},1\right\}\right) \]
at which the student is indifferent between working and delegating:
\begin{equation}
    \label{eq:complement-ic-intensity-definition}
    u^*(s,\eta\IC(s;q);q)=\pi q\ybar.
\end{equation}
Moreover, the incentive constraint~\eqref{eq:complement-ic-constraint} holds if and only if~$\eta_t\le\eta\IC(s_t;q)$.%
\footnote{
See \cref{appsec:ic-intensity} for the proof that~$\eta\IC(s;q)$ exists.
}

\begin{figure}
    \centering
    \includegraphics[width=0.55\linewidth]{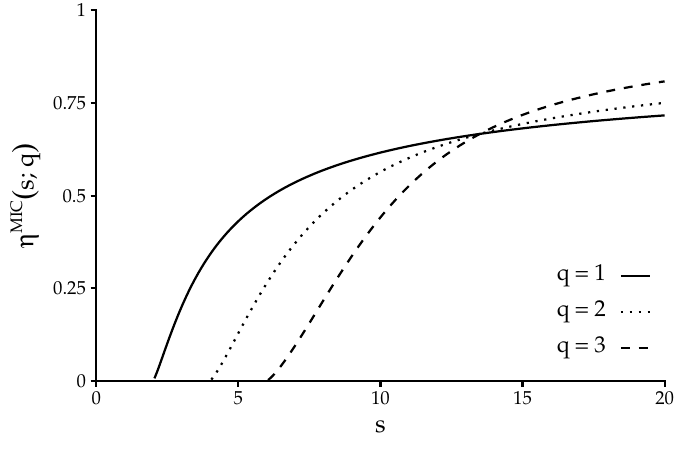}
    \caption{
    Maximum IC effort intensities~$\eta\IC(s;q)$ when~$\pi=1$ and~$\ybar=2$.
    }
    \label{fig:complement-ic-intensity}
\end{figure}
\cref{fig:complement-ic-intensity} shows that whether~$\eta\IC(s;q)$ is increasing in the AI quality parameter~$q$ depends on whether~$\eta\IC(s;q)>2/3$.
\cref{lem:complement-ic-intensity-q} characterizes this dependence in terms of the skill level~$s$.

\begin{lemma}[Incentive compatibility and AI quality]
    \label{lem:complement-ic-intensity-q}
    Suppose~\eqref{eq:complement-sb-restriction-1} holds, fix a skill level~$s>q\ybar$, and define
    \begin{equation}
        \label{eq:complement-ic-intensity-q-threshold}
        s^\dag\equiv\frac{27\,\ybar^2}{8\,\pi}.
    \end{equation}
    Then~$\eta\IC(s;q)$ is decreasing in the AI quality parameter~$q$ when~$s<s^\dag$, constant in~$q$ when~$s=s^\dag$, and increasing in~$q$ when~$s>s^\dag$.
    Moreover, we have~$s<s^\dag$ if and only if~$\eta\IC(s;q)<2/3$.
\end{lemma}
\ifbodyproofs\subsection{Proof of \texorpdfstring{\cref{lem:complement-ic-intensity-q}}{Lemma~\ref{lem:complement-ic-intensity-q}}}

\begin{proof}[\nopunct\unskip]
    Differentiating~\eqref{eq:complement-ic-intensity-definition} implicitly with respect to~$q$ gives
    \[ \parfrac{u^*(s,\eta;q)}{\eta}\bigg\vert_{\eta=\eta\IC(s;q)}\parfrac{\eta\IC(s;q)}{q}+\parfrac{u^*(s,\eta;q)}{q}\bigg\vert_{\eta=\eta\IC(s;q)}=\pi\ybar. \]
    Since~$u^*(s,\eta;q)$ is decreasing in~$\eta$ at~$\eta=\eta\IC(s;q)$, we have~$\partial\eta\IC(s;q)/\partial q>0$ if and only if
    \begin{equation}
        \label{eq:complement-ic-intensity-q-proof-condition}
        \parfrac{u^*(s,\eta;q)}{q}\bigg\vert_{\eta=\eta\IC(s;q)}>\pi\ybar.
    \end{equation}
    But
    \[ \parfrac{u^*(s,\eta;q)}{q}=\frac{2\eta}{2-\eta}\cdot\frac{u^*(s,\eta;q)}{q}, \]
    and therefore
    \[ \parfrac{u^*(s,\eta;q)}{q}\bigg\vert_{\eta=\eta\IC(s;q)}=\frac{2\pi\eta\IC(s;q)\ybar}{2-\eta\IC(s;q)} \]
    by~\eqref{eq:complement-ic-intensity-definition}.
    So~\eqref{eq:complement-ic-intensity-q-proof-condition} holds if and only if
    \[ \frac{2\eta\IC(s;q)}{2-\eta\IC(s;q)}>1, \]
    which, since~$0<\eta\IC(s;q)<1$, holds precisely when~$\eta\IC(s;q)>2/3$.

    Now by \cref{lem:complement-ic-intensity}, the intensity~$\eta\IC(s;q)$ rises continuously from zero to one as the skill level~$s$ rises from~$q\ybar$ to~$\infty$.
    So, by the intermediate value theorem, there is a unique~$s^\dag>q\ybar$ such that~$\eta\IC(s^\dag;q)=2/3$ and
    \[ \parfrac{\eta\IC(s;q)}{q}\gtreqless0
        \ \ \text{as}\ \ 
        s\gtreqless s^\dag. \]
    The skill level~$s^\dag$ satisfies
    \begin{align*}
        \pi q\ybar
        &= u^*\left(s^\dag,\frac{2}{3};q\right) \\
        &= \frac{1}{2/3}\left(1-\frac{2/3}{2}\right)\left(x^*\left(s^\dag,\frac{2}{3};q\right)\right)^2 \\
        &= \left(s^\dag\left(\frac{\pi(2/3)q^{2/3}}{s^\dag}\right)^{1/(2-2/3)}\right)^2 \\
        &= \left(\frac{2\pi}{3}\right)^{3/2}q(s^\dag)^{1/2},
    \end{align*}
    which can be rearranged for~\eqref{eq:complement-ic-intensity-q-threshold}.
    (\cref{lem:complement-sb-restriction} implies~$\ybar>\pi q/2$, and so~$s^\dag>27q\ybar/16>q\ybar$ as required.)
\end{proof}
\fi

To see why the effect of raising~$q$ on~$\eta\IC(s;q)$ depends on whether~$\eta\IC(s;q)>2/3$, fix an effort intensity~$\eta\in(0,1]$, and decompose the student's best-response payoff~\eqref{eq:complement-br-payoff} into the product
\[ u^*(s,\eta;q)=q^{2\eta/(2-\eta)}\,u^*(s,\eta;1) \]
of terms that do and do not depend on~$q$.
Thus, the AI quality parameter~$q$ enters the best-response payoff~$u^*(s,\eta;q)$ with exponent~$2\eta/(2-\eta)$ and the delegation payoff~$\pi q\ybar$ with exponent~$1$.
So if~$q$ rises, then~$u^*(s,\eta;q)$ rises faster than~$\pi q\ybar$ precisely when~$2\eta/(2-\eta)>1$, which holds if and only if~$\eta>2/3$.
But~$u^*(s,\eta;q)$ is decreasing in~$\eta$ in a neighborhood of~$\eta\IC(s;q)$.%
\footnote{
\cref{lem:complement-br-payoff,lem:complement-ic-intensity} imply~$\eta\IC(s;q)$ belongs to an interval over which~$u^*(s,\eta;q)$ is decreasing in~$\eta$.
}
So if~$u^*(s,\eta;q)$ rises faster than~$\pi q\ybar$ for~$\eta$ near~$\eta\IC(s;q)$, then~$\eta\IC(s;q)$ must rise to preserve the identity~\eqref{eq:complement-ic-intensity-definition}.
Thus, the effect of increasing~$q$ on~$\eta\IC(s;q)$ depends on whether~$\eta\IC(s;q)>2/3$.
The threshold~$s^\dag$ is precisely the skill level at which~$\eta\IC(s^\dag;q)=2/3$.

\subsubsection{Optimal effort intensities}
\label{sec:complement-solution}

Suppose~\eqref{eq:complement-sb-restriction-1} and~\eqref{eq:complement-sb-restriction-2} hold so that~$s_t>q\ybar$ for each~$t\in\Tcal$.
As in \cref{sec:sb}, the teacher optimally sets each effort intensity~$\eta_t\in(0,1]$ equal to the minimum of the maximum IC intensity~$\eta\IC(s_t;q)$ and first-best intensity~$\eta\FB(s_t;q)$.
I characterize this minimum in \cref{lem:complement-sb-intensity}, which generalizes \cref{lem:sb-intensity} to cases with~$q>1$.

\begin{lemma}[Second-best effort intensities with complementary AI]
    \label{lem:complement-sb-intensity}
    Suppose~\eqref{eq:complement-sb-restriction-1} holds and define
    \begin{equation}
        \label{eq:complement-sb-intensity}
        \eta\SB(s;q)\equiv\min\left\{\eta\IC(s;q),\,\eta\FB(s;q)\right\}
    \end{equation}
    for all skill levels~$s>q\ybar$.
    There is a unique level
    \[ \shat(q)\in\left(\max\left\{q\ybar,\pi eq^2\right\},\,\infty\right) \]
    satisfying~$\eta\IC(\shat(q);q)=\eta\FB(\shat(q);q)$.
    Moreover,
    \begin{equation}
        \label{eq:complement-sb-intensity-piecewise}
        \eta\SB(s;q)=\begin{cases}
            \eta\IC(s;q) & \text{if}\ s<\shat(q) \\
            \eta\FB(s;q) & \text{if}\ s\ge\shat(q)
        \end{cases}
    \end{equation}
    is increasing in~$s$ when~$s<\shat(q)$ and decreasing in~$s$ when~$s>\shat(q)$.
    Furthermore, the threshold~$\shat(q)$ is increasing in the AI quality parameter~$q$.
\end{lemma}
\ifbodyproofs\subsection{Proof of \texorpdfstring{\cref{lem:complement-sb-intensity}}{Lemma~\ref{lem:complement-sb-intensity}}}
\label{appsec:complement-sb-intensity}

\begin{proof}[\nopunct\unskip]
    By \cref{lem:complement-ic-intensity}, the maximum IC effort intensity~$\eta\IC(s;q)\in(0,1)$ rises continuously from zero to one as~$s$ rises from~$q\ybar$ to~$\infty$.
    In contrast, by \cref{lem:complement-fb-intensity}, the first-best effort intensity~$\eta\FB(s;q)$ equals one when~$s\le\pi eq^2$, and falls continuously from one to zero as~$s$ rises from~$\pi eq^2$ to~$\infty$.
    So, by the intermediate value theorem, there is a unique
    \[ \shat\in\left(\max\left\{q\ybar,\,\pi e q^2\right\},\,\infty\right) \]
    such that~$\eta\IC(\shat;q)=\eta\FB(\shat;q)$, and~\eqref{eq:complement-sb-intensity-piecewise} holds for all~$s>q\ybar$.
    It follows that~$\eta\SB(s;q)$ is increasing in~$s$ when~$s<\shat$ and decreasing in~$s$ when~$s>\shat$.

    It remains to show~$\shat$ is increasing in~$q$.

    Define~$\etahat\equiv\eta\IC(\shat;q)=\eta\FB(\shat;q)$ for convenience.
    Then~\eqref{eq:complement-fb-intensity-foc} implies
    \[ \shat=\pi\etahat q^2\exp\left(\frac{2}{\etahat}-1\right), \]
    and so~\eqref{eq:complement-br-payoff-expanded} and the definition of~$\eta\IC$ imply
    \begin{align*}
        \pi q\ybar
        &= u^*(\shat,\etahat;q) \\
        &= \left(1-\frac{\etahat}{2}\right)\pi^{2/(2-\etahat)}\etahat^{\etahat/(2-\etahat)}q^{2\etahat/(2-\etahat)}\left(\pi\etahat q^2\exp\left(\frac{2}{\etahat}-1\right)\right)^{2(1-\etahat)/(2-\etahat)} \\
        &= \frac{\pi^2\etahat (2-\etahat)q^2}{2}\exp\left(\frac{2(1-\etahat)}{\etahat}\right).
    \end{align*}
    Defining
    \[ Q(\eta)\equiv\frac{2\ybar}{\pi\eta(2-\eta)}\exp\left(-\frac{2(1-\eta)}{\eta}\right) \]
    then yields an implicit equation~$Q(\etahat)=q$ defining~$\etahat$ in terms of~$q$.
    But the function~$Q:(0,1]\to\R$ is increasing on its domain because~$Q(\eta)>0$ and
    \begin{align*}
        \frac{Q'(\eta)}{Q(\eta)}
        &= \parfrac{\log(Q(\eta))}{\eta} \\
        &= \parfrac{}{\eta}\left[\log\left(\frac{2\ybar}{\pi}\right)-\log(\eta(2-\eta))-\frac{2(1-\eta)}{\eta}\right] \\
        &=\frac{2\left((\eta-1)^2+1\right)}{\eta^2(2-\eta)} \\
        &> 0
    \end{align*}
    for all~$\eta\in(0,1]$.
    So if~$q$ rises, then~$\etahat$ must also rise to maintain~$Q(\etahat)=q$.
\end{proof}
\fi

Thus, as in the case with~$q=1$ (see \cref{sec:sb}), the second-best effort intensity~$\eta\SB(s;q)$ is increasing-then-decreasing in the skill level~$s$:
\[ \parfrac{\eta\SB(s;q)}{s}>0
    \ \ \text{if and only if}\ \ 
    s<\shat(q). \]
The threshold~$\shat(q)$ at which~$\partial\eta\SB(s;q)/\partial s$ changes sign is increasing in the AI quality parameter~$q$.
This is because~$\shat(q)$ belongs to the set~$(\pi eq^2,\infty)$ of skill levels~$s$ at which increasing~$q$ lowers the first-best payoff~$u\FB(s;q)$ (see \cref{prop:complement-fb-br-q}).
As~$q$ rises, working on a first-best task becomes less attractive to the student while delegating becomes more attractive (since the delegation payoff~$\pi q\ybar$ rises).
As a result, it takes more accumulated skill for the payoff from working on a first-best task to be more attractive than delegating.

The skill level~$s$ also determines whether~$\eta\SB(s;q)$ rises or falls when AI quality improves.
To see why, recall the thresholds~$s^\dag$ and~$\shat(q)$ defined in \cref{lem:complement-ic-intensity-q,lem:complement-sb-intensity}.
If~$s<\shat(q)$, then~$\eta\SB(s;q)=\eta\IC(s;q)$ is decreasing in~$q$ when~$s<s^\dag$ and increasing in~$q$ when~$s>s^\dag$.
In contrast, if~$s\ge\shat(q)$, then~$\eta\SB(s;q)=\eta\FB(s;q)$ is increasing in~$q$ always.
Thus, the effect of increasing~$q$ on~$\eta\SB(s;q)$ changes from negative to positive as~$s$ passes through~$\min\{s^\dag,\shat(q)\}$ from below.
Intuitively, if the student has low skill, then AI quality improvements raise AI's productivity faster than they make the student's effort more productive, and so the teacher must distort effort intensities downward to induce effort.
In contrast, if the student has high skill, then AI quality improvements make his effort more productive faster than they raise AI's productivity, and the teacher makes tasks more effort-intensive to capitalize on the student's higher willingness to exert effort.

\subsubsection{Best-response efforts}

\begin{figure}
    \centering
    \includegraphics[width=0.55\linewidth]{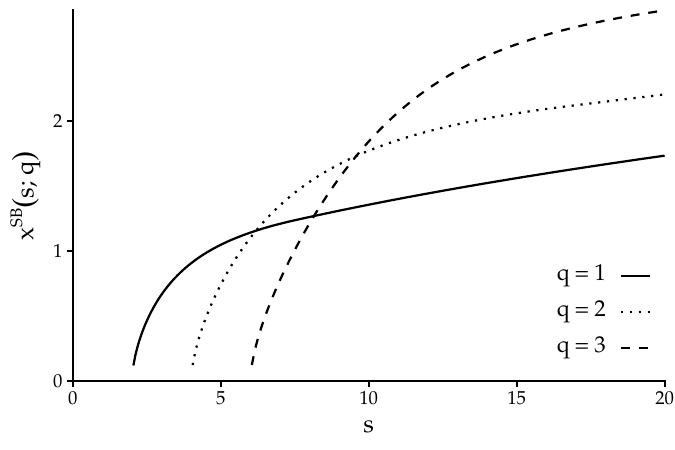}
    \caption{
    Second-best efforts~$x\SB(s;q)$ when~$\pi=1$ and~$\ybar=2$.
    }
    \label{fig:complement-sb-effort}
\end{figure}
Suppose~\eqref{eq:complement-sb-restriction-1} holds.
\cref{fig:complement-sb-effort} shows how the second-best effort
\[ x\SB(s;q)\equiv x^*(s,\eta\SB(s;q);q) \]
grows with the skill level~$s>q\ybar$ when~$(\pi,\ybar)=(1,2)$ and~$q\in\{1,2,3\}$.
If~$s=5$, then~$x\SB(s;1)>x\SB(s;2)$.
In contrast, if~$s=15$, then~$x\SB(s;1)<x\SB(s;2)$.
So~$x\SB(s;q)$ can be increasing or decreasing in the AI quality parameter~$q$, depending on the skill level~$s$.%
\footnote{
The second-best effort~$x\SB(s;q)$ can also be non-monotone in~$q$.
For example, \cref{fig:complement-sb-effort} shows~$x\SB(s;3)<x\SB(s;1)<x\SB(s;2)$ when~$(\pi,\ybar)=(1,2)$ and~$s=7$.
}
I characterize this dependence in the main result of this section:

\begin{theorem}[Second-best efforts and AI quality]
    \label{thm:complement-sb-effort}
    Suppose~\eqref{eq:complement-sb-restriction-1} holds, fix a skill level~$s>q\ybar$, and define~$s^\dag$ and~$\shat(q)$ as in \cref{lem:complement-ic-intensity-q,lem:complement-sb-intensity}.
    There is a unique skill level
    \[ s^\ddag(q)\in\left(q\ybar,\,\min\left\{s^\dag,\shat(q)\right\}\right) \]
    such that the second-best effort~$x\SB(s;q)$ is decreasing in the AI quality parameter~$q$ when~$s<s^\ddag(q)$ and increasing in~$q$ when~$s>s^\ddag(q)$.
    Moreover, the threshold~$s^\ddag(q)$ is increasing in~$q$.
\end{theorem}
\ifbodyproofs\subsection{Proof of \texorpdfstring{\cref{thm:complement-sb-effort}}{Theorem~\ref{thm:complement-sb-effort}}}

\crefalias{enumi}{claim}

The result contains three claims:
\begin{enumerate}

    \item
    \label{claim:complement-sb-effort-1}
    There exists~$s^\ddag\in(q\ybar,\shat(q))$ such that~$x\SB(s;q)$ is decreasing in~$q$ when~$s<s^\ddag$ and increasing in~$q$ when~$s>s^\ddag$.
    
    \item
    \label{claim:complement-sb-effort-2}
    $s^\ddag<s^\dag$.

    \item
    \label{claim:complement-sb-effort-3}
    $s^\ddag$ is increasing in~$q$.

\end{enumerate}
I prove these claims below.

\subsubsection{Proof of \texorpdfstring{\cref{claim:complement-sb-effort-1}}{Claim~\ref{claim:complement-sb-effort-1}}}

My proof of \cref{claim:complement-sb-effort-1} invokes the following lemma.

\begin{lemma}
    \label{lem:complement-ic-intensity-ddq}
    Suppose~\eqref{eq:complement-sb-restriction-1} holds, define~$\shat(q)$ as in \cref{lem:complement-sb-intensity}, let~$s\in(q\ybar,\shat(q))$, and define
    \[ h(s,\eta;q)\equiv\log\left(\frac{s}{\pi\eta q^2}\right) \]
    for all~$\eta\in(0,1]$.
    Then
    \begin{equation}
        \label{eq:complement-ic-intensity-ddq}
        \parfrac{\eta\IC(s;q)}{q}=\frac{(2-\eta\IC(s;q))(3\eta\IC(s;q)-2)}{2q\,h(s,\eta\IC(s;q);q)}.
    \end{equation}
\end{lemma}
\begin{proof}
    Differentiating~\eqref{eq:complement-ic-intensity-definition}  with respect to~$q$ gives
    \begin{equation}
        \label{eq:complement-ic-incensity-ift-q}
        \parfrac{u^*(s,\eta;q)}{q}\bigg\vert_{\eta=\eta\IC(s;q)}+\parfrac{u^*(s,\eta;q)}{\eta}\bigg\vert_{\eta=\eta\IC(s;q)}\parfrac{\eta\IC(s;q)}{q}=\pi\ybar.
    \end{equation}
    Now~\eqref{eq:complement-br-payoff-expanded} implies
    \begin{align}
        \parfrac{u^*(s,\eta;q)}{q}\bigg\vert_{\eta=\eta\IC(s;q)}
        &= \frac{2\eta u^*(s,\eta;q)}{(2-\eta)q}\bigg\vert_{\eta=\eta\IC(s;q)} \notag \\
        &= \frac{2\pi\eta\IC(s;q)\ybar}{2-\eta\IC(s;q)}. \label{eq:complement-ic-intensity-dudq}
    \end{align}
    Moreover, the proof of \cref{lem:complement-br-payoff} implies
    \begin{align}
        \parfrac{u^*(s,\eta;q)}{\eta}\bigg\vert_{\eta=\eta\IC(s;q)}
        &= -\frac{2u^*(s,\eta;q)}{(2-\eta)^2}h(s,\eta;q)\bigg\vert_{\eta=\eta\IC(s;q)} \notag \\
        &= -\frac{2\pi q\ybar}{(2-\eta\IC(s;q))^2}h(s,\eta\IC(s;q);q), \label{eq:complement-ic-intensity-dudeta}
    \end{align}
    where~$h(s,\eta\IC(s;q);q)>0$ because~$\eta\IC(s;q)<s/\pi q^2$.
    Substituting~\eqref{eq:complement-ic-intensity-dudq} and~\eqref{eq:complement-ic-intensity-dudeta} into~\eqref{eq:complement-ic-incensity-ift-q} gives~\eqref{eq:complement-ic-intensity-ddq}.
\end{proof}

\begin{proof}[Proof of \cref{claim:complement-sb-effort-1}]
    Suppose~$s\ge\shat(q)$.
    Then~$x\SB(s;q)=x\FB(s;q)$, which, by \cref{prop:complement-fb-br-q}, is increasing in~$q$.

    Now suppose~$s<\shat(q)$.
    Then~$\eta\SB(s;q)=\eta\IC(s;q)<1$.
    But~$\eta\IC(s;q)$ satisfies~\eqref{eq:complement-ic-intensity-definition}, which, by \cref{lem:complement-br-payoff}, implies
    \[ x\SB(s;q)=\sqrt{\frac{2\pi\eta\IC(s;q)q\ybar}{2-\eta\IC(s;q)}}. \]
    Differentiating with respect to~$q$ gives
    \begin{align*}
        \parfrac{x\SB(s;q)}{q}
        &= \frac{2\pi q\ybar}{x\SB(s;q)\left(2-\eta\IC(s;q)\right)^2}\left(\frac{\eta\IC(s;q)\left(2-\eta\IC(s;q)\right)}{2q}+\parfrac{\eta\IC(s;q)}{q}\right),
    \end{align*}
    which is positive if and only if
    \begin{equation}
        \label{eq:complement-sb-effort-condition-1}
        \parfrac{\eta\IC(s;q)}{q}>-\frac{\eta\IC(s;q)\left(2-\eta\IC(s;q)\right)}{2q}.
    \end{equation}
    Our goal is to express~\eqref{eq:complement-sb-effort-condition-1} as a condition on the skill level~$s$.
    Define~$h(s,\eta;q)$ as in \cref{lem:complement-ic-intensity-ddq} and define
    \[ g(s;q)\equiv\eta\IC(s;q)\left(3+h(s,\eta\IC(s;q);q)\right)-2 \]
    for all~$s\in(q\ybar,\shat(q))$.
    Substituting~\eqref{eq:complement-ic-intensity-ddq} into~\eqref{eq:complement-sb-effort-condition-1} and rearranging the result gives
    \begin{equation}
        \label{eq:complement-sb-effort-condition-2}
        g(s;q)>0.
    \end{equation}
    Conditions~\eqref{eq:complement-sb-effort-condition-1} and~\eqref{eq:complement-sb-effort-condition-2} are equivalent, and so~$x\SB(s;q)$ is increasing in~$q$ if and only if~\eqref{eq:complement-sb-effort-condition-2} holds.
    Now~$g(s;q)\to-2$ as~$s\to q\ybar$ from above because~$\eta\IC(s;q)\,h(s,\eta\IC(s;q);q)\to0$.
    Moreover, since~$h(s,\eta;q)$ is continuous in its arguments and~$\eta\IC(s;q)\to\eta\FB(s;q)$ as~$s\to\shat(q)$ from below, we have~$h(s,\eta\IC(s;q);q)\to h(s,\eta\FB(s;q);q)=2/\eta\FB(s;q)-1$ and hence~$g(s;q)\to2\eta\FB(s;q)>0$ as~$s\to\shat(q)$ from below.
    But~$g(s;q)$ is continuous in~$s$ and
    \begin{align*}
        \parfrac{g(s;q)}{s}
        &= \parfrac{\eta\IC(s;q)}{s}\left(2+h(s,\eta\IC(s;q);q)\right)+\frac{\eta\IC(s;q)}{s}
    \end{align*}
    is positive because~$\eta\IC(s;q)$ is increasing in~$s$ (by \cref{lem:complement-ic-intensity}).
    So~$g(s;q)$ rises continuously from a negative number to a positive number as~$s$ rises from~$q\ybar$ to~$\shat(q)$.
    So, by the intermediate value theorem, there is a unique skill level
    \[ s^\ddag\in\left(q\ybar,\,\shat(q)\right) \]
    such that~$g(s^\ddag;q)=0$, and~\eqref{eq:complement-sb-effort-condition-2} holds if and only if~$s>s^\ddag$.
    It follows that~$x\SB(s;q)$ is decreasing in~$q$ when~$s<s^\ddag$ and increasing in~$q$ when~$s>s^\ddag$.
\end{proof}

\subsubsection{Proof of \texorpdfstring{\cref{claim:complement-sb-effort-2}}{Claim~\ref{claim:complement-sb-effort-2}}}

\begin{proof}[\nopunct\unskip]
    Suppose~$s^\dag<\shat(q)$, and define~$g(s;q)$ and~$h(s,\eta;q)$ as in the proof of \cref{claim:complement-sb-effort-1}.
    Then \cref{lem:complement-ic-intensity} implies
    \begin{align*}
        g(s^\dag;q)
        &= \eta\IC(s^\dag;q)\left(3+h(s^\dag,\eta\IC(s^\dag;q);q)\right)-2 \\
        &= \frac{2}{3}\left(3+h\left(s^\dag,\frac{2}{3};q\right)\right)-2 \\
        &= \frac{2}{3}h\left(s^\dag,\frac{2}{3};q\right) \\
        &> 0,
    \end{align*}
    and therefore~$s^\ddag<s^\dag$ because~$g(s;q)$ is increasing in~$s$ and~$g(s^\ddag;q)=0<g(s^\dag;q)$.

    Now suppose~$s^\dag\ge\shat(q)$.
    Then~$s^\ddag<\shat(q)\le s^\dag$ and thus~$s^\ddag<s^\dag$.
\end{proof}

\subsubsection{Proof of \texorpdfstring{\cref{claim:complement-sb-effort-3}}{Claim~\ref{claim:complement-sb-effort-3}}}

\begin{proof}[\nopunct\unskip]
    Define~$g(s;q)$ as in the proof of \cref{claim:complement-sb-effort-1} and define~$\eta^\ddag\equiv\eta\SB(s^\ddag;q)\in(0,2/3)$ for convenience.
    Since~$s^\ddag<\shat(q)$, we have~$\eta\SB(s^\ddag;q)\equiv\eta\IC(s^\ddag;q)$.
    Thus
    \[ 0
        = g(s^\ddag;q)
        = \eta^\ddag\left(3+\log\left(\frac{s^\ddag}{\pi\eta^\ddag q^2}\right)\right)-2, \]
    which can be rewritten as
    \begin{equation}
        \label{eq:complement-sb-effort-sddag}
        s^\ddag=\pi\eta^\ddag q^2\exp\left(\frac{2}{\eta^\ddag}-3\right).
    \end{equation}
    Now~$s\mapsto\log(s)$ is an increasing transformation on~$(0,\infty)$, so~$s^\ddag$ is increasing in~$q$ if and only if
    \begin{align}
        \parfrac{\log(s^\ddag)}{q}
        &= \parfrac{}{q}\left[\log(\pi)+\log(\eta^\ddag)+2\log(q)+\frac{2}{\eta^\ddag}-3\right] \notag \\
        &= \frac{1}{\eta^\ddag}\left(1-\frac{2}{\eta^\ddag}\right)\parfrac{\eta^\ddag}{q}+\frac{2}{q} \label{eq:complement-sb-effort-dlogsdq}
    \end{align}
    is positive.

    To determine~\eqref{eq:complement-sb-effort-dlogsdq}, we need to determine~$\partial\eta^\ddag/\partial q$.
    Substituting~$(s,\eta)=(s^\ddag,\eta^\ddag)$ and~\eqref{eq:complement-sb-effort-sddag} into~\eqref{eq:complement-br-payoff-expanded} gives
    \begin{align*}
        u^*(s^\ddag,\eta^\ddag;q)
        &= \left(1-\frac{\eta^\ddag}{2}\right)\pi^2\eta^\ddag q^2\exp\left(\frac{2(2-3\eta^\ddag)(1-\eta^\ddag)}{\eta^\ddag(2-\eta^\ddag)}\right).
    \end{align*}
    But~$u^*(s^\ddag,\eta^\ddag;q)=\pi q\ybar$ by definition.
    Combining these two expressions yields an implicit equation
    \begin{equation}
        \label{eq:complement-sb-effort-3-implicit}
        q=Q(\eta^\ddag)\equiv\frac{2\ybar}{\pi\eta^\ddag(2-\eta^\ddag)}\exp\left(-\frac{2(2-3\eta^\ddag)(1-\eta^\ddag)}{\eta^\ddag(2-\eta^\ddag)}\right)
    \end{equation}
    defining~$\eta^\ddag$.

    Now~$Q:(0,2/3)\to\R$ is differentiable on its domain.
    Denoting its derivative by~$Q'$ gives
    \begin{align*}
        \frac{Q'(\eta)}{Q(\eta)}
        &= \parfrac{\log(Q(\eta))}{\eta} \\
        &= \parfrac{}{\eta}\left[\log\left(\frac{2\ybar}{\pi}\right)-\log\left(\eta(2-\eta)\right)-\frac{2(2-3\eta)(1-\eta)}{\eta(2-\eta)}\right] \\
        &= -\frac{2P_1(\eta)}{\eta^2(2-\eta)^2},
    \end{align*}
    where~$P_1(\eta)\equiv\eta^3-2\eta^2+6\eta-4<0$ for all~$\eta\in(0,2/3)$.
    Differentiating~\eqref{eq:complement-sb-effort-3-implicit} with respect to~$q$ and rearranging the result then gives
    \begin{align}
        \parfrac{\eta^\ddag}{q}
        &= \frac{1}{Q'(\eta^\ddag)} \notag \\
        &= -\frac{(\eta^\ddag)^2(2-\eta^\ddag)^2}{2qP_1(\eta^\ddag)} \label{eq:complement-sb-effort-detadq}
    \end{align}
    because~$Q(\eta^\ddag)=q$.
    Moreover, substituting~\eqref{eq:complement-sb-effort-detadq} into~\eqref{eq:complement-sb-effort-dlogsdq} gives
    \begin{align*}
        \parfrac{\log(s^\ddag)}{q}
        &= \frac{1}{\eta^\ddag}\left(1-\frac{2}{\eta^\ddag}\right)\left(-\frac{(\eta^\ddag)^2(2-\eta^\ddag)^2}{2qP_1(\eta^\ddag)}\right)+\frac{2}{q} \\
        &= \frac{P_2(\eta^\ddag)}{2qP_1(\eta^\ddag)},
    \end{align*}
    where~$P_2(\eta)\equiv3\eta^3-2\eta^2+12\eta-8<0$ for all~$\eta\in(0,2/3)$.
    Thus~$\partial\log(s^\ddag)/\partial q>0$, and so~$s^\ddag$ is increasing in~$q$.
\end{proof}
\fi

In the second-best course, the student gains skill~$s_{t+1}-s_t=rx\SB(s_t;q)$ proportional to the effort~$x\SB(s_t;q)$ he exerts on each task~$t\in\Tcal$.
Thus, \cref{thm:complement-sb-effort} says that under a second-best design, improvements in AI quality make high-skill students gain skill faster but low-skill students gain skill slower.
This is because high-skill students become more willing to work (since their productivity rises faster than AI's) but low-skill students become more tempted to delegate (since their productivity rises slower than AI's).

\cref{thm:complement-sb-effort} predicts heterogeneity in the impact of AI across students who use AI differently.
Students who use AI to augment their effort will learn more when AI quality improves.
In contrast, students who use AI as a substitute for effort will learn less when AI quality improves.
This prediction is consistent with empirical evidence that the educational impact of AI depends on how it is used \citep{Contractor-Reyes-2026-a,McNichols-etal-2026-Proc.LAK2616thInt.Learn.Anal.Knowl.Conf.,Pallant-etal-2026-StudHighEd,Shen-Tamkin-2026-,Zamarripa-2026-}.

\cref{thm:complement-sb-effort} also predicts that improvements in AI quality will widen the distribution of skill across students.
This may appear to conflict with empirical evidence that AI narrows the distribution of measured performance among students \citep{Bergenholtz-etal-2026-AMLE,Hausman-etal-2025-} and workers \citep{Brynjolfsson-etal-2025-QJE,Noy-Zhang-2023-Science}.
However, this conflict can be resolved by distinguishing output from skill.
In my model, delegation raises the student's output when he has low skill and cannot produce much output on his own.
So delegating to AI raises a low-skill student's output toward high-skill students' while leaving his skill unchanged.
Thus, compression in performance and divergence in skill are not competing predictions but joint implications of the same mechanism.

\subsection{Second-best skill gains}
\label{sec:complement-sb-gains}

Suppose~\eqref{eq:complement-sb-restriction-1} and~\eqref{eq:complement-sb-restriction-2} hold, and the teacher designs a second-best course: $\eta_t=\eta\SB(s_t;q)$ for each task~$t\in\Tcal$.
Then the student's overall skill gain
\[ s_T-s_0=r\sum_{t\in\Tcal}x\SB(s_t;q) \]
depends on his best-response efforts~$x\SB(s_t;q)$, which depend on his skill levels~$s_t$ and the AI quality parameter~$q$.
By \cref{thm:complement-sb-effort}, increasing~$q$ raises~$x\SB(s_t;q)$ when~$s_t$ is large but lowers~$x\SB(s_t;q)$ when~$s_t$ is small.
The overall effect of increasing~$q$ on~$(s_T-s_0)$ depends on whether the positive effects at high skill levels dominate the negative effects at low skill levels.

\begin{figure}
    \centering
    \includegraphics[width=0.55\linewidth]{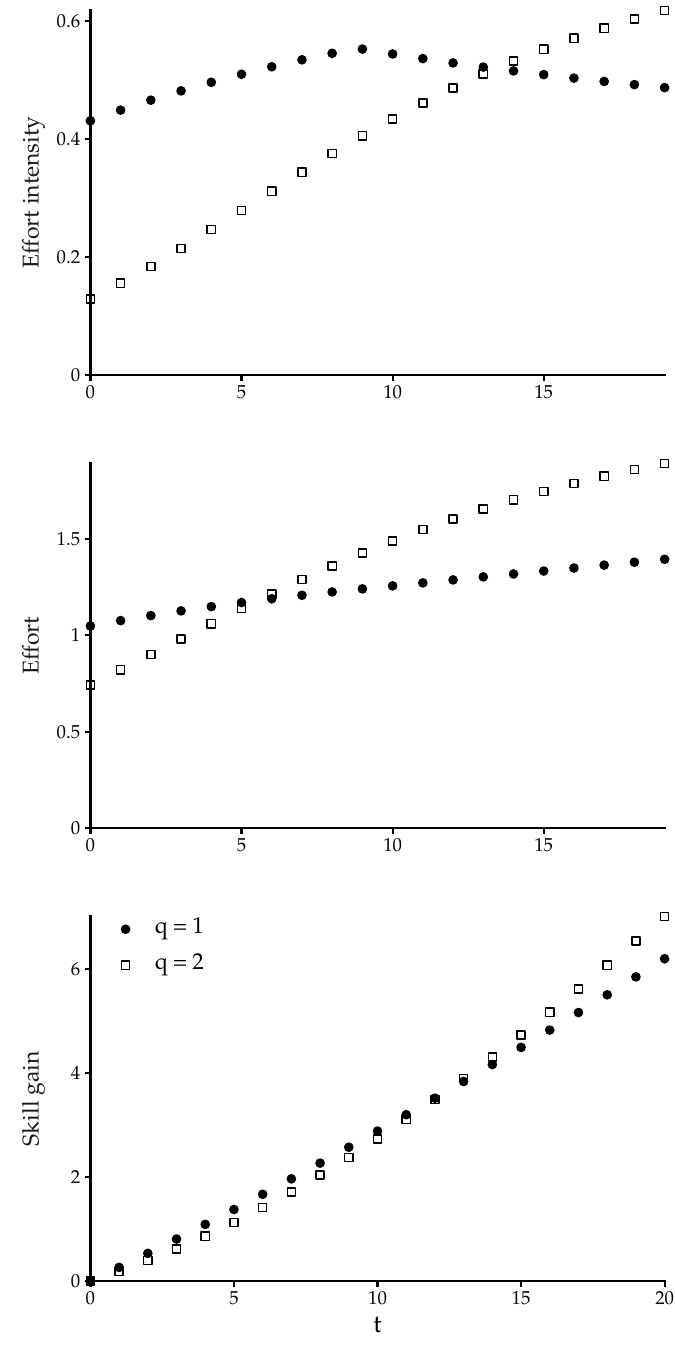}
    \caption{
    Effort intensities~$\eta_t$, best-response efforts~$x^*(s_t,\eta_t;q)$, and cumulative skill gains~$(s_t-s_0)$ in second-best course when~$(s_0,\pi,r,\ybar)=(5,1,0.25,2)$.
    }
    \label{fig:complement-sb-overall-gain}
\end{figure}
For example, suppose the student has initial skill~$s_0=5$, has motivation parameter~$\pi=1$, and builds skill at rate~$r=0.25$, and the AI generates output~$q\ybar$ with~$\ybar=2$.
\cref{fig:complement-sb-overall-gain} shows that the student exerts more effort~$x\SB(s_0;q)$ in the second-best course when~$q=1$ than when~$q=2$.
As a result, his skill~$s_1=s_0+rx\SB(s_0;q)$ after completing the initial task is higher when~$q=1$ than when~$q=2$.
However, his effort on the next few tasks grows faster when~$q=2$.
For each task~$t\ge6$, the student exerts more effort---and so gains skill faster---when~$q=2$.
Eventually his skill level~$s_t$ under~$q=2$ surpasses the level under~$q=1$.
Indeed, for tasks~$t\ge13$ onwards, the student's skill level~$s_t$ is higher when~$q=2$ than when~$q=1$.

This example shows that the effect of increasing the AI quality~$q$ on the student's overall skill gain~$(s_T-s_0)$ depends on the number of tasks~$T$.
In the example, increasing~$q$ from one to two raises~$(s_T-s_0)$ when~$T\ge13$ but lowers~$(s_T-s_0)$ when~$T\le12$.

The effect of increasing~$q$ on~$(s_T-s_0)$ also depends on the skill accumulation rate~$r$.
Intuitively, the larger is~$r$, the faster the student's skill level~$s_t$ passes the threshold~$s^\ddag(q)$, and so the higher is the proportion of tasks on which increasing~$q$ has positive effects (by raising the marginal product of effort) rather than negative effects (by making delegation more attractive).

\begin{table}
    \centering
    \small
    \caption{
    Overall skill gains~$(s_T-s_0)$ from second-best course when~$(s_0,\pi,\ybar)=(5,1,2)$.
    }
    \label{tab:complement-sb-overall-gain}
    \begin{tabular}{crrrrrr}
\toprule
& \multicolumn{2}{c}{$r=0.25$} & \multicolumn{2}{c}{$r=0.5$} & \multicolumn{2}{c}{$r=1$} \\
\cmidrule(lr){2-3} \cmidrule(lr){4-5} \cmidrule(lr){6-7}
$T$ & $q=1$ & $q=2$ & $q=1$ & $q=2$ & $q=1$ & $q=2$ \\
\midrule
5 & 1.38 & 1.13 & 2.86 & 2.61 & 6.04 & 6.34\\
10 & 2.88 & 2.73 & 6.15 & 6.79 & 13.64 & 16.60\\
20 & 6.20 & 7.01 & 13.85 & 17.24 & 32.96 & 41.89\\
\bottomrule
\end{tabular}

\end{table}
I quantify this intuition in \cref{tab:complement-sb-overall-gain}.
It shows how~$(s_T-s_0)$ depends on the number of tasks~$T$ and skill accumulation rate~$r$ when~$(s_0,\pi,\ybar)=(5,1,2)$ and~$q\in\{1,2\}$.
If~$T=5$, then raising~$q$ from one to two lowers~$(s_T-s_0)$ when~$r\in\{0.25,0.5\}$ but raises~$(s_T-s_0)$ when~$r=1$.
In contrast, if~$T=20$, then raising~$q$ from one to two raises~$(s_T-s_0)$ when~$r\in\{0.25,0.5\}$ \emph{and} when~$r=1$.
Intuitively, if the course is short, then the student must gain skill quickly for the positive effects of AI quality improvements on later tasks to outweigh the negative effects on early tasks.
If the course is longer, then the student does not need to gain skill as quickly for AI quality improvements to have positive effects overall.

The patterns above suggest AI quality improvements increase the return on ``learning how to learn'' \citep{Novak-Gowin-1984-}.
If the student gains more skill per unit of effort (i.e., $r$ is higher), then he overcomes his temptation to delegate earlier in the course and so is faster to realize the productivity-enhancing benefits of AI.

\section{Discussion of modeling assumptions}
\label{sec:assumptions}

This section discusses my modeling assumptions and the implications of changing them.
To economize on notation, I focus on the case with AI quality~$q=1$ and suppress~$q$ from my discussion.
However, my conclusions do not change qualitatively when I allow~$q\ge1$ to be arbitrary.

\subsection{Cobb-Douglas production}
\label{sec:cobb-douglas}

I assume the student produces output
\begin{equation}
    \label{eq:output-cobb-douglas}
    y(x,s,\eta)\equiv x^\eta s^{1-\eta}
\end{equation}
when he exerts effort~$x\ge0$, has skill~$s>0$, and works on a task with effort-intensity~$\eta\in(0,1]$.
The Cobb-Douglas production function~\eqref{eq:output-cobb-douglas} is increasing and weakly concave in~$x$, ensuring the student's best-response effort is unique and the teacher's problem is well-defined (see \cref{sec:quadratic-costs}).
The function also has two properties that are essential for my results:
\begin{enumerate}

    \item[(i)]
    A single argument~$\eta$ controls the relative productivities of effort and skill.

    \item[(ii)]
    The cross-derivative
     \[ \parfrac{^2y(x,s,\eta)}{x\,\partial s}=\eta(1-\eta)x^{\eta-1}s^{-\eta} \]
    is non-negative (and positive when~$\eta<1$).
    Thus, effort and skill are complementary inputs to production: having more of one raises the marginal productivity of the other.

\end{enumerate}
Property~(i) allows me to characterize first- and second-best courses as one-dimensional mappings from skill levels~$s$ to effort intensities~$\eta\FB(s)$ and~$\eta\SB(s)$.
Property~(ii) generates the constant-then-decreasing shape of first-best intensities established in \cref{lem:fb-intensity}.
To see why, let~$\rho\in(-1,\infty)$ be non-zero and suppose the student produces output
\begin{equation}
    \label{eq:output-ces}
    y(x,s,\eta)=\left(\eta x^{-\rho}+(1-\eta)s^{-\rho}\right)^{-1/\rho}.
\end{equation}
Then his effort and skill have constant elasticity of substitution (CES)~$1/(1+\rho)$, which measures how much he can compensate for exerting less effort by being more skilled (and vice versa).%
\footnote{
\cite{Heckman-2007-PNAS} and \cite{Cunha-Heckman-2007-AER} also use CES production functions to model skill formation.
}
As with Cobb-Douglas production, the effort intensity~$\eta\in(0,1]$ controls the relative productivities of effort and skill.
But the CES production function~\eqref{eq:output-ces} has an additional parameter~$\rho$ that controls the complementarity between effort and skill.
Raising~$\rho$ makes the inputs more complementary: as~$\rho$ rises from~$-1$ to~$\infty$, the elasticity of substitution~$1/(1+\rho)$ between effort and skill falls from~$\infty$ (i.e., the inputs are perfect substitutes) to zero (i.e., they are perfect complements).
Intuitively, as~$\rho$ rises, it becomes harder for the student to substitute effort for skill because production relies more on combining both inputs in fixed proportions.%
\footnote{
Indeed, in the limit as~$\rho\to\infty$, the CES production function~\eqref{eq:output-ces} converges to the Leontief function~$\min\{x,s\}$.
But if~$y(x,s,\eta)=\min\{x,s\}$, then the student optimally combines effort and skill in fixed proportions regardless of~$\eta$.
}

In the limit as~$\rho\to0$, the CES production function~\eqref{eq:output-ces} converges to the Cobb-Douglas function~\eqref{eq:output-cobb-douglas}, and effort and skill have elasticity of substitution equal to one.%
\footnote{
Since~$z\mapsto\exp(z)$ and~$z\mapsto\log(z)$ are continuous mappings on~$(0,\infty)$, we have
\[ \lim_{\rho\to0}\left(\eta x^{-\rho}+(1-\eta)s^{-\rho}\right)^{-1/\rho}=\exp\left(-\lim_{\rho\to0}\frac{1}{\rho}\log\left(\eta x^{-\rho}+(1-\eta)s^{-\rho}\right)\right). \]
The limit inside the bracket on the right-hand side is indeterminate.
So L'H\^opital's rule implies
\begin{align*}
    \lim_{\rho\to0}\left(\eta x^{-\rho}+(1-\eta)s^{-\rho}\right)^{-1/\rho}
    &= \exp\left(-\lim_{\rho\to0}\frac{\left(\eta x^{-\rho}(-\log(x))+(1-\eta)s^{-\rho}(-\log(s))\right)}{\eta x^{-\rho}+(1-\eta)s^{-\rho}}\right) \\
    &= \exp\left(-\left(\log\left(x^{-\eta}\right)+\log\left(s^{-(1-\eta)}\right)\right)\right) \\
    &= x^{\eta}s^{1-\eta}.
\end{align*}
Thus, the CES production function~\eqref{eq:output-ces} converges to the Cobb-Douglas function~\eqref{eq:output-cobb-douglas} in the limit as~$\rho\to0$.
}
Letting~$\rho$ be greater or less than zero allows me to embed more or less complementarity than in my baseline model with Cobb-Douglas production.
In particular, I can analyze how the constant-then-decreasing shape of~$\eta\FB(s)$ depends on the complementarity between effort and skill.
\cref{lem:fb-intensity-ces} extends this shape to CES production.%
\footnote{
Condition~\eqref{eq:fb-intensity-ces-foc} is analogous to the condition~\eqref{eq:fb-intensity-foc} derived in \cref{sec:fb} for the case with Cobb-Douglas production.
In fact, since
\[ \lim_{\rho\to0}\,(1+\rho)^{-1/\rho-1}\left(\frac{\eta}{\eta+\rho}\right)^{-2/\rho-1}=\exp\left(\frac{2}{\eta}-1\right) \]
for all~$\eta\in(0,1]$, the left-hand side of~\eqref{eq:fb-intensity-ces-foc} converges to the left-hand side of~\eqref{eq:fb-intensity-foc} as~$\rho\to0$.
Likewise, as~$\rho\to0$, the skill threshold~$\pi(1+\rho)^{1/\rho}$ identified in \cref{lem:fb-intensity-ces} converges to the threshold~$\pi e$ identified in \cref{lem:fb-intensity}.
Thus, the first-best intensity~$\eta\FB(s)$ defined in \cref{lem:fb-intensity-ces} nests the intensity defined for Cobb-Douglas production and generalizes it to CES production.
}

\begin{lemma}[First-best effort intensities with CES production]
    \label{lem:fb-intensity-ces}
    Suppose the student's production function has the CES form~\eqref{eq:output-ces}.
    For all skill levels~$s>0$, there is a unique~$\eta\FB(s)\in(0,1]$ such that:
    \begin{enumerate}

        \item[(i)]
        if~$s\le\pi(1+\rho)^{1/\rho}$, then~$\eta\FB(s)=1$;

        \item[(ii)]
        if~$s>\pi(1+\rho)^{1/\rho}$, then~$\eta\FB(s)\in(\max\{-\rho,0\},1)$ satisfies
        \begin{equation}
            \label{eq:fb-intensity-ces-foc}
            \eta\FB(s)(1+\rho)^{-1/\rho-1}\left(\frac{\eta\FB(s)}{\eta\FB(s)+\rho}\right)^{-2/\rho-1}=\frac{s}{\pi},
        \end{equation}
        and is decreasing in~$s$ and the complementarity parameter~$\rho$.

    \end{enumerate}
    Moreover, the best-response effort~$x^*(s,\eta)$ attains its maximum over~$\eta\in(0,1]$ when~$\eta=\eta\FB(s)$.
\end{lemma}
\ifbodyproofs\subsection{Proof of \texorpdfstring{\cref{lem:fb-intensity-ces}}{Lemma~\ref{lem:fb-intensity-ces}}}

\cref{lem:fb-intensity-ces} is the special case of \cref{lem:complement-fb-intensity-ces} with~$q=1$.

\begin{lemma}[First-best effort intensities with CES production and complementary AI]
    \label{lem:complement-fb-intensity-ces}
    Let~$\rho\in(-1,\infty)$ be non-zero and suppose the student's production function has the CES form
    \begin{equation}
        \label{eq:complement-output-ces}
        y(x,s,\eta;q)=\left(\eta(qx)^{-\rho}+(1-\eta)s^{-\rho}\right)^{-1/\rho}.
    \end{equation}
    For all skill levels~$s>0$, there is a unique~$\eta\FB(s;q)\in(0,1]$ such that:
    
    \begin{enumerate}

        \item[(i)]
        if~$s\le\pi(1+\rho)^{1/\rho}q^2$, then~$\eta\FB(s;q)=1$;

        \item[(ii)]
        if~$s>\pi(1+\rho)^{1/\rho}q^2$, then~$\eta\FB(s;q)\in(\max\{-\rho,0\},1)$ satisfies
        \begin{equation}
            \label{eq:complement-fb-intensity-ces-foc}
            \eta\FB(s;q)(1+\rho)^{-1/\rho-1}\left(\frac{\eta\FB(s;q)}{\eta\FB(s;q)+\rho}\right)^{-2/\rho-1}=\frac{s}{\pi q^2},
        \end{equation}
        is decreasing in~$s$ and~$\rho$, and is increasing in~$q$;

        \item[(iii)]
        the best-response effort~$x^*(s,\eta;q)$  attains its maximum over~$\eta\in(0,1]$ when~$\eta=\eta\FB(s;q)$.

    \end{enumerate}
    Moreover, $\eta\FB(s;q)$ is continuous in both~$s$ and~$q$, and converges to~$\max\{-\rho,0\}$ as~$s\to\infty$.
\end{lemma}

\cref{lem:complement-fb-intensity-ces} generalizes \cref{lem:complement-fb-intensity} to CES production functions.

My proof of \cref{lem:complement-fb-intensity-ces} invokes the following result.

\begin{lemma}
    \label{lem:complement-fb-intensity-ces-function}
    Let
    \[ Z\equiv\{(z_1,z_2)\in\R^2:\max\{-z_2,0\}<z_1\le1,\ z_2>-1,\ z_2\not=0\} \]
    and let~$G:Z\to\R$ be the function defined by
    \[ G(z_1,z_2)=z_1(1+z_2)^{-1/z_2-1}\left(\frac{z_1}{z_1+z_2}\right)^{-2/z_2-1}. \]
    Then~$G$ is continuously differentiable and decreasing in both arguments on its domain.
    Moreover,
    \[ \lim_{z_1\to^+\max\{-z_2,0\}}G(z_1,z_2)=\infty\ \ \text{and}\ \ G(1,z_2)=(1+z_2)^{1/z_2} \]
    for all non-zero~$z_2>-1$.
\end{lemma}
\begin{proof}
    I first establish the boundary values of~$G(z_1,z_2)$ on~$Z$.
    Fix~$z_2>-1$ with~$z_2\not=0$ and write
    \[ G(z_1,z_2)=z_1^{-2/z_2}(z_1+z_2)^{2/z_2+1}(1+z_2)^{-1/z_2-1}. \]
    Clearly~$G(1,z_2)=(1+z_2)^{1/z_2}$.
    For the limit as~$z_1\to\max\{-z_2,0\}$ from above, there are two cases to consider:
    \begin{enumerate}

        \item[(i)]
        Suppose~$z_2>0$ so that~$\max\{-z_2,0\}=0$.
        Then, as~$z_1\to0$ from above, we have~$z_1^{-2/z_2}\to\infty$ and~$(z_1+z_2)^{2/z_2+1}\to z_2^{2/z_2+1}\in(0,\infty)$, and so~$G(z_1,z_2)\to\infty$.

        \item[(ii)]
        Now suppose~$z_2\in(-1,0)$ so that~$\max\{-z_2,0\}=-z_2$.
        Then, as~$z_1\to-z_2$ from above, we have~$z_1^{-2/z_2}\to(-z_2)^{-2/z_2}\in(0,\infty)$ and~$(z_1+z_2)^{2/z_2+1}\to\infty$, and so~$G(z_1,z_2)\to\infty$.

    \end{enumerate}
    Thus~$G(z_1,z_2)\to\infty$ as~$z_1\to\max\{-z_2,0\}$ from above.

    Next, I establish continuous differentiability.
    Define the functions~$G_1,G_2,G_3:Z\to\R$ by
    \[ G_1(z_1,z_2)\equiv z_1,\ \ G_2(z_1,z_2)\equiv(1+z_2)^{-1/z_2-1},\ \ \text{and}\ \ G_3(z_1,z_2)\equiv\left(\frac{z_1}{z_1+z_2}\right)^{-2/z_2-1}. \]
    Each function is positive and continuously differentiable on its domain, so their product~$G$ is also positive and continuously differentiable.

    It remains to show~$G(z_1,z_2)$ is decreasing in its arguments over~$Z$.
    Let~$(z_1,z_2)\in Z$ so that~$0<z_1<2$ and~$z_1+z_2>0$.
    Then
    \begin{align*}
        \parfrac{\log\left(G(z_1,z_2)\right)}{z_1}
        &= \parfrac{}{z_1}\left[\log(z_1)-\left(\frac{1}{z_2}+1\right)\log(1+z_2)-\left(\frac{2}{z_2}+1\right)\log\left(\frac{z_1}{z_1+z_2}\right)\right] \\
        &= \frac{z_1-2}{z_1(z_1+z_2)} \\
        &< 0
    \end{align*}
    and therefore
    \[ \parfrac{G(z_1,z_2)}{z_1}=G(z_1,z_2)\parfrac{\log(G(z_1,z_2))}{z_1}<0. \]
    Similarly
    \begin{align*}
        \parfrac{G(z_1,z_2)}{z_2}
        &= G(z_1,z_2)\parfrac{\log(G(z_1,z_2))}{z_2} \\
        &= G(z_1,z_2)\frac{1}{z_2^2}H(z_1,z_2),
    \end{align*}
    where~$H:Z\to\R$ is defined by
    \[ H(z_1,z_2)\equiv\log(1+z_2)+2\log\left(\frac{z_1}{z_1+z_2}\right)-\frac{z_2(z_1-2)}{z_1+z_2}. \]
    Since~$G(z_1,z_2)>0$ and~$z_2^2>0$, it suffices to show~$H(z_1,z_2)<0$.
    Now
    \begin{align*}
        \parfrac{H(z_1,z_2)}{z_2}
        &= -\frac{z_2\left((z_1-1)^2+(1+z_2)\right)}{(1+z_2)(z_1+z_2)^2}
    \end{align*}
    has the same sign as~$-z_2$ because~$1+z_2>0$ and~$(z_1+z_2)^2>0$.
    So~$H(z_1,z_2)$ is increasing in~$z_2$ over~$(-z_1,0)$ and decreasing in~$z_2$ over~$(0,\infty)$.
    But each term in~$H(z_1,z_2)$ tends to zero as~$z_2\to0$, so~$H(z_1,z_2)\to0$ as~$z_2\to0$.
    It follows that~$H(z_1,z_2)<0$ for all~$(z_1,z_2)\in Z$.
\end{proof}

\begin{proof}[Proof of \cref{lem:complement-fb-intensity-ces}]
    Fix~$s>0$ and~$\eta\in(0,1]$, and define~$g(x,s,\eta;q)\equiv\eta(qx)^{-\rho}+(1-\eta)s^{-\rho}$.
    Differentiating the student's payoff
    \[ u(x,s,\eta;q)=\pi g(x,s,\eta;q)^{-1/\rho}-\frac{x^2}{2} \]
    twice with respect to~$x$ and simplifying the result gives
    \begin{align*}
        \parfrac{^2u(x,s,\eta;q)}{x^2}
        &= -\pi(1+\rho)\eta(1-\eta)s^{-\rho}q^2(qx)^{-\rho-2}g(x,s,\eta;q)^{-1/\rho-2}-1,
    \end{align*}
    which is negative when~$x>0$, $\rho>-1$, and~$\rho\not=0$.
    So there is a unique best-response effort~$x^*\equiv x^*(s,\eta;q)>0$ satisfying the first-order condition
    \[ 0
        = \parfrac{u(x,s,\eta;q)}{x}\bigg\vert_{x=x^*}
        = \pi g(x^*,s,\eta;q)^{-1/\rho-1}\left(\eta q^{-\rho}(x^*)^{-\rho-1}\right)-x^*, \]
    which can be rewritten as
    \begin{equation}
        \label{eq:complement-br-effort-ces-foc}
        (qx^*)^{2+\rho}=\pi\eta q^2 g(x^*,s,\eta;q)^{-1/\rho-1}.
    \end{equation}
    Moreover, since~$u(x,s,\eta;q)$ is concave, the implicit function theorem implies~$x^*$ is continuously differentiable (and, thus, continuous) in~$\eta$ and that the derivative~$\partial x^*/\partial\eta$ has the same sign as
    \begin{align}
        \parfrac{^2u(x,s,\eta;q)}{x\,\partial\eta}\bigg\vert_{x=x^*}
        &= \pi g(x^*,s,\eta;q)^{-1/\rho-2}q^{-\rho}(x^*)^{-\rho-1}\left(s^{-\rho}+\frac{\eta}{\rho}\left(s^{-\rho}-(qx^*)^{-\rho}\right)\right). \label{eq:complement-d2udxdeta-ces}
    \end{align}
    If~$\eta\le\max\{-\rho,0\}$, then~\eqref{eq:complement-d2udxdeta-ces} is always positive;
    if~$\eta>\max\{-\rho,0\}$, then~\eqref{eq:complement-d2udxdeta-ces} is positive if and only if
    \begin{equation}
        \label{eq:complement-fb-intensity-ces-condition-1}
        s<\left(1+\frac{\rho}{\eta}\right)^{1/\rho}qx^*.
    \end{equation}
    Now the first-order condition~\eqref{eq:complement-br-effort-ces-foc} implies~$x^*\to0$ as~$\eta\to0$.
    Therefore, the continuity of~$x^*$ in~$\eta$ extends from~$(0,1]$ to the compact interval~$[0,1]$.
    Continuity and compactness then imply~$x^*$ attains its maximum over~$\eta\in[0,1]$.
    But its maximizer(s) must be positive since~$x^*\to0$ as~$\eta\to0$ while~$x^*>0$ for all~$\eta>0$.
    Therefore, there exists~$\eta^*\in(0,1]$ such that~$x^*$ attains its maximum over~$\eta\in(0,1]$ when~$\eta=\eta^*$.

    Assume~$\eta^*$ is an interior maximizer:~$\eta^*<1$.
    Then
    \begin{equation}
        \label{eq:complement-fb-intensity-ces-foc-interior-1}
        s=\left(1+\frac{\rho}{\eta^*}\right)^{1/\rho}qx^*,
    \end{equation}
    which requires~$\eta^*>\max\{-\rho,0\}$.
    Substituting~\eqref{eq:complement-fb-intensity-ces-foc-interior-1} into the first-order condition~\eqref{eq:complement-br-effort-ces-foc} and rearranging the result gives
    %
    \begin{equation}
        \label{eq:complement-fb-intensity-foc-interior-2}
        \frac{s}{\pi q^2}=G(\eta^*,\rho),
    \end{equation}
    where~$G$ is the function defined in \cref{lem:complement-fb-intensity-ces-function}.
    But~$G(\eta,\rho)$ falls continuously from~$\infty$ to~$(1+\rho)^{1/\rho}$ as~$\eta$ rises from~$\max\{-\rho,0\}$ to one.
    So there are two cases to consider:
    \begin{enumerate}

        \item[(i)]
        Suppose~$s\le\pi(1+\rho)^{1/\rho}q^2$.
        Then~$s/\pi q^2<G(\eta,\rho)$ for all~$\eta<1$, contradicting~\eqref{eq:complement-fb-intensity-foc-interior-2}.
        So the assumption that~$\eta^*<1$ was false, implying~$\eta^*=1$.

        \item[(ii)]
        Now suppose~$s>\pi(1+\rho)^{1/\rho}q^2$.
        Then~~$\eta^*<1$ satisfies the implicit equation~\eqref{eq:complement-fb-intensity-foc-interior-2} uniquely and maximizes the best-response effort~$x^*(s,\eta;q)$ over~$\eta\in(0,1]$.%
        \footnote{
        To see why~$\eta^*<1$ in case~(ii), assume towards a contradiction that~$\eta^*=1$.
        Then~$x^*=\pi q$ and so~$s>(1+\rho)^{1/\rho}qx^*$.
        But then~\eqref{eq:complement-d2udxdeta-ces} is negative and hence~$\partial x^*/\partial\eta<0$, contradicting the maximality of~$x^*(s,\eta^*;q)$.
        }
        As~$s$ rises or~$q$ falls, the left-hand side of~\eqref{eq:complement-fb-intensity-foc-interior-2} rises, and so~$\eta^*$ must fall because~$G(\eta,\rho)$ is decreasing in~$\eta$.
        Likewise, as~$\rho$ rises, the maximizer~$\eta^*$ must fall because~$G(\eta,\rho)$ is decreasing in~$\eta$ and~$\rho$ (by \cref{lem:complement-fb-intensity-ces-function}) but~$s/\pi q^2$ does not change.
        It follows that~$\eta^*$ is decreasing in~$s$ and~$\rho$, and increasing in~$q$.

    \end{enumerate}
    Letting~$\eta\FB(s;q)\equiv\eta^*$ yields parts~(i)--(iii) of the result.
    It remains to show~$\eta\FB(s;q)$ is continuous in~$s$ and~$q$, and converges to~$\max\{-\rho,0\}$ as~$s\to\infty$.

    \cref{lem:complement-fb-intensity-ces-function} implies that for all non-zero~$\rho>-1$, the mapping~$\eta\mapsto G(\eta,\rho)$ is continuously differentiable on~$(\max\{-\rho,0\},1]$ and has~$\partial G(\eta,\rho)/\partial\eta\not=0$.
    The implicit function theorem then implies~$\eta\FB(s;q)$ is (differentiable and thus) continuous in both~$s$ and~$q$ when~$s>\pi(1+\rho)^{1/\rho}q^2$.
    We also have~$\eta\FB(s;q)\to1$ as~$s\to\pi(1+\rho)^{1/\rho}q^2$ from above, from which it follows that~$\eta\FB(s;q)$ is globally continuous in~$s$ and~$q$.

    Finally, the right-hand side of~\eqref{eq:complement-fb-intensity-ces-foc} grows without bound as~$s\to\infty$, while the left-hand side diverges to~$\infty$ only as~$\eta\FB(s;q)\to\max\{-\rho,0\}$.
    Thus~$\eta\FB(s;q)\to\max\{-\rho,0\}$ as~$s\to\infty$.
\end{proof}
\fi

For all skill levels~$s>0$, the first-best effort intensity~$\eta\FB(s)$ is constant in the complementarity parameter~$\rho$ when~$s\le\pi(1+\rho)^{1/\rho}$ and is decreasing in~$\rho$ when~$s>\pi(1+\rho)^{1/\rho}$.
The intuition is as follows.
The teacher designs her first-best course to maximize the student's best-response effort.
He chooses this effort by equating its marginal productivity and cost.
As~$\rho$ rises, effort and skill become more complementary: production relies more on using both inputs, not just one.
As a result, the student's productivity-cost trade-off becomes more sensitive to his skill level~$s$.
The teacher internalizes this increase in sensitivity by making the first-best effort intensity~$\eta\FB(s)$ more sensitive to~$s$.

\begin{figure}
    \centering
    \includegraphics[width=0.55\linewidth]{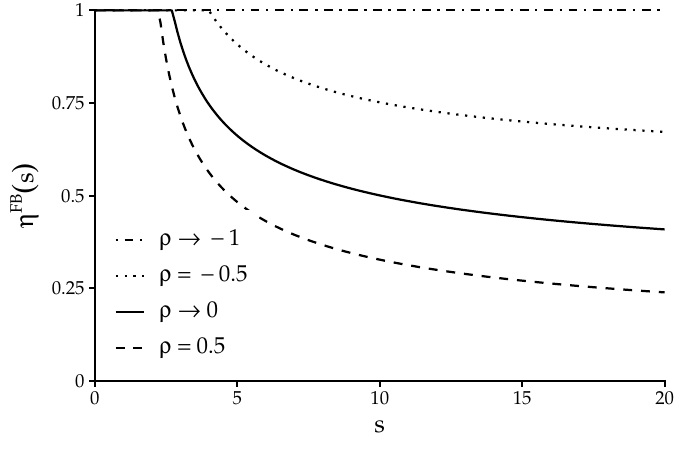}
    \caption{
    First-best effort intensities~$\eta\FB(s)$ with CES production and~$\pi=1$.
    }
    \label{fig:fb-intensity-ces}
\end{figure}
\cref{fig:fb-intensity-ces} compares the cases with~$\rho>0$ and~$\rho\in(-1,0)$ to the Cobb-Douglas case with~$\rho\to0$.
If~$\rho>0$, then effort and skill are \emph{more} complementary than under Cobb-Douglas production, and the first-best course design is \emph{more} sensitive to the student's skill level.
In contrast, if~$\rho\in(-1,0)$, then effort and skill are \emph{less} complementary than under Cobb-Douglas production, and the first-best design is \emph{less} sensitive to the student's skill level.
In the limit as~$\rho\to-1$ (i.e., when effort and skill are perfect substitutes), the student's best-response effort~$x^*(s,\eta)\to\pi\eta$ is independent of his skill level~$s$ and the teacher maximizes~$x^*(s,\eta)$ by choosing~$\eta=1$ regardless of~$s$.

Thus, the constant-then-decreasing shape of~$\eta\FB(s)$ stems from the complementarity between effort and skill.
Assuming complementarity is appropriate for my setting: a student with higher skill allocates effort more efficiently.
Changing the degree of complementarity (as measured by the parameter~$\rho$) changes the student's and teacher's choices, but does not change my qualitative result that first-best courses comprise tasks with non-increasing effort intensities (\cref{thm:fb-solution}).

\subsection{Quadratic costs}
\label{sec:quadratic-costs}

I assume the student exerts effort~$x\ge0$ at quadratic cost~$c(x)\equiv x^2/2$.
This allows me to derive an explicit, closed-form expression for the student's best-response effort (see \cref{lem:br-effort}).
This would not be possible if the cost function~$c:[0,\infty)\to\R$ were arbitrary.
Assuming~$c$ is convex ensures the student's objective~$u(x,s,\eta)\equiv\pi x^\eta s^{1-\eta}-c(x)$ is concave in~$x$ for all skill levels~$s>0$ and effort intensities~$\eta\in(0,1]$.
As a result, there is always a unique and finite maximizer that the student chooses as his best-response effort.%
\footnote{
In particular, suppose the teacher designs a task with full effort intensity~$\eta=1$.
Then the student's valuation~$\pi x$ of his output is linear in his effort~$x$.
If the cost of exerting effort is also linear, then the student optimally exerts zero or infinite effort, depending on whether the constant marginal cost of effort exceeds~$\pi$.
}
Thus, convex costs ensure the teacher's problem---maximize best-response effort subject to an incentive constraint---is well-defined.
Quadratic costs are a convenient normalization.

Convex costs are also necessary for my result that first-best courses comprise tasks with gradually decreasing effort intensities (\cref{thm:fb-solution}).
For example, suppose the student exerts effort~$x\ge0$ at cost~$c(x)\equiv x^{1+\gamma}/(1+\gamma)$, where the parameter
\[ \gamma=\frac{c''(x)x}{c'(x)} \]
is non-negative and is higher when costs are more convex.
Following the proof of \cref{lem:br-effort}, the student exerts best-response effort
\[ x^*(s,\eta)=s^{(2-\eta)/(\gamma+1-\eta)}\left(\frac{\pi\eta}{s}\right)^{1/(\gamma+1-\eta)}, \]
which is increasing in~$\eta$ if
\begin{equation}
    \label{eq:br-effort-intensity-condition-gamma}
    \eta\exp\left(\frac{\gamma+1}{\eta}-1\right)>\frac{s^\gamma}{\pi}
\end{equation}
and non-increasing otherwise.
If~$\gamma=0$ (i.e., effort costs are linear), then the condition~\eqref{eq:br-effort-intensity-condition-gamma} is independent of the skill level~$s$ and so the effort intensity~$\eta$ that maximizes~$x^*(s,\eta)$ is independent of~$s$.
In contrast, if~$\gamma>0$ (i.e., effort costs are convex), then the right-hand side of~\eqref{eq:br-effort-intensity-condition-gamma} is increasing in~$s$ and so the maximizer falls as~$s$ rises.
Thus, first-best courses comprise tasks with gradually decreasing effort intensities if and only if~$\gamma>0$.

Intuitively, as the student completes tasks and builds skill, he becomes more willing to exert effort because his higher skill makes effort more productive.
But if he exerts more effort, then the marginal cost of effort rises relative to skill, since effort costs are convex but skill is costless.
To recover the optimal trade-off between productivity and cost, the teacher steers the student toward the less-marginally-expensive input by making tasks less effort-intensive and more skill-intensive.

\subsection{AI-generated output is constant}
\label{sec:constant-ybar}

I assume that when the student delegates a task~$t$, the AI-generated output~$\ybar$ does not depend on his skill level~$s_t$ or the task's effort intensity~$\eta_t$.
Independence from~$s_t$ follows from my definition of ``delegation'': the student hands off the entire task to AI.
So when the student delegates, the output he obtains reflects AI's ability rather than his own.
In contrast, when the student works and uses AI as a complement to effort (see \cref{sec:complement}), both his and AI's abilities determine output.
I use separate parameters ($\ybar$ and~$q$) to quantify AI's ability as a delegate and complement to keep the two modes of AI assistance separate.

Independence from~$\eta_t$ is a more substantive assumption.
It is a middle ground between two natural alternatives: letting~$\ybar$ increase or decrease with~$\eta_t$.
I consider these alternatives below.

When the teacher designs second-best courses (see \cref{sec:sb}), she satisfies the incentive constraint~\eqref{eq:ic-constraint} by lowering~$\eta_t$ on early tasks~$t$.
This makes the tasks more skill-intensive and raises the student's payoff~$u^*(s_t,\eta_t)$ from working on them.
If~$\ybar$ were increasing in~$\eta_t$, then lowering~$\eta_t$ would also loosen~\eqref{eq:ic-constraint} because it would lower the payoff~$\ubar\equiv\pi\ybar$ from delegating the task.
So the teacher would not have to distort tasks as far toward skill intensity as \cref{thm:sb-solution} describes.

In contrast, if~$\ybar$ were decreasing in~$\eta_t$, then lowering~$\eta_t$ may \emph{tighten}~\eqref{eq:ic-constraint} because the delegation payoff could rise by more than the payoff from working.%
\footnote{
In particular, lowering~$\eta_t$ would tighten~\eqref{eq:ic-constraint} if
\[ -\pi\parfrac{\ybar}{\eta}\bigg\vert_{\eta=\eta_t}>-\parfrac{u^*(s_t,\eta)}{\eta}\bigg\vert_{\eta=\eta_t}. \]
}
Then the teacher would need to distort tasks \emph{away} from skill intensity, leading to qualitatively different courses than my model predicts.

\subsection{Skill grows with effort}
\label{sec:accumulation-spec}

I assume the student builds skill proportional to the effort he exerts on a task.
Proportionality is a convenient normalization and allows me to analyze the comparative statics of skill development by varying a single parameter: the skill accumulation rate~$r$ (see \cref{sec:sb-gains,sec:complement-sb-gains}).
However, proportionality is not necessary for my characterization of first- and second-best courses: so long as the skill increment~$(s_{t+1}-s_t)$ is increasing in the student's best-response effort~$x^*(s_t,\eta_t)$, the teacher cannot do better than choose the effort intensity~$\eta_t\in(0,1]$ that maximizes~$x^*(s_t,\eta_t)$.

Thus, my substantive assumption is that skill grows with effort.
This is consistent with cognitive scientific research on deliberate practice \citep{Ericsson-etal-1993-PsychRev} and ``desirable difficulty'' \citep{Bjork-Bjork-2011-Psychologyandtherealworld:Essaysillustratingfundamentalcontributionstosociety}, which emphasizes that learning requires effortful engagement.
The assumption also aligns with \citeapos{Arrow-1962-REStud} model of learning-by-doing, which assumes productivity grows with cumulative investment in a productive input.

An alternative is to assume the student's skill grows proportionally to the output
\begin{equation}
    \label{eq:br-output}
    y^*(s,\eta)\equiv\left(x^*(s,\eta)\right)^\eta s^{1-\eta}
\end{equation}
he produces from exerting effort~$x^*(s,\eta)$.
This would align my model with \citeapos{Rosen-1972-QJE} model of learning-by-doing, which assumes productivity grows with cumulative output.
\cref{lem:br-output} characterizes the ``best-response output''~$y^*(s,\eta)$ and its dependence on the model parameters.
\begin{lemma}[Best-response outputs]
    \label{lem:br-output}
    Given a skill level~$s>0$ and effort intensity~$\eta\in(0,1]$, the student's best-response effort~$x^*(s,\eta)$ produces output
    \[ y^*(s,\eta)=\frac{1}{\pi\eta}\left(x^*(s,\eta)\right)^2. \]
    This output is
    (i)~positive,
    (ii)~increasing in the motivation parameter~$\pi$,
    (iii)~increasing in~$s$ when~$\eta<1$ and constant in~$s$ when~$\eta=1$, and
    (iv)~decreasing in~$\eta$ when
    \begin{equation}
        \label{eq:br-output-intensity-condition}
        \eta\exp\left(1-\frac{\eta}{2}\right)<\frac{s}{\pi}
    \end{equation}
    and non-decreasing otherwise.
\end{lemma}
\ifbodyproofs\subsection{Proof of \texorpdfstring{\cref{lem:br-output}}{Lemma~\ref{lem:br-output}}}

\cref{lem:br-output} is the special case of \cref{lem:complement-br-output} with~$q=1$.

\begin{lemma}[Best-response outputs with complementary AI]
    \label{lem:complement-br-output}
    Given a skill level~$s>0$ and effort intensity~$\eta\in(0,1]$, the student's best-response effort~$x^*(s,\eta;q)$ produces output
    \[ y^*(s,\eta;q)=\frac{1}{\pi\eta}\left(x^*(s,\eta;q)\right)^2. \]
    This output is
    (i)~positive,
    (ii)~increasing in the motivation parameter~$\pi$,
    (iii)~increasing in~$s$ when~$\eta<1$ and constant in~$s$ when~$\eta=1$,
    (iv)~decreasing in~$\eta$ when
    \begin{equation}
        \label{eq:complement-br-output-intensity-condition}
        \eta\exp\left(1-\frac{\eta}{2}\right)<\frac{s}{\pi q^2}
    \end{equation}
    and non-decreasing otherwise, and
    (v)~increasing in the AI quality parameter~$q$.
\end{lemma}
\begin{proof}
    Let~$x^*\equiv x^*(s,\eta;q)$ and~$y^*\equiv y^*(s,\eta;q)$ for convenience.
    Then
    \[ y^*
        =\frac{1}{\pi}\left(u^*(s,\eta;q)+\frac{(x^*)^2}{2}\right)
        =\frac{(x^*)^2}{\pi\eta} \]
    by \cref{lem:complement-br-payoff}.
    Parts~(i), (iii), and~(v) follow immediately from \cref{lem:complement-br-effort}.
    For part~(ii), we have
    \begin{align*}
        \parfrac{x^*}{\pi}
        &= x^*\cdot\parfrac{\log(x^*)}{\pi} \\
        &= \frac{x^*}{\pi(2-\eta)}
    \end{align*}
    from the proof of \cref{lem:complement-br-effort} and therefore
    \begin{align*}
        \parfrac{y^*}{\pi}
        &= -\frac{(x^*)^2}{\pi^2\eta}+\frac{2x^*}{\pi\eta}\cdot\parfrac{x^*}{\pi} \\
        &= \frac{(x^*)^2}{\pi^2(2-\eta)},
    \end{align*}
    which is positive because~$x^*>0$ and~$\pi^2(2-\eta)>0$.
    For part~(iv), we have
    \begin{align*}
        \parfrac{x^*}{\eta}
        &= x^*\cdot\parfrac{\log(x^*)}{\eta} \\
        &= \frac{x^*}{(2-\eta)^2}\left(\frac{2-\eta}{\eta}-\log\left(\frac{s}{\pi\eta q^2}\right)\right)
    \end{align*}
    from the proof of \cref{lem:complement-br-effort} and therefore
    \begin{align*}
        \parfrac{y^*}{\eta}
        &= -\frac{(x^*)^2}{\pi\eta^2}+\frac{2x^*}{\pi\eta}\cdot\parfrac{x^*}{\eta} \\
        &= -\frac{(x^*)^2}{\pi\eta^2}+\frac{2(x^*)^2}{\pi\eta(2-\eta)^2}\left(\frac{2-\eta}{\eta}-\log\left(\frac{s}{\pi\eta q^2}\right)\right) \\
        &= \frac{2(x^*)^2}{\pi\eta(2-\eta)^2}\left(-\frac{(2-\eta)^2}{2\eta}+\frac{2-\eta}{\eta}-\log\left(\frac{s}{\pi\eta q^2}\right)\right) \\
        &= \frac{2(x^*)^2}{\pi\eta(2-\eta)^2}\left(1-\frac{\eta}{2}-\log\left(\frac{s}{\pi\eta q^2}\right)\right),
    \end{align*}
    which is negative if and only if~\eqref{eq:complement-br-output-intensity-condition} holds.
\end{proof}
\fi

When skill grows proportionally to \emph{effort}, the mapping from skill levels to first-best effort intensities is continuous (see \cref{lem:fb-intensity}).
In contrast, when skill grows proportionally to \emph{output}, the mapping is discontinuous.
To see why, let~$\ybar=0$ so that the student works on each task~$t\in\Tcal$ regardless of its effort intensity~$\eta_t$.
Suppose his skill grows proportionally to output:
\[ s_{t+1}-s_t=ry^*(s_t,\eta_t). \]
\cref{lem:br-output} implies~$y^*(s_t,\eta_t)$ is non-decreasing in~$s_t$.
So, analogously to \cref{sec:solution}, for each task~$t$ the teacher chooses~$\eta_t\in(0,1]$ to maximize~$y^*(s_t,\eta_t)$.
But the left-hand side of~\eqref{eq:br-output-intensity-condition} is increasing in~$\eta$, which implies~$y^*(s_t,\eta)$ is either U-shaped or decreasing in~$\eta\in(0,1]$.%
\footnote{
The left-hand side of~\eqref{eq:br-output-intensity-condition} rises continuously from zero to~$\sqrt{e}$ as~$\eta$ rises from zero to one.
So there are two cases:
\begin{enumerate}

    \item
    Suppose~$s_t<\pi\sqrt{e}$.
    Then, by \cref{lem:br-output} and the intermediate value theorem, there is a unique effort intensity~$\eta_{\text{min-}y}(s_t)\in(0,1)$ that satisfies
    \[ \eta_{\text{min-}y}(s_t)\exp\left(1-\frac{\eta_{\text{min-}y}(s_t)}{2}\right)=\frac{s_t}{\pi} \]
    and minimizes~$y^*(s_t,\eta)$ over~$\eta\in(0,1]$.
    The output~$y^*(s_t,\eta_t)$ falls from
    \[ \lim_{\eta\to0}y^*(s_t,\eta)=s_t \]
    to its minimum~$y^*(s_t,\eta_{\text{min-}y}(s_t))$ as~$\eta$ rises from zero to~$\eta_{\text{min-}y}(s_t)$, then rises to~$y^*(s_t,1)=\pi$ as~$\eta$ rises to one.
    Thus, the output~$y^*(s_t,\eta)$ is U-shaped in~$\eta$.

    \item
    Suppose~$s_t\ge\pi\sqrt{e}$.
    Then~\eqref{eq:br-output-intensity-condition} holds for all~$\eta\in(0,1)$ and \cref{lem:br-output} implies~$y^*(s_t,\eta)$ is decreasing in~$\eta$.

\end{enumerate}
}
So if the teacher wants to maximize~$y^*(s_t,\eta_t)$, then she prefers effort intensities at the boundaries of~$(0,1]$.
Choosing~$\eta_t=1$ induces output~$y^*(s_t,1)=\pi$, while letting~$\eta_t\to0$ induces output
\[ \lim_{\eta\to0}y^*(s_t,\eta)=s_t. \]
Her preference shifts from choosing~$\eta_t=1$ to letting~$\eta_t\to0$ as~$s_t$ passes through~$\pi$ from below.
Thus, if skill grows with output, then the teacher wants to design qualitatively different courses than those she designs when skill grows with effort: her first-best courses comprise tasks that are fully effort- or skill-intensive, rather than tasks that become more skill-intensive gradually.

However, letting~$\eta_t\to0$ is not feasible because I assume~$\eta_t\in(0,1]$ for each task~$t\in\Tcal$.
This assumption is without loss of generality when skill grows with effort: if~$\eta_t=0$, then the student produces~$s_t$ independently of how much effort he exerts, so he exerts none and builds no skill---an outcome the teacher never wants.
This logic breaks down when skill grows with output: if~$\eta_t=0$, then the student produces~$s_t$ and his skill grows by a factor of~$(1+r)$, so the teacher could raise the student's skill exponentially while requiring no effort from him.
Thus, assuming skill grows with output does more than change the courses the teacher designs; it also severs the link between effort and learning that motivates my model.

\section{Conclusion}
\label{sec:conclusion}

This paper introduces a model of learning-by-doing and course design, and analyzes the impact of AI on optimal course design and a student's skill development.
The model delivers four main results.
First, if the student is never tempted to delegate to AI, then the teacher optimally designs courses that become gradually less effort-intensive and more skill-intensive (\cref{thm:fb-solution}).
Second, if the student can be tempted to delegate, then the teacher must distort the course to be more skill-intensive (\cref{thm:sb-solution}).
Third, the student's overall skill gain rises when he starts with more skill (i.e., ``skill begets skill'') and falls when delegation is more tempting (\cref{thm:sb-overall-gain}).
Fourth, if AI complements effort, then improvements in AI quality make high-skill students build skill faster but low-skill students build skill slower (\cref{thm:complement-sb-effort}).

To make the model tractable, I make several assumptions that could be relaxed in future research.
One is that the student is myopic.
My model maps to a setting where a student cares only about current performance, and not future performance or skill.
Yet many real-world students pursue education precisely because they want to gain skill.
A skill-seeking student would internalize the future benefit of building skill, making them more willing to exert effort.
The teacher would capitalize on this willingness by making tasks more effort-intensive.

Another assumption is that the teacher designs her course for one student with a known initial skill level.
Yet real-world teachers design courses for many students with unknown initial skill levels.
In my model, students with more initial skill gain more skill overall (see \cref{thm:sb-overall-gain}).
So if the teacher designs tasks for students with a distribution of initial skill levels, then she may prefer to forgo low-skill students' development for high-skill students'.
In particular, she may ``leave low-skill students behind'' by allowing them to delegate.
Doing so would loosen the incentive constraints on high-skill students and allow the teacher to optimize the course for such students.
Thus, my model speaks to a debate among real-world teachers: should they teach to the best, median, or worst student?

A third assumption is that the student's motivation~$\pi$ is constant across tasks.
This assumption is not realistic: in real-world courses, some tasks carry more grade weight than others, some tasks are more interesting than others, and students tire as the course progresses.
Letting motivation vary exogenously across tasks would change when the incentive constraint~\eqref{eq:ic-constraint} binds.
It currently binds on early tasks and becomes slack as the student accumulates skill.
If the student's motivation fell across tasks, then he might be tempted to delegate late rather than early.

More interesting is to let motivation vary \emph{endogenously} across tasks: to let the teacher choose the sequence of effort intensities~$(\eta_t)_{t\in\Tcal}$ \emph{and} the sequence of motivation parameters~$(\pi_t)_{t\in\Tcal}$.
Then she can induce effort by designing tasks \emph{or} by rewarding their completion.
Since the student's second-best effort is increasing in~$\pi$ (by \cref{lem:complement-sb-effort}), the teacher would raise~$\pi_t$ without bound unless she faced constraints.
A natural constraint is a fixed budget~$\sum_{t\in\Tcal}\pi_t$ (e.g., representing total grade weight) to allocate across tasks.
Two forces push in opposite directions.
Raising~$\pi_t$ on earlier tasks lets the teacher satisfy the incentive constraint with less course distortion and less reduction in skill development.
However, raising~$\pi_t$ on later tasks lets her capitalize on the student's higher willingness to exert skill-building effort.
The relative sizes of these two forces determine whether the teacher should front- or back-load rewards.
\cite{Shaikh-2025-} shows that front-loading is optimal when course content is cumulative.
Content is cumulative in my model too, in the sense that skill compounds across tasks.
But \citeauthor{Shaikh-2025-}'s teacher does not design tasks and \citeauthor{Shaikh-2025-}'s student cannot delegate to AI.
Whether front-loading is optimal when design and delegation are possible remains an open question.

\clearpage
{%
\small
\raggedright
\bibliographystyle{apalike}
\bibliography{references}
}

\appendix
\crefalias{section}{appendix}
\crefalias{subsection}{appsec}

\makeatletter
\renewcommand\theequation{\thesection.\@arabic\c@equation}
\renewcommand\thefigure{\thesection.\@arabic\c@figure}
\renewcommand\thelemma{\thesection.\@arabic\c@lemma}
\renewcommand\theproposition{\thesection.\@arabic\c@proposition}
\renewcommand\thetheorem{\thesection.\@arabic\c@theorem}
\makeatother

\clearpage
\section{Proofs}
\label{app:proofs}

\setcounter{equation}{0}
\setcounter{figure}{0}
\setcounter{lemma}{0}
\setcounter{proposition}{0}
\setcounter{theorem}{0}

\renewcommand{\theHequation}{app.\thesection.\arabic{equation}}
\renewcommand{\theHfigure}{app.\thesection.\arabic{figure}}
\renewcommand{\theHlemma}{app.\thesection.\arabic{lemma}}
\renewcommand{\theHproposition}{app.\thesection.\arabic{proposition}}
\renewcommand{\theHtheorem}{app.\thesection.\arabic{theorem}}

For each result presented in \crefrange{sec:model}{sec:solution}, I present general versions that allow for the AI quality parameter~$q$ (introduced in \cref{sec:complement}) to exceed one.

\ifbodyproofs\else\fi
\ifbodyproofs\else\fi

\ifbodyproofs\else\fi
\ifbodyproofs\else\fi
\ifbodyproofs\else\subsection{Proof of \texorpdfstring{\cref{prop:fb-br-s}}{Proposition~\ref{prop:complement-fb-br-s}}}

\cref{prop:fb-br-s} is the special case of \cref{prop:complement-fb-br-s} with~$q=1$.

\begin{proposition}[First-best outcomes with complementary AI]
    \label{prop:complement-fb-br-s}
    The first-best effort~$x\FB(s;q)$ and payoff~$u\FB(s;q)$ defined in \cref{sec:complement-fb} are
    continuous in~$s$,
    constant in~$s$ when~$s\le\pi eq^2$, and
    increasing in~$s$ when~$s>\pi eq^2$.
\end{proposition}
\begin{proof}
    Suppose~$s\le\pi eq^2$.
    Then~$\eta\FB(s;q)=1$ by \cref{lem:complement-fb-intensity}, and so~$x\FB(s;q)=x^*(s,1;q)=\pi q$ and~$u\FB(s;q)=u^*(s,1;q)=\pi^2q^2/2$ are constant in~$s$.

    Now suppose~$s>\pi eq^2$.
    Then~$\eta\FB(s;q)<1$ is the interior solution to
    \[ \max_\eta x^*(s,\eta;q)\ \ \text{subject to}\ \ 0<\eta\le1. \]
    So the envelope theorem implies
    \begin{align*}
        \parfrac{x\FB(s;q)}{s}
        &= \parfrac{x^*(s,\eta;q)}{s}\bigg\vert_{\eta=\eta\FB(s;q)} \\
        &= \frac{1-\eta}{2-\eta}\left(\frac{\pi\eta q^\eta}{s}\right)^{1/(2-\eta)}\bigg\vert_{\eta=\eta\FB(s;q)} \\
        &= \frac{1-\eta\FB(s;q)}{2-\eta\FB(s;q)}\cdot\frac{x\FB(s;q)}{s} \\
        &> 0.
    \end{align*}
    Thus~$x\FB(s;q)$ is increasing in~$s$.
    Then, by \cref{lem:complement-br-payoff,lem:complement-fb-intensity}, the best-response payoff
    \[ u\FB(s;q)=\frac{1}{\eta\FB(s;q)}\cdot\left(1-\frac{\eta\FB(s;q)}{2}\right)\cdot(x\FB(s;q))^2 \]
    equals the product of three factors that are increasing in~$s$.
    Thus~$u\FB(s;q)$ is also increasing in~$s$.

    Finally, since~$x\FB(s;q)$ and~$u\FB(s;q)$ are continuous functions of~$s$ and~$\eta\FB(s;q)$, and~$\eta\FB(s;q)$ is continuous in~$s$, both~$x\FB(s;q)$ and~$u\FB(s;q)$ are continuous in~$s$.
\end{proof}
\fi

\ifbodyproofs\else\subsection{Proof of \texorpdfstring{\cref{lem:sb-restriction}}{Lemma~\ref{lem:sb-restriction}}}
\label{appsec:sb-restriction}

\cref{lem:sb-restriction} is the special case of \cref{lem:complement-sb-restriction} with~$q=1$.

\begin{lemma}[Parameter restrictions with complementary AI]
    \label{lem:complement-sb-restriction}
    ~
    \begin{enumerate}

        \item[(i)]
        If~\eqref{eq:complement-sb-restriction-1} holds, then~$\pi q/2<\ybar$.

        \item[(ii)]
        If~\eqref{eq:complement-sb-restriction-1} and~\eqref{eq:complement-sb-restriction-2} hold, then~$q\ybar<s_0$.

    \end{enumerate}
\end{lemma}
\begin{proof}
    If~$s_0\le\pi eq^2$, then~$u\FB(s_0;q)=u^*(s_0,1;q)=\pi^2q^2/2$.
    If~$s_0>\pi eq^2$, then~$\eta\FB(s_0;q)<1<e<s_0/\pi q^2$ and so~$u\FB(s_0;q)=u^*(s_0,\eta\FB(s_0;q);q)>u^*(s_0,1;q)=\pi^2q^2/2$ by \cref{lem:complement-br-payoff}.
    Thus~$u\FB(s_0;q)\ge\pi^2q^2/2$ independently of~$s_0$.
    So if~\eqref{eq:complement-sb-restriction-1} holds, then~$\pi^2q^2/2<q\ubar=\pi q\ybar$ and thus~$\pi q/2<\ybar$.

    Suppose~\eqref{eq:complement-sb-restriction-1} and~\eqref{eq:complement-sb-restriction-2} hold.
    Then, by \cref{lem:complement-br-payoff}, the best-response payoff~$u^*(s_0,\eta;q)$ is decreasing in~$\eta$ when~$\eta<s_0/\pi q^2$ and increasing in~$\eta$ when~$\eta>s_0/\pi q^2$.
    It follows that
    \begin{align}
        \sup_{\eta\in(0,1]}u^*(s_0,\eta;q)
        &= \max\left\{\lim_{\eta\to0}u^*(s_0,\eta;q),\,u^*(s_0,1;q)\right\} \notag \\
        &= \max\left\{\pi s_0,\,\frac{\pi^2q^2}{2}\right\} \label{eq:sb-restriction-supremum}.
    \end{align}
    But~$\pi^2q^2/2<q\ubar$ from above.
    So if~$\pi s_0\le\pi^2q^2/2$, then~\eqref{eq:sb-restriction-supremum} contradicts~\eqref{eq:complement-sb-restriction-2}.
    Thus~$\pi s_0>\pi^2q^2/2$ and therefore
    \[ q\ybar=\frac{q\ubar}{\pi}<\frac{1}{\pi}\sup_{\eta\in(0,1]}u^*(s_0,\eta;q)=s_0.\qedhere \]
\end{proof}
\fi
\ifbodyproofs\else\fi
\ifbodyproofs\else\subsection{Proof of \texorpdfstring{\cref{lem:sb-intensity}}{Lemma~\ref{lem:sb-intensity}}}

\cref{lem:sb-intensity} is the special case of \cref{lem:complement-sb-intensity} with~$q=1$.
I prove \cref{lem:complement-sb-intensity} in \cref{appsec:complement-sb-intensity}.
\fi
\ifbodyproofs\else\fi
\ifbodyproofs\else\fi
\ifbodyproofs\else\fi
\ifbodyproofs\else\subsection{Proof of \texorpdfstring{\cref{thm:sb-overall-gain}}{Theorem~\ref{thm:sb-overall-gain}}}

\cref{thm:sb-overall-gain} is the special case of \cref{thm:complement-sb-overall-gain} with~$q=1$.

\begin{theorem}[Overall skill gain with complementary AI]
    \label{thm:complement-sb-overall-gain}
    Suppose~\eqref{eq:complement-sb-restriction-1} and~\eqref{eq:complement-sb-restriction-2} hold, and~$(\eta_t)_{t\in\Tcal}$ solves~\eqref{eq:complement-problem}.
    Then~$(s_T-s_0)$ is 
    (i)~increasing in the student's initial skill level~$s_0$,
    (ii)~increasing in the motivation parameter~$\pi$, 
    (iii)~increasing in the skill accumulation rate~$r$,
    and
    (iv)~decreasing in the AI-generated output~$\ybar$.
\end{theorem}
\begin{proof}
    The result follows immediately from \cref{lem:complement-sb-increments} and the fact that
    \[ s_T-s_0=\sum_{t\in\Tcal}(s_{t+1}-s_t).\qedhere \]
\end{proof}
\fi

\ifbodyproofs\else\fi
\ifbodyproofs\else\fi
\ifbodyproofs\else\fi
\ifbodyproofs\else\fi

\ifbodyproofs\else\fi
\ifbodyproofs\else\fi

\end{document}